\documentclass[11pt]{amsart}
\allowdisplaybreaks[1]

\usepackage{a4wide}
\usepackage{amssymb}
\usepackage{amsmath}
\usepackage{amscd}
\usepackage{graphicx}
\usepackage{color}

\newtheorem{theorem}{Theorem}[section]
\newtheorem{lemma}[theorem]{Lemma}
\newtheorem{prop}[theorem]{Proposition}
\newtheorem{coro}[theorem]{Corollary}

\theoremstyle{definition}
\newtheorem{definition}[theorem]{Definition}
\newtheorem{remark}[theorem]{Remark}
\newtheorem{example}[theorem]{Example}

\newtheorem{assumption}[theorem]{Assumption}

\newcommand{\RR}{\mathbb{R}}
\newcommand{\bbS}{\mathbb{S}}

\newcommand{\CC}{\mathbb{C}}
\newcommand{\NN}{\mathbb{N}}
\newcommand{\TT}{\mathbb{T}}
\newcommand{\ZZ}{\mathbb{Z}}
\newcommand{\QQ}{\mathbb{Q}}
\newcommand{\GG}{\mathbb{G}}
\newcommand{\Gdual}{\widehat{\mathbb{G}}}

\newcommand{\dd}{\,{\rm d}}

\newcommand{\Bragg}{\mathcal{S}}

\newcommand{\ev}{\mathcal{E}}

\newcommand{\loops}{\mathcal{Z}}
\newcommand{\alg}{\mathcal{P}}

\newcommand{\cF}{\mathcal{F}}
\newcommand{\cN}{\mathcal{N}}

\newcommand{\cA}{\mathcal{A}}

\newcommand{\cHH}{\mathcal{H}}
\newcommand{\cM}{\mathcal{M}}

\newcommand{\cZ}{\mathcal{Z}}

\newcommand{\cC}{\mathcal{C}}
\newcommand{\cS}{\mathcal{S}}

\newcommand{\difm}{\omega}
\newcommand{\haarG}{l_\GG}

\newcommand{\bk}{{\bf{k}}}
\newcommand{\bl}{{\bf{l}}}

\newcommand{\ts}{\hspace{0.5pt}}

\newcommand{\Ghat}{\widehat{\GG}}

\newcommand{\charF}{{\bf 1}}

\DeclareMathOperator{\supp}{supp}
\DeclareMathOperator{\vol}{vol}
\DeclareMathOperator{\len}{len}

\newcommand{\Hm}[1]{\leavevmode{\marginpar{\tiny%
$\hbox to 0mm{\hspace*{-0.5mm}$\leftarrow$\hss}%
\vcenter{\vrule depth 0.1mm height 0.1mm width \the\marginparwidth}%
\hbox to
0mm{\hss$\rightarrow$\hspace*{-0.5mm}}$\\\relax\raggedright #1}}}

\begin{document}

\title[Pure point diffraction]{stationary processes with pure point diffraction}

\author{Daniel Lenz}
\address{Mathematisches Institut, Fakult\"at f\"ur Mathematik und Informatik, Friedrich-Schiller-Universit\"at Jena,
D- 07743 Jena, Germany}
\email{daniel.lenz@uni-jena.de}
\urladdr{http://www.tu-chemnitz.de/mathematik/analysis/dlenz}

\author{Robert V.\ Moody}
\address{Department of Mathematics and Statistics,
University of Victoria, \newline
\hspace*{12pt}Victoria, BC, V8W3P4, Canada}
\email{rmoody@uvic.ca}
\date{\today}
\thanks{RVM thanks the Natural Sciences and Engineering Council of Canada
for its support of this work.}

\begin{abstract} We consider the construction and classification of some new mathematical objects, called ergodic spatial stationary processes, on locally compact Abelian groups, which provide a natural and very general setting for studying diffraction and the famous inverse problems associated with it. In particular we can construct complete families of solutions to the inverse problem from any given pure point measure that is chosen to be the diffraction.
In this case these processes can be classified by the dual of the group of relators based on the set of Bragg peaks, and this gives a solution to the homometry problem for pure point diffraction.

An ergodic spatial stationary process consists of a measure theoretical dynamical system
and a mapping linking it with the ambient space in which diffracting density is supposed to
exist. After introducing these processes we study their general properties and link pure point diffraction to almost periodicity.

Given a pure point measure we show how to construct from it and a given set of phases a corresponding ergodic spatial stationary process. In fact we do this in two separate ways, each of which sheds its own light on the nature of the problem. The first construction can be seen as an elaboration of the Halmos--von Neumann theorem, lifted from the domain of dynamical systems to that of stationary processes. The second is a Gelfand construction obtained by defining a suitable Banach algebra out of the putative eigenfunctions of the desired dynamics.
\end{abstract}

\maketitle

\section*{Outline}
This paper is concerned with mathematics of diffraction. More specifically  we are interested
in the famous inverse problem for diffraction: given something that is putatively the diffraction
of something, what are all the somethings that could have produced this diffraction.

Diffraction has been a mainstay in crystallography for almost a hundred years.\footnote{The Knipping-Laue experiment establishing x-ray diffraction was
carried out in 1912.} With the proliferation of extraordinary new materials with varying degrees of order and disorder, the importance of diffraction in revealing internal structure continues to be central. The precision, complexity, and variety of modern diffraction images is striking, see for instance the recent review article \cite{Withers}. Nonetheless, in spite of many advances, the fundamental question of diffraction, the inverse problem of deducing physical structure from diffraction images, remains as challenging as ever.

The discovery of aperiodic tilings and quasicrystals revived interest in the mathematics of diffraction, particularly since the types of diffraction images generated by such structures -- essentially or exactly pure point diffraction together with symmetries not occurring in ordinary crystals --
had not been foreseen by mathematicians, crystallographers, or materials research scientists.\footnote{Indeed, the discovery of quasicrystals by Shechtman  via diffraction experiments   met with substantial disbelief for a period of about two years.
His determination to communicate it to a skeptical scientific community  is  particularly stressed  by the committee awarding him the  Nobel prize in Chemistry in 2011 for this discovery. }
Starting with \cite{Hof1},
 the resulting
new techniques created in order to understand mathematical diffraction have been considerable.
Among them we mention the introduction of ideas from the harmonic analysis of translation bounded measures \cite{BM,BL}, ergodic dynamical systems \cite{Bell,Radin0,Radin,RW}, and stochastic point processes \cite{Gouere,XR, BBM}.

The diffraction of a structure is defined as the Fourier transform of its autocorrelation. The latter is a positive-definite measure in the ambient space of the structure -- typically $\RR^3$, but mathematically any locally compact Abelian
group $\GG$ -- and the diffraction is then a centrally symmetric positive measure in the corresponding Fourier reciprocal space  -- either $\widehat{\RR^3} \simeq \RR^3$ or generally $\widehat{\GG}$ as the case may be.

In this paper we examine three basic questions that arise from this theory. The first question
asks whether {\em every} positive centrally symmetric  measure is actually the diffraction of something. The second asks about the classification of the all various
`somethings' that have this given diffraction. The third asks what kinds of structures
the `somethings' can be. In the case of pure point diffraction (when the diffraction measure
is a pure point measure) we can give a satisfactory answer to these problems.

We will be in the setting of locally compact
Abelian groups. The difficulties of the questions, at least as we examine them here,  have little to do with the generality of the setting. They are just as hard for $\RR^3$. Indeed the second question is a very old
one in crystallography, called the {\em homometry problem} and is part of {\em the} fundamental problem of diffraction theory, namely how to unravel the nature of a structure that has created a given diffraction pattern.

Our first question immediately raises the third question. Exactly what do we mean by a structure that can diffract? Typically diffractive structures are conceived of as discrete measures, for instance point measures describing the positions of atoms, vertices in tilings, etc., weighted by appropriate scattering strengths, or as continuous measures describing the (scattering) distribution of the material structure in space. More recently, starting with the work of Gou\'{e}r\'{e}  \cite{Gouere0, Gouere} there has been a
shift to discussing the diffraction of certain point processes, which in effect are random point measures, and various distortions of these (see e.g. \cite{BBM,XR,LM}). In this paper we are led to introduce new  structures which we call a {\em spatial stationary processes}. We shall show that these always lead
to autocorrelation and diffraction measures and contain the typical theories of diffraction
discussed above\footnote{There is a technical requirement about the existence
of second moments.}. In terms of these we get an attractive answer our three questions, in the case of pure point diffraction. However, the answer to the third question is interesting, because it does not describe the underlying structure explicitly, but rather in terms of the ways that it appears to compactly supported continuous functions on the ambient space. In this respect it resembles the description of measures given by the Riesz representation theorem or the theory of distributions with their spaces of test functions.
In our case, the convergence conditions of measures do not seem to be generally valid. In this respect there is much to be learned about the full scope of which
structures can diffract.

Let $\GG$ be a locally compact Abelian group and $\Gdual$ its dual group.
We treat these groups in additive notation. We begin with an abstract notion
of a stationary process for $\GG$ and then shift our attention to the more concrete notion of
a spatial stationary process for $\GG$ (\S\ref{Stochasticprocesses} and \S\ref{spatial}).  As for terminology let us note that outside of the discussion of usual stochastic processes, the word `stochastic' is not used in our discussion of processes and stationary processes.

A spatial stationary process basically  consists of a $\GG$-invariant mapping
\[ N: C_c(\GG) \longrightarrow L^2(X,\mu) \]
where $(X,\mu)$ is a probability space on which there is a measure preserving
action of $\GG$.

The idea here derives directly from the theory
of stochastic point processes, and it is perhaps useful to discuss this briefly.
We follow \cite{Karr} here. A stochastic point
process, say in $\RR^d$, is a random variable $M$ on a probability space
$(\Omega, \cF, P)$ with values in the space $\cM_p$ of point measures
$\xi$ on $\RR^d$ for which $\xi(B) \in \ZZ_{\ge 0}$ for all bounded subsets $B$ of $\RR^d$.
These are point measures and can be written in the form $\xi = \sum_{j=1}^K \delta_{x_j}$
where $K$ can be a positive integer or infinity and the points $x_j$ are (not necessarily distinct)
points of $\RR^d$. Usually the space $\cM_p$ is far larger than
needed. For one thing it allows points in $\RR^d$ to appear with multiplicities and for another
it does not require any minimal separation between distinct points (a hard core condition) that
would usually be imposed in crystallography or in tiling theories. We assume then that we are
really only interested in some subset
$X \subset \cM_p$ of permissible outcomes for the point process, and that almost surely
the point process produces measures in $X$. Typically we wish $X$ to be invariant under translation
since the position of our point sets is not particularly relevant, and we assume that the law $\mu$
of the process (the probability measure $\mu$ induced on $X$ by the random variable)
is invariant under translation (i.e. the process is stationary).

Now the point is this. If $f\in C_c(\RR^d)$ then we obtain a random variable on $(\Omega, \cF, P)$
by $\omega \mapsto N_f(\xi):= \sum_{j=1}^K f(x_j)$, where $\xi = M(\omega) \in X$ is the measure indicated above.
In effect $N$ can be viewed as a mapping from $C_c(\RR^d)$ into mappings on $X$. Under
mild assumptions $N:C_c(\RR^d) \longrightarrow L^2(\RR^d,\mu)$. As is often the case with
random variables, we ignore the underlying probability space $(\Omega, \cF, P)$
since we have everything we need from the probability space $(X,\mu)$ and the mapping $N$.
Then $N$ itself can be referred to as the stochastic point process.

In this paper $\RR^d$ is replaced by any locally compact Abelian group $\GG$.
Also $X$ is not assumed to have any particular
form (e.g. a set of point sets, or a set of measures). We do not wish to assume
any particular mathematical structure on the outcomes of the processes we are studying. Our point
of view is that we will start with a measure in $\Gdual$ which we would like to show
is the diffraction of something. But we do not know in advance what sort of mathematical
objects could produce this diffraction (point sets, continuous distributions of density, Schwartz
distributions, etc.).  We derive information
about $X$ only from the behaviour of the functions $N(F)$, $F\in C_c(\GG)$, upon it.
This is, physically, rather  appropriate. In a physical situation one determines
the structure of a physical object by measurements upon it and the results of measurements
are ultimately all that one can know about it.

We are particularly interested in pure point diffraction. However, the concept of stationary processes
seems interesting in its own right and we develop the theory of these a bit, independent of the question of pure point diffraction. It is even useful initially to  ignore the fact that the resulting process is {\it{spatial}},
that is, it is supposed to result from some structure in $\GG$.
So at the beginning we often only require that we have a $\GG$-invariant mapping $N:C_c(\GG) \longrightarrow
\cHH$ where $\cHH$ is a Hilbert space with conjugation -- we call these
stationary processes. Assuming existence of second moments, we can then, for  each such process, find an `autocorrelation measure' on $\GG$ and a
`diffraction measure' on $\Gdual$. The latter is always a centrally symmetric positive
measure,  and through it we can define continuous and pure point diffraction. We show how pure point and continuous diffraction are tied
to notions of almost periodicity of $N$. Also we show how to decompose a sum of spatial processes into sub-processes which are respectively pure point and continuous.

Coming to the pure point part of the paper, \S\ref{Purepoint}, we assume given a centrally symmetric positive  pure point  measure $\difm$ on $\Gdual$, which can be expressed as a  Fourier transform.
We wish to create a probability space $(X,\mu)$ on which there is an ergodic action of $\GG$ and a continuous stationary process $N: C_c(\GG) \longrightarrow L^2(X,\mu)$ whose diffraction is $\difm$.

Not surprisingly phase factors play a crucial role in this since
it is phase information that is lost in the process of diffraction. There is a natural combinatorial type of group $Z$
generated by the positions of the Bragg peaks. This group can be seen as an abeleanized homotopy group of an associated Cayley graph.
It is the dual group of $Z$, or the dual group of a natural factor $\cZ$ of this group,
that classifies all the spatial processes on $\GG$ with diffraction $\difm$, the
classification being up to isomorphism or up to translational isomorphism respectively.

We call these phase factors {\em phase forms}. Background on phase forms is given in  Section \ref{Cycles}. In \S\ref{CyclestoStochasticprocesses} we show how to
use $\difm$ and a phase form $a^*$ to construct a spatial process whose diffraction
is $\difm$. Uniqueness is already shown in \S\ref{Uniqueness}.
How all these ingredients give a solution to the homometry problem is discussed in
\S\ref{Homometry}.

 Our solution to the pure point problem is in some sense an elaboration of the famous
Halmos-von Neumann result about realizing ergodic pure point dynamical systems in terms of the character theory of compact Abelian groups. There the objective was to
classify pure pointedness at the level of the spectrum. In our case we wish to classify the
diffraction, which amounts to not only knowing the spectrum but also the
intensity of diffraction at each point of the spectrum. Our approach to the construction
of a spatial process is to realize $X$ as the dual to the discrete group generated
by the set of Bragg peaks, but at the same time to include all the phase information which derives from the Bragg intensities. In this context we mention the recent work of Robinson \cite{Rob3}  on how to realize systems with a given group of eigenvalues via so called cut-and-project schemes.

In  \S\ref{compactG} we take a closer look at the situation that the underlying group is compact. In this case all processes can be realized as measure processes
under some weak assumption. We then specialize even further and in \S\ref{simplePeriodic}
discuss some results of Gr\"unbaum and Moore
on rational diffraction from one dimensional periodic structures. These pertain to the simple case that $\GG = \ZZ/N\ZZ$, but even so, the results are interesting and not easy to obtain.

Finally in \S\ref{altConstruction} we sketch out a second approach to the construction of a spatial
process, this time based on Gel'fand theory.

Summarizing:
\begin{itemize}
\item Each pure point ergodic spatial process gives rise to a pure point diffraction  measure $\difm$ and a phase form $a^*$.

\item Spatial processes with the same pure point diffraction measure and the same phase form are naturally isomorphic up to translation.

 \item The set of all possible phase forms associated with a given pure point measure $\difm$ form an Abelian group $\cZ(\difm)$, and for each choice of
 $a^* \in \cZ(\difm)$, there is a ergodic spatial process with diffraction $\difm$ and
 phase form $a^*$.

\end{itemize}

The paper has two main features, after the introduction of stationary and spatial
processes: one is a development of a general theory of these types of processes
and the other is a detailed study of pure point diffraction and examples of how it
looks in special cases. Readers interested primarily in the pure point theory can skip
sections $3.3, 3.4 ,4,5$ on first reading if they so wish.

\section{Notation} \label{notation}
In this section we gather some notation used throughout.

\bigskip

$\NN:= \{1,2, \dots, \}$. Let $\GG$ be a locally compact abelian group.
We let $C_c(\GG)$ denote the space of compactly supported complex valued continuous
functions on $\GG$. For compact $K \subset \GG$, $C_K(\GG)$
is the subspace of elements of $C_c(\GG)$ whose support is in $K$. For $F\in C_c(\GG)$, $\widetilde F \in C_c(\GG)$ is defined
by $\widetilde F(x) = \overline{F(-x)}$. The translation action of $\GG$ on itself is
denoted by $\tau$: $\tau(t)(x) = t +x$ for all $t,x \in \GG$. This action determines, in the usual way, an action on $C_c(\GG)$.
By $l_\GG$ we denote the (fixed) Haar measure on $\GG$.

The convolution of two functions $F,G\in C_c (\GG)$ is defined to be the function $F\ast G\in C_c (\GG)$ given by
$$F\ast G (t) = \int_\GG F( t- s ) G (s) d l_\GG (s).$$

The dual group of $\GG$ is denoted by $\Ghat$ and the \textit{Fouriert ransform} of an $F\in C_c (\GG)$ is denoted by $\widehat{F}$ i.e.
$$\widehat{F} (\gamma) = \int_\GG \overline{(\gamma,t)} F(t) d l_\GG (t).$$

A measure $\nu$ on $\GG$ is called \textit{transformable} if there exists a measure $\sigma$  on   $\Ghat$ with
$$ \int_{\Ghat} |\widehat{F}|^2 d\sigma = \int_\GG F\ast \widetilde{F} d\nu$$
for all $F\in C_c (\GG)$.
In this case $\sigma$ is uniquely determined  and called the \textit{Fourier transform of $\nu$}. In such a situation we call $\sigma$ \textit{backward transformable}. Unlike all other notation introduced in this section the concept of backward transformability is not standard. However, it is clearly very useful in our context as it is exactly the property shared by the diffraction measures we investigate.

A  sequence $\{A_n\}$ of compact subsets
of $\GG$ is called a \textit{van Hove sequence}
if for every compact $K\subset \GG$
\begin{equation} \label{van-hove}
    \lim_{n\to\infty} \frac{l^{}_\GG
    (\partial^K A_n)}
    {l^{}_\GG (A_n)}  \; = \; 0.
\end{equation}
Here,
for compact $A,K$, the
``$K$-boundary''  $\partial^K A $  of $A$ is defined as
$$ \partial^K A := ((A + K)\setminus A^\circ) \cup
    (\overline{\GG\setminus A} - K) \cap A.$$
The existence of van Hove
sequences for all $\sigma$-compact locally compact Abelian groups is shown in
\cite[p.~249]{Martin}, see also Section 3.3 and Theorem (3.L)
of \cite[Appendix]{Tempelman}.

\section{Stationary processes introduced}\label{Stochasticprocesses}
In this section we introduce our concept of stationary process on the locally compact Abelian group $\GG$. To do so we first discuss some background on unitary representations.
 The concept of stationary process is somewhat more general than needed in the remainder of the paper. The reason for this is twofold: On  the one hand this general treatment does not lead to more complicated proofs but rather makes proofs more transparent. On the other hand for a the study of mixed spectra this general concept may prove to be especially  useful.

\bigskip

Let $\GG$ be a locally compact abelian group and let $\cHH$ be a Hilbert space with
inner product $\langle\cdot | \cdot \rangle$. We assume linearity in the first variable
and conjugate linearity in the second.
A strongly continuous unitary representation of the group $\GG$ on $\cHH$ is a map $T$ from $\GG$ into the bounded  operators on $\cHH$ such that $r \mapsto \langle T_r f, T_r f\rangle$ is constant on $\GG$ for any $f\in \cHH$,    $T_0$ is the identity on $\cHH$,   $T_{t+s} = T_t T_s$  holds for all $s,t\in \GG$  and  $t\mapsto T_t f$ is continuous for any $f\in \cHH$. Then, obviously each  $T_s$ is unitary (i.e. bijective and norm preserving).
 An $f\in \cHH$ is then  called an {\it eigenfunction\/} of $T$ with
{\it eigenvalue\/} $\xi\in \Ghat$ if $T_t f = \overline{(\xi,t)} f$ for every
$t\in \GG$.  The closed subspace of $\cHH$ spanned by all eigenfunctions is denoted by $\cHH_p$.  The representation $T$ has {\it pure point spectrum} if  $\cHH= \cHH_p$, or equivalently if $\cHH$  has  an orthonormal basis consisting of eigenfunctions \footnote{Often the eigenfunction condition is written
as $T_t f = (\xi,t) f$ for eigenvalue $\xi\in \Ghat$. For our purposes things work
out more smoothly by formulating it with the conjugate value.}.

By Stone's theorem, compare \cite[Sec.~36D]{Loomis}, there exists a projection-valued measure
\begin{equation*}
  E_T\! : \; \mbox{Borel sets of $\Ghat$} \; \longrightarrow \;
  \mbox{projections on $\cHH$}
\end{equation*}
with
\begin{equation} \label{spectralmeasure}
  \langle f \mid T_t f \rangle \; = \;
  \int_{\Ghat} (\xi, t) \dd \rho^{}_f (\xi)\,
\end{equation}
for all $t\in\GG$ and $f\in \cHH$,
where $\rho^{}_f$ is the measure on $\Ghat$ defined by
$\rho^{}_f (B) := \langle f \mid E_T (B)f\rangle$. Then $\rho^{}_f$ is called the \textit{spectral measure} of $f$. It is the unique measure on $\Gdual$ satisfying \eqref{spectralmeasure}.

A mapping $\overline{(\cdot)}$ of a Hilbert space $(\cHH, \langle \cdot \mid \cdot \rangle)$ into itself is said to be a {\it conjugation} of $\cHH$ if
\begin{itemize}
\item[{(i)}] $\overline{(\cdot)}$ is conjugate-linear and $\overline{\overline{\xi}} = \xi$;
\item[{(ii)}] for all $\xi, \psi \in \cHH$, $\langle \overline{\xi} \mid \overline{\psi}\rangle
= \langle \psi \mid \xi\rangle$.
\end{itemize}

\begin{definition} Let $\GG$ be a locally compact abelian group. A {\it stationary process}
or $\cHH$-{\it stationary process} on $\GG$ is a   triple $\cN=(N, \cHH, T)$ consisting of a Hilbert space $\cHH$ with
conjugation, a strongly continuous unitary representation $T$ of $\GG$ on $\cHH$ and a linear
 $\GG$-map $N: C_c (\GG)\longrightarrow \cHH$ such that for all
$F\in C_c (\GG)$ $N(\overline F) = \overline{N(F)}$.  A stationary  process is called {\it ergodic} if the eigenspace of $T$ for the eigenvalue $0$ is one-dimensional.
\footnote{The origins of the definition lie in the theory of stochastic point processes,
hence the name. We will often simply write {\em process}
instead of {\em stationary process}. We do not use the word {\em stochastic}, although
our definition allows for processes that are stochastic in nature.}
\end{definition}

\medskip

The set of all $\cHH$-stationary processes on $\GG$ forms a real vector space in the obvious
way by taking linear combinations of mappings.  There is also a canonical notion of isomorphism between stochastic processes: The processes $(N_1, \cHH_1, T_1)$ and $(N_2,\cHH_2, T_2)$ are called \textit{isomorphic} if there exists an invertible unitary map $U : \cHH_1\longrightarrow \cHH_2$ which intertwines $T_ 1$ and $T_2$ and satisfies $N_2 = U N_1$.

 \medskip

We will be interested in the spectral theory of stationary processes. An important notion is given next.

\begin{definition} A stationary process $\cN= (N,\cHH,T)$ is said to have {\it pure point spectrum} if the
representation $T$ of $\GG$ on $\cHH$ has pure point spectrum.
\end{definition}

\bigskip

\begin{definition} A stationary process $\cN=(N,\cHH,T)$ is called a {\it spatial
stationary process} if there is a probability space $(X,\mu)$ and an action
$T: \GG \times X  \longrightarrow X, (t,x)\mapsto t\cdot x$ such that $\mu$ is invariant
under the $\GG$-action on $X$ through $T$, $\cHH = L^2(X,\mu)$ with the usual inner product $\langle f \mid g\rangle = \int_X f \overline{g} d\mu$,
and the action of $T$ is extended to an action on $\cHH$ by $T_t f(\cdot) = f((-t)(\cdot))$ for all $t\in\GG$, $f\in L^2(X,\mu)$. The conjugation
here is the natural one: complex conjugation of functions on $X$.  We write
$\cN=(N,X,\mu,T)$. The stationary
spatial process is called {\it full} if the algebra generated by  functions of the
form $\psi \circ N(F)$, $\psi \in C_c(\CC)$, $F\in C_c(\GG)$, is dense in
$L^2(X,\mu)$.
\end{definition}

\begin{remark}
In some sense the assumption of fullness is only a convention. More precisely, given any spatial stationary process we may create a full process with essentially the same properties by a variant of Gelfand construction. Details are  given below in Theorem~ \ref{fullProcess}.
\end{remark}

For the sake of economy, we have used the same notation, $T$, for the action of $\GG$ on $X$
and the corresponding action of $\GG$ on $\cHH := L^2(X,\mu)$. We often omit the word `stationary' in the sequel, but we will always understand
it to be in force.

The definition of stationary
processes implicitly defines them as {\em real} processes: $N(F)$ is real for all real-valued
functions $F$ on $\GG$ in the sense that $\overline{N(F)} = N(F)$. In the case of
spatial processes, all real-valued functions on $\GG$ are mapped by $N$ to real-valued functions on $X$. One could relax the conditions to allow complex-valued processes, but we shall not do that
here.

\bigskip

\begin{definition}\label{definition-isomorphism}
Two stationary spatial processes $\cN = (N,X,\mu,T)$ and $\cN'=(N',X',\mu',T')$ are {\it  spatially  isomorphic} if there is an invertible unitary mapping
$M : L^2 (X,\mu)\longrightarrow L^2 (X',\mu')$ with
\begin{itemize}
\item $M\circ N  = N'$;
\item $M \circ T_t = T_t' \circ M$ for all $t\in \GG$;
\item $M(f g) = M(f) M(g)$ for all $f,g\in L^\infty (X,\mu)$ and $M^{-1} (f' g') = M^{-1} (f') M^{-1} (g')$ for all $f',g'\in L^\infty (X',\mu')$.
\end{itemize}
The map $M$ is then called spatial isomorphism between the processes.
\end{definition}
The isomorphism class of $\cN= (N,X,\mu,T)$ is denoted by $[\cN]=[(N,X,\mu,T)]$.

\begin{remark} (a) The first two points of the definition just say that the processes are isomorphic. The special requirement for spatial isomorphy is the third point. As shown in the next proposition, the spatial isomorphism are exactly those isomorphisms which are compatible with the $L^\infty$ module structure of $L^2$ in the sense that $M(f g) = M(f) M(g)$ for all $f\in L^\infty$ and $g\in L^2$.

(b)  Our definition of spatial  isomorphism parallels the
notion of spectral isomorphism for measure preserving group actions
on probability spaces. Our conditions on $M$ imply that spatial  isomorphism implies conjugacy in the sense
that there is a measure isomorphism between the Borel sets of $X$ and
those of $X'$ \cite{Walters}, Theorem~ 2.4. If in addition $X,X'$ are complete separable metric spaces then this can be improved to being an isomorphism between $X$ and $X'$,
that is a bijective intertwining mapping between two subsets of full measure
in $X$ and $X'$ respectively \cite{Walters}, Theorem~2.6.
\end{remark}

\begin{prop} Let $\cN = (N,X,\mu,T)$ and $\cN'=(N',X',\mu',T')$ be spatially isomorphic processes with spatial isomorphy $M$. Then, $M$ maps $L^{\infty} (X,\mu)$ into $L^\infty (X',\mu')$ and $M(f g) = M(f) M(g)$ holds for all $f\in L^{\infty} (X,\mu)$ and $g\in L^2 (X,\mu)$.  The corresponding statements hold for $M^{-1}$ as well.
\end{prop}
\begin{proof} It suffices to consider the statements on $M$.  By assumption we have  that for $f,g\in L^\infty$ the product $M(f) M(g)$ belongs to $L^2$ and in fact equals $M (f g)$ (and similarly for $M^{-1}$). By a simple limit argument, we then see that for any $f\in L^\infty$ the equation   $M (f g) = M(f) M(g)$ must hold for any $g\in L^2$. This shows that the operator of  multiplication by $M(f)$ is defined on the whole of  $L^2$ (as $M$ is onto). Furthermore, the graph of every multiplication operator can easily be seen to be closed. Thus,  the closed graph theorem  gives that the operator of multiplication by $M(f)$  is a bounded operator on $L^2$. This in turn yields that $M(f)$ belongs to $L^\infty$.
\end{proof}

\section{Diffraction theory for stochastic processes}
  In this section we show how each stationary process with a second moment comes with a positive definite  measure $\gamma$  on $\GG$  called the autocorrelation, a positive measure  $\difm$ on $\Gdual$ called the diffraction measure,  and  the diffraction-to-dynamics map $\theta$. We then discuss how the autocorrelation comes about as a limit and characterize existence of a second moment.

\bigskip

\subsection{Autocorrelation  and diffraction}
Let a stationary process $\cN=(N,\cHH,T)$ be given.

 We say that $\cN$ has a \textit{second moment} (or is with second moment) if there is a measure $\mu^{(2)}:=\mu^{(2)}_N$
 on $C^{\RR}_c(\GG \times \GG)$ (necessarily unique if it exists) satisfying
\begin{equation}\label{2momentCondition}
\mu^{(2)}(F \otimes G) = \langle N(F) \mid N(G) \rangle \,
\end{equation}
for all real valued  $F,G\in C_c (\GG)$.
 Much of this paper is based on the assumption that the second moment measure exists.
 Since $N$ is real, so is $\mu^{(2)}$, and hence it is real valued
 \footnote{The $m$th moment of $\cN$ will be defined by
 $\mu^{(m)}(F_1\otimes \cdots \otimes F_m) = \int_X N(F_1) \cdots N(F_m) \, d \mu$
 which is different from what we get from the inner product $\langle \cdot \mid \cdot \rangle$
 when $m=2$ unless we restrict to real-valued functions. This is why we restrict to real-valued
 functions in this definition.}.   Furthermore $\mu^{(2)}$ is evidently
 translation bounded since $N$ is $\GG$-invariant.

\begin{prop} \label{def:gamma}
Let $\cN = (N,\cHH,T)$ be a process.
Then $\cN$ has a second moment if and only if  there is a measure $\gamma = \gamma_N$ on $\GG$
satisfying
\begin{equation}\label{acEquation}
\gamma(F * \tilde G) = \langle N(F) \mid N (G)\rangle
\end{equation}
for all $F,G \in C_c(\GG)$.
In the case that such a measure exists, it is unique and furthermore it is real (i.e. assigns real values to real valued functions), positive definite, and translation bounded.
\end{prop}
\begin{proof} It suffices to restrict attention to real-valued $F$ and $G$ in $C_c (\GG)$ as $N$ is invariant under a conjugation and hence both $F\ast \tilde{G}$ and $\langle N(F) \mid N (G)\rangle $ are linear in $F$ and antilinear in $G$.
Let $\psi := \mu^{(2)}$ above.
The assumption of stationarity
of $N$ shows that $\psi$ is invariant under simultaneous translation of both of its two variables.
 Define the mapping
\begin{eqnarray}
\GG \times \GG &\longrightarrow& \GG \times \GG\\
(x,y) &\mapsto & (x-y,y) =:(u,v) \nonumber
\end{eqnarray}
For $F_1,F_2 \in C_c^{\RR}(\GG)$ we define $H \in C_c(\GG\times \GG)$ by
$H(u,v) = H(x-y,y) = F_1(x)F_2(y)$. This change of variables defines a new measure
$\psi^*$ via
\[ \int_\GG \int_\GG F_1(x)F_2(y) \dd\psi(x,y) = \int_\GG \int_\GG H(u,v) \dd \psi^*(u,v) \,.\]

Since $(\tau_t x,\tau_t y) \leftrightarrow (u,\tau_t v)$, the translation invariance of $\psi$
translates to translation invariance of the second variable of $\psi^*$. Thus for
measurable sets $A,B \subset \GG$ we have
$\psi^*(A \times B) = \psi^*(A \times (t+B))$ for all $t\in \GG$. For $A$ fixed this is leads to a translational
invariant measure on $\GG$ which is hence a multiple $c(A)$ of Haar measure
 $\haarG$:
\[ \psi^*(A \times B) = c(A) \haarG(B) \,. \]

The mapping $A \mapsto c(A)$ is a measure on $\GG$, and this measure
is the reduction $\psi^{\mbox{red}}$ of $\psi$:
\[\psi^*(A\times B) = \psi^{\mbox{red}}(A)\, \haarG(B) \,.\]

Now for $H(u,v) = F_1(u+v)F_2(v)$,
\begin{eqnarray*}
\int\int H(u,v) \dd \psi^*(u,v) &=& \int\int F_1(u+v) F_2(v) \dd \psi^{\mbox{red}}(u) \dd \haarG(v)\\
&=&\int\left( \int F_1(u+v) F_2(v) \dd \haarG(v)\right) \dd \psi^{\mbox{red}}(u)\\
&=& \int (F_1*\tilde{F_2})(u) \dd \psi^{\mbox{red}}(u) \, ,
\end{eqnarray*}
so finally
\[ \psi(F_1 \otimes F_2) = \psi^{\mbox{red}}(F_1 *\tilde{F_2}) \,.\]
and  $\psi^{\mbox{red}}$ is the desired measure, which we now
denote by $\gamma$.

It is positive definite since $\gamma(F * \tilde F) = \mu^{(2)}(F \otimes F)
= ||(N(F)||^2\ge 0$. All positive definite measures are translation bounded \cite{BF}.
Since $\psi := \mu^{(2)}$ is a real measure, one sees that $\gamma$ is also
a real measure. The denseness of the set of functions $F*\tilde G$ in $C_c(\GG)$ shows
that $\gamma$ is unique.

We can reverse all the steps of this proof, and from the existence of the real positive
definite measure $\gamma$ satisfying \eqref{acEquation} derive the corresponding
second moment measure \eqref{2momentCondition}.
\end{proof}

Existence of a second moment immediately  implies some continuity properties of $N$.

\begin{coro}\label{continuity-N} Let $(N,\cHH,T)$ be a stationary process with  second moment. Let $1 \leq p,q \leq \infty$ be given with $1/p + 1/q = 1$. Then, for any compact $K\subset \GG$, there exists a $C_K\geq 0$ with
$$ |\langle N(F) \mid  N(G)\rangle |\leq C_K \|F\|_{L^p (\GG)} \| G\|_{L^q (\GG)}$$
for all $F,G\in C_c (\GG)$ with support contained in $K$. In particular the map $N : C_c (\GG)\longrightarrow L^2 (X,m)$ is continuous (with respect to the $L^2$-norm on $C_c(\GG)$).
Also $N:C_K(\GG) \longrightarrow  L^2 (X,m)$ is continuous with
respect to the sup norm on $C_K(\GG)$.
\end{coro}

\begin{proof} For $F,G\in C_c (\GG)$ with support contained in $K$ the support of $F\ast \widetilde{G}$ is contained in $K  - K$ and  $\|F\ast \widetilde{G}\|_\infty \leq \|F\|_{L^p (\GG)} \|G\|_{L^q (\GG)}$. As $|\langle N(F), N(G)\rangle | = |\gamma (F\ast \widetilde{G} ) |$, this easily gives the first statement. Now, the second statement follows by taking $p = q = 2$. Finally,
for $F\in C_K(\GG)$, $||F||_{L^2(\GG)} \le ||F||_\infty (l_{\GG}(K))^{1/2}$.
\end{proof}

The continuity property of the preceding corollary allows one to extend the map $N$. This is discussed next.

\begin{coro} \label{extendingN}
Let  $\cN= (N,\cHH,T)$ be a stationary process with second moment. The map  $N$ can be uniquely extended to the vector space of all measurable functions $f$ on $\GG$ all of which vanish outside some compact set and are square integrable with respect to the Haar measure
on $\GG$.
 This extension (again denoted by $N$) is $\GG$-equivariant and  satisfies
\[ \int_\GG F * \tilde F \dd \gamma = \langle N(F) | N(F) \rangle \, ,\]
meaning in particular that the integral exists and is finite. \end{coro}

\begin{proof}  It suffices to show that $N$ can be uniquely extended to $L^2 (U, l_U)$ for any open $U$ in $\GG$ with compact closure.   Let such an $U$ be given.

Let $F\in L^2 (U, l_U)$ be given. As, $U$ is open, we can then find a sequence $\{F_n\}$ in $C_c (U)$ converging to $F$ in the sense of $L^2 (U, l_U)$.
By the previous corollary, we infer then that $\{N (F_n)\}$ must be a Cauchy-sequence, whose limit does not depend on the choice of the approximating sequence $\{F_n\}$. This shows that $N$ can be extended to $L^2 (U, l_U)$.

\smallskip

By construction of $N(F)$  and the definition of $\gamma$ we have furthermore
$$\langle N (F)\mid N(F)\rangle = \lim_{n\to \infty} \int_\GG F_n \ast \tilde F_n d\gamma.$$
 Moreover, a direct application of Cauchy-Schwartz inequality shows that $F_n \ast \tilde F_n$ converges to $F \ast \tilde F$ with respect to the supremum norm and has support contained in $\overline{U}- \overline{U}$. This easily gives that
 $$\lim_{n\to \infty} \int_\GG F_n \ast \tilde F_n d\gamma = \int_\GG F\ast \tilde F d\gamma.$$
This finishes the proof.
 \end{proof}

\begin{remark}
Note that any bounded measurable function with compact support belongs to $L^2 (U,d\,l_\GG)$ for any open  $U$ containing the support. Thus, $N$ can in particular be extended to bounded measurable functions with compact support.
\end{remark}

The above  discussion  shows that any stationary process with a second moment gives rise to a real positive definite measure $\gamma$. As $\gamma$ is positive definite it is transformable i.e. its  Fourier transform $\difm$ exists \cite{BF,SM}, and since $\gamma$ is a real measure, $\difm$
is centrally symmetric ($\difm(A) = \difm(-A)$
for all measurable sets $A$). The measures $\gamma$ and $\difm$ lie at the heart of our investigation.

\begin{definition} Let $(N,\cHH,T)$ be a stationary process with second moment. Then, $\gamma:=\gamma_N$ is called the {\it autocorrelation} of the stationary process. Its Fourier transform $\omega:=\omega_N:=\widehat{\gamma_N}$ is called the {\it diffraction} or {\it diffraction measure} of the stationary process.
\end{definition}

\begin{definition} A stationary process $(N,\cHH,T)$ with second moment is said to have {\it pure point diffraction}
(resp. {\it continuous diffraction}) if the diffraction measure $\difm$ is a pure point measure
(resp. continuous measure).
\end{definition}

The definition of $\difm$ yields that it is a centrally symmetric positive measure and is the unique measure  satisfying

\begin{equation}\label{omega-gamma}
 \int_{\Gdual} |\widehat F|^2 \dd\difm = \int_{\GG} F * \widetilde F \dd \gamma \, ,
 \end{equation}
or equivalently, by linearizing this,
\begin{equation*}
 \int_{\Gdual} \widehat F \overline{\widehat G} \dd \difm
= \int_\GG F*\tilde G \dd\gamma
\end{equation*}
for all $F,G \in C_c(\GG)$. Considering  the Hilbert space  $L^2(\Gdual, \difm)$ with the corresponding inner product
being written $\langle \cdot \mid \cdot \rangle_{\Gdual}$ and using the definition of $\gamma$, we obtain
for any $F,G \in C_c(\GG)$,
\begin{equation} \label{innerProduct}
\langle \widehat F \mid \widehat G \rangle_{\Gdual}
 = \int_{\Gdual} \widehat F \overline{\widehat G} \dd \difm\\
  =  \int_\GG F*\tilde G \dd\gamma\\
 = \langle N(F) \mid N(G) \rangle \,.
\end{equation}
This shows that the map $\widehat{F} \mapsto N(F)$ is an isometry.
We define an action $\hat T$ of $\GG$ on $L^2(\Gdual,\difm)$ by defining
$\hat T_t f$ for all $f\in L^2(\Gdual,\difm)$ through
\[\hat T_t f(k) = \overline{\langle k, t \rangle} f(k) \]
for all $t \in \GG$ and all $k\in \Gdual$.

\begin{remark} We note that if $N$ is a $\cHH$-stationary process on $\GG$
with second moment
then for all $c\in \RR$, $cN$ is another such process. The corresponding
autocorrelation and diffraction measures are then scaled by $|c|^2$.
\end{remark}

\subsection{The diffraction to dynamics map}\label{diff2dyn}
Any process with a second moment comes with an isometric $\GG$-map $\theta$  from $L^2 (\Gdual,\difm)$ to $\cHH$. This allows one to transfer certain questions from $\cHH$ to $L^2 (\Gdual,\difm)$. Also, it means that for certain eigenvalues there are canonical eigenfunctions available. These topics are studied next.

\begin{prop} \label{diffractiontodynamics} Let $\cN=(N,\cHH,T)$ be a stationary process with
second moment.
Then $\{\widehat{F}: F\in C_c (\GG)\}$ is dense in  $L^2 (\Gdual,\difm)$ and
there exists a unique isometric embedding
\begin{eqnarray}\label{difmToDyn}
\theta=\theta_N : L^2(\Gdual, \difm)\longrightarrow \cHH
 \end{eqnarray}
 with $\widehat{F}\mapsto N(F)$ for each $F\in C_c (\GG)$.
This mapping is a $\GG$-map and furthermore, for all $D\in L^2(\Gdual,\difm)$,
 \[\theta(\widetilde{D}) = \overline{\theta(D)}\,.\]
 \end{prop}
\begin{proof} We begin with the denseness statement.
 Let $L$ be the closure of $\{\widehat{F}: F\in C_c (\GG)\}$ in $L^2 (\Gdual,\difm)$.
 We wish to show that $L=L^2(\Gdual,\difm)$.
 As $\difm$ is a translation bounded measure, the space $C_c (\Gdual)$ is dense in $L^2 (\Gdual,\difm)$. Thus it suffices to show that $C_c (\Gdual)$ belongs to $L$.
As the convolutions of elements from $ C_c (\GG)$ belong to $C_c (\GG)$ as well,  we infer that for all $G_1,\ldots, G_m \in C_c (\GG)$
$$ \widehat{F_1}\cdots \widehat{F_n} \cdot \widehat{G_1}\cdots \widehat{G_m} \in L
$$
for all $F_i\in C_c (\GG)$.  Now  $\widehat{F_j}$ belongs to the set $C_0 (\Gdual)$ of continuous functions on $\Gdual$ vanishing at infinity. An application of the Stone-Weierstrass theorem then shows that
$$  f \widehat{G_1}\cdots \widehat{G_m}\in L$$
for each $f\in  C_0 (\Gdual)$. Fixing $f\in C_c (\Gdual)$ we can then use another application of the Stone-Weierstrass theorem, to infer that
$$ f g\in L$$
for all $f\in C_c (\Gdual)$ and $g \in C_0 (\Gdual)$. This yields  $C_c (\Gdual)\subset L$ as we intended to show.

The existence of $\theta$ now follows by defining it initially
by $\theta: \widehat F \mapsto N(F)$ for all $F \in C_c(\GG)$,
using the fact that it is an isometric mapping  by \eqref{innerProduct}, and then extending it
to the $L^2$-closure $L$ which is $L^2(\Gdual,\difm)$.

As for the second statement, first we note that the $\GG$-action on $L^2(\Gdual, \difm)$ is
designed to make the mapping $\hat F \mapsto F$ into a $\GG$-map, while
$F\mapsto N(F)$ is already a $\GG$-map, and second, for all $F\in C_c(\GG)$,
$\widetilde{\widehat{F}} = \widehat{\overline{F}}$, from which
\[\theta\left(\widetilde{\widehat{F}}\right) = \theta( \widehat{\overline{F}})
= N({\overline{F}}) = \overline{N(F)} = \overline{\theta(\widehat F)}\,.\]
We finish using the denseness of the first part.
\end{proof}

\begin{definition} The map $\theta=\theta_N$ associated to a stationary process $\cN=(N,\cHH,T)$ with second moment is called the {\it diffraction to dynamics map}.
\end{definition}

The relevance of $\theta$ for spectral theoretic considerations comes from the following consequence of the previous proposition.

\begin{theorem} \label{theta-as-spectralmeasure}  Let $\cN=(N,\cHH,T)$ be a stationary process with
second moment and $\theta$ the associated diffraction to dynamics map. Then, for any $H\in L^2 (\Gdual,\difm)$ the spectral measure $\rho^{}_{\theta (H)}$ is given by
$$ \rho^{}_{\theta (H)} = |H|^2 \difm.$$
\end{theorem}
\begin{proof} As $\theta$ is a $\GG$-map and an isomorphism a short calculation gives
$$ \langle \theta (H) \mid  T_t \theta (H)\rangle = \langle \theta (H)\mid \theta ( (t,\cdot) H) \rangle = \int_{\Gdual} (t,\xi) |H|^2 d \difm.$$
Thus, the (inverse) Fourier transform  of $|H|^2 \difm$ is given by $t\mapsto \langle \theta (H) \mid  T_t \theta (H)\rangle$.
By the characterization of the spectral measure in \eqref{spectralmeasure} the theorem follows.
\end{proof}

From the previous theorem  it is clear  that if $\cHH$ has pure
point spectrum then $\difm$ is discrete and the diffraction is pure point. However,
it is not clear that pure point diffraction implies pure point spectrum but as we shall
see, it is so: each implies the other in the case of spatial processes.

\medskip

The diffraction to dynamics map connects pure point diffraction to eigenvectors
of the $\GG$ action on $\cHH$. Later, we shall see that in the case of
spatial processes, this allows us to make a complete correspondence between
pure point diffraction and pure point dynamics.

For a  pure point diffractive stationary processes with diffraction measure $\difm$ we define
the set of its {\it atoms} by
$$\Bragg =\Bragg(\difm) :=\{k\in \Gdual: \difm (\{k\})  \ne 0\} = \{k\in \Gdual: \difm (\{k\}) > 0\} .$$
The set $\Bragg$ is often called the {\it Bragg spectrum} of the process, and
its elements are sometimes called {\it Bragg peaks}\footnote{Sometimes, the term Bragg peak
is also taken to mean both the position $k$ of the atom and its intensity $\difm(\{k\})$.}.

In the sequel we will often omit the brackets when dealing with one element sets of the form $\{k\}$. In particular, we will set $ \difm(k):=\difm(\{k\})$
for $k\in \Gdual$.

Let $1_k$ be the characteristic function of the set $\{k\} \subset \Gdual$ i.e. $1_k(k')$ is
$1$ or $0$ according as $k'$ equals $k$ or not. It is easy to see that these functions
are the only possible eigenfunctions for our action of $\GG$ on $L^2(\Gdual, \difm)$.
 Define $f_k$, $k\in \Bragg$, by $f_k = \theta(1_{{k}})$.
Then each $f_k$ is an eigenfunction in
$\cHH$ for the eigenvalue $k$, in the sense that $T_t f_k = \overline{\langle k,t\rangle}f_k$.

\medskip

\begin{lemma}\label{fminusk} Let $\cN=(N,\cHH,T)$ be a stationary process with
second moment. Then,
for all $k \in \Bragg$, $-k \in \Bragg$,
and $f_{-k} = \overline{f_k}$ and $||f_{-k}|| = ||f_k|| =\difm(k)^{1/2}$.
\end{lemma}
\begin{proof} Since $\difm$ is centrally symmetric, $\difm(\{-k\}) = \difm(\{k\})$,
which gives the first statement. For the second, using Prop. \ref{diffractiontodynamics}
we have $f_{-k} = \theta(1_{-k}) = \theta(\widetilde{1_k})  =
\overline{\theta(1_k)} = \overline{f_k}$. Then
$||f_{-k}||^2 = \langle f_{-k} \mid f_{-k}\rangle = \langle f_{k} \mid f_{k}\rangle =
||f_k||^2 = \langle 1_k, 1_k\rangle = \difm(k)$.
\end{proof}

We can use the preceding considerations to compute the map $N$ in the case of pure point diffraction.

\begin{prop} \label{eigenExpansionOfN(F)}
Let $(N,\cHH,T)$ be a stationary process with pure point diffraction and  associated Bragg peaks $\Bragg$. Let $\theta$ be the associated diffraction to dynamics map and $f_k = \theta (1_k)$, $k\in \Bragg$. Then
$$ N(F)=\sum_{k\in \Bragg} \widehat{F} (k) f_k$$
for all $F\in C_c (\GG)$.
\end{prop}
\begin{proof} By definition of $\theta$ we have $N(F) = \theta(\widehat{F})$. Obviously,  $\widehat{F} = \sum_{k\in \Bragg} \widehat{F} (k) 1_k$ and the claim follows.
\end{proof}

\subsection{Spatial stationary processes:  Two-point correlation} \label{spatial}

In this section  we show how the autocorrelation can be given a meaning that agrees with
`classical' two-point correlation associated to stationary point processes.
\bigskip

Assume  that we are given an ergodic spatial stationary process $\cN= (N,X,\mu,T)$.
As discussed in Corollary  \ref{extendingN} we can assume that $N$ is defined on all measurable bounded functions with compact support.
We also assume that $\GG$ has a countable basis of topology and is hence metrizable. Let  $\{A_n\}_{n=1}^\infty$ be a van Hove sequence for $\GG$ (as discussed in  Section \ref{notation}). We assume that
this is fixed once and for all.

\begin{definition} \label{acDef} For $\xi \in X$ we define the {\em two-point correlation or  autocorrelation} of $\xi$
at any $F \in C_c(\GG)$ as
\[\lim_{B \to 0} \frac{1}{\haarG(B)} \left( \lim_{n\to\infty} \frac{1}{\haarG(A_n)}
\int_{A_n}
(N(1_B) N(F))(T_{-x} \xi) \dd \haarG(x)\right) \,\]
whenever the limit exists.
\end{definition}

This needs some comments: The limits in the definition are taken in the order indicated: first the inner and then the outer.
$B$ is an open (or measurable) neighbourhood of $0$ in $\GG$. The statement
$B\to 0$ means that we take a nested descending sequence $\{B_m\}$ of such neighbourhoods, all within some fixed compact set $K$, and that $\haarG(B_m) \to 0$. We are using the notation $\dd x$ to stand for the longer $\dd \haarG(x)$.
The definition requires that we give meaning to $N(1_B)$ as a measurable
function on $\GG$. This uses the extension of $N$ from $C_c (\GG)$ to $L^2$-functions with compact support given in Corollary  \ref{extendingN}.

The intuition behind the definition is as follows. The two-point correlation  at
$\xi \in X$ for $F\in C_c(\GG)$
should look something like
\[ \lim_{n\to\infty} \frac{1}{\haarG(A_n)} \int_{A_n}\int_{A_n}  \xi(x)\xi(y) F(-x+y) \dd x
\dd y
= \lim_{n\to\infty} \frac{1}{\haarG(A_n)} \int_{A_n} \xi(x)
\int_\GG \xi(y) (T_x F)(y) \dd y
\dd x  \, ,\]
where the right hand side arises by using the usual trick from van Hove sequences and the compactness of the support
of $F$. (That is, the difference of the two sides of the equation is due to the difference
between $-A_n$ and $-A_n+ \supp(F)$, which by the van Hove assumption is irrelevant in the limit.)  Of course in our case $\xi$ is not a function of $x$ and the integrands
do not make sense. But the inner integral on the right-hand side is what should
be $N(T_x F)(\xi) = N(F)(T_{-x} \xi)$ and we can rewrite this `autocorrelation' as
\[\lim_{n\to\infty} \frac{1}{\haarG(A_n)} \int_{A_n} (T_{-x}\xi)(0) N(F)(T_{-x} \xi)
\dd l_\GG(x)  \, .\]
The term $(T_{-x}\xi)(0)$ has no meaning. But Palm theory tells
us how to go around this. We instead average over a small neighbourhood $B$
of $0$, and this brings us to Definition~\ref{acDef}.

\smallskip

\begin{theorem} \label{acThm}  Let $\GG$ be a locally compact group whose topology has a countable basis and let $(N,X,\mu,T)$ be an ergodic spatial stationary process on
$\GG$ with second moment.
The two-point correlation of $\xi \in X$ exists $\mu$-almost surely and
its value at $F\in C_c(\GG)$ is $\gamma(F)$.
\end{theorem}
\begin{proof} The function
$(N(1_B)N(F))(T_{-x} \xi)$ is  measurable as a function of $x\in \GG$ and the Birkhoff
theorem says that almost surely
\begin{equation} \label{BET}
\lim_{n\to\infty}\frac{1}{\haarG(A_n)} \int_{A_n} (N(1_B) N(F))(T_{-x} \xi) \,\dd x
= \int_X N(1_B)N(F) \dd \mu \, ,
\end{equation}
meaning that the limit will exist and equal the right hand side. We shall prove that

\begin{equation}\label{Bto0}
\lim_{B \to 0} \frac{1}{\haarG(B)}\int_X N(1_B)N(F) \dd \mu
= \gamma(F) \, ,
\end{equation}
proving that the definition of Definition~\ref{acDef} works almost
surely in $\xi$ for each $F \in C_c(\GG)$ and that $\gamma$ does have
the meaning of an autocorrelation.

This has to be made to work simultaneously for all $F$ in $C_c(\GG)$, which will be shown in the usual way from a countable basis of $C_c(\GG)$. This will prove
Theorem ~\ref{acThm}.

\smallskip

Here are the details:  It is enough to prove \eqref{Bto0} for real-valued
$F$. Using \eqref{innerProduct}  we have for all
$F,G\in C^\RR_c(\GG)$,
\begin{equation*} \label{linearization}
\int_\GG F*\tilde G \dd \gamma = \langle N(F) \mid N(G) \rangle \,.
\end{equation*}

 Let $F\in C_c^\RR(\GG)$ and choose $\epsilon >0$.  Let $B$ be a measurable subset of $\GG$ with compact
 closure. Then by Prop.~\ref{extendingN}, $N(B):= N(1_B)$ is defined and is
 a measurable $L^2$-function on $X$. By linearization
 we have
\begin{eqnarray*}
\int_X N(1_B)N(F) \dd \mu  &=& \int_\GG 1_B *\tilde F \dd\gamma
\nonumber\\
&=& \int_\GG \int_\GG 1_B(x)\tilde F(y-x) \dd x \dd \gamma(y)
=\int_\GG \int_\GG 1_B(x) F(x-y) \dd x \dd \gamma(y) \, .
\end{eqnarray*}
Since $F$ is uniformly continuous on $\GG$, for any sufficiently small
neighbourhood $B$ of $0$, $|F(x-y) - F(-y)| <\epsilon$ for all  $x\in B$ and
for all $y \in \GG$. Then
\[
\left|\int_\GG 1_B(x) F(x-y) \dd x -\haarG(B) F(-y) \right|
\le \int_\GG 1_B(x) | F(x-y) - F(-y)| \dd x < \epsilon\, \haarG(B)\]
for all $y \in \GG$, so
\begin{eqnarray*}
\frac{1}{\haarG(B)}&&\left|\int_X N(B)N(F) \dd \mu
- \haarG(B) \int_\GG F(-y) \dd\gamma(y) \right| \\
&=&\frac{1}{\haarG(B)} \left| \int_\GG \int_\GG 1_B(x) F(x-y) \dd x \dd \gamma(y) - \haarG(B) \int_\GG F(-y) \dd\gamma(y) \right| \\
&\le& \frac{1}{\haarG(B)} \int_\GG \int_\GG 1_B(x) | F(x-y) - F(-y)| \dd x
\dd\gamma(y)\\
& \le& \frac{1}{\haarG(B)} \epsilon \,\haarG(B)
|\gamma|(-\supp(F) +B).
\end{eqnarray*}
Since $\gamma(y) = \gamma(-y)$ ($\gamma$ is positve-definite and real, Prop.~\ref{def:gamma}),
we see that as $B \to 0$ (and correspondingly $\epsilon \to 0$)
we have a proof of \eqref{Bto0}. Together with \eqref{BET} we have
the desired interpretation of Def.~\ref{acDef} as the autocorrelation
at $\xi$ and its almost sure equality with $\gamma(F)$, and Prop.~\ref{acThm}
is proved.
\end{proof}

\begin{remark} This type of result was first established for certain uniformly discrete point processes on Euclidean space  by Hof \cite{Hof1} (see \cite{Martin} for the case of general locally compact abelian groups). This was then extended to rather general point processes by Gou\'{e}r\'{e}  \cite{Gouere} and to certain measure processes by Baake / Lenz  \cite{BL}. A unified treatment was then given by Lenz / Strungaru  \cite{LS}. Our result contains all these results (provided the underlying process is real).
\end{remark}

\subsection{A second glance at second moments}
The discussion above has shown that existence of second moments has strong consequences. It implies  existence of
 the diffraction to dynamics map by Proposition \ref{diffractiontodynamics} and  further  continuity properties
 by Corollary \ref{continuity-N}.  It turns out that a converse of  sorts holds. This is investigated in this
 section. Along the way, we will also show that continuity of $N$ implies an intertwining property of the map $N$ and this will be crucial for our considerations.

\medskip

Any process $(N, \cHH, T)$ gives rise to two representations of $C_c (\GG)$ (and in fact even of $L^1 (\GG)$): The representation $L$ lives on $L^2 (\GG)$ and acts by
$$ L_G F = G\ast F$$
for $F,G\in C_c (\GG)$.
The representation $T$ (extending the action $T$ and therefore denoted by the same letter) acts by
$$ T_G f := \int G (t) T_t f \, dl_\GG (t) $$
for $G\in C_c (\GG)$ and $f\in \cHH$.  Continuity of $N$ now yields an intertwining property:

\begin{lemma} \label{intertwiner} Let $(N, \cHH, T)$ be a process with a continuous $N$. Then,
$$N\circ L_G = T_G \circ N$$
for all $G\in C_c (\GG)$.
\end{lemma}
\begin{proof} We have
\[N(L_G (F)) = N(G \ast F) = N\left(\int_\GG G(t)T_t F \,d\,l_\GG(t)\right)\,.\]
By the linearity and continuity of $N$,  it commutes with taking integrals and
\[N\left(\int_\GG G(t)T_t F dl_\GG(t)\right) = \int_\GG  G(t) N(T_t F)\,  dl_\GG(t) =\int_\GG G (t) T_t N(F) dl_\GG (t) \]
holds, which finishes the proof.
\end{proof}

We will need a somewhat stronger continuity property of $N$.

\begin{definition} Let $1\leq p,q\leq \infty$ with $1/p + 1/q = 1$ be given.  The process $(N, \cHH, T)$  is said to be weakly $(p,q)$-continuous if for any
compact $K\subset \GG$ there exists a $C_K$ with
$$|\langle N(F) \mid N(G)\rangle| \leq C_K \|F \|_{L^p (\GG) } \|G\|_{L^q (\GG)}$$
for all $F,G\in C_c (\GG)$ with support contained in $K$ (see Cor.~\ref{continuity-N}).
\end{definition}
\begin{remark} (a)  Note that continuity of $N$ is exactly weak $(2,2)$-continuity.

(b) By standard interpolation theory, we can conclude that weak $(2,2)$ continuity together with weak $(1,\infty)$ continuity implies
weak $(p,q)$ continuity for all $p,q$ with $1\leq p,q\leq \infty $ and $1/p + 1/q = 1$.
\end{remark}

\begin{theorem} Let $(N, \cHH, T)$ be a stochastic process.  Then, the following assertions are equivalent:

\begin{itemize}
\item[(i)] $N$ has a second moment.

\item[(ii)] $N$ is weakly $(p,q)$-continuous for all $1\leq p,q\leq \infty$ with $1/p + 1/q = 1$.

\end{itemize}

In this case $N$ is continuous and  has the intertwining property  that
$$ N \circ L_G (F) = T_G N (F)$$
for all $F,G\in C_c (\GG)$.
\end{theorem}
\begin{proof} (i)$\Longrightarrow$ (ii): This follows immediately from Corollary \ref{continuity-N}.

\smallskip

(ii)$\Longrightarrow$ (i): As $N$  is weakly $(2,2)$ continuous, it is continuous and the intertwining property follows from the previous lemma.
In the remaining part of the proof we will only consider real valued functions. By Proposition \ref{def:gamma}, it suffices to show
existence of a measure $\gamma$ on $\GG$ with
$$\langle N(F) \mid N(G)\rangle = \gamma (F\ast \tilde G)$$
for all $F,G\in C_c (\GG)$. A short calculation shows that for such a $\gamma$ the equality
$$ \langle N(F) \mid N(G)\rangle = \int F(y) (\gamma \ast G) (y) d\haarG(y)$$
 must hold for all $F,G\in C_c (\GG)$. The idea is now to `reverse' this reasoning to conclude existence of $\gamma$. The main issue
is to show continuity of the object $H_G$ replacing the not yet defined $\gamma \ast G$. This is shown using the intertwining property. Here are the details:
We choose an arbitrary compact $K$ in $G$ and assume without loss of generality that $0\in K$. Let $U$ be an open neighborhood of $0$ with compact closure and set
$K_1 $ to be the closure of $K + U$. Set $L = K_1 - K_1 + K_1 - K_1$.

As $N$ is weak $(1,\infty)$ continuous, and $L^\infty (L)$ is the dual of $L^1 (L)$, we can find  to each $G\in C_c (\GG)$ with support in
 $K_1$ a function $H_G$ in $L^\infty (L)$ with
 $$ \langle N(F) \mid N(G)\rangle = \int F(y) H_G (y) d\haarG (y)$$
 for all $F\in C_c (\GG)$ with support in $L$ and
\begin{equation}\label{star2.9}  \|H_G\|_\infty \leq C \|G\|_\infty, \end{equation}
 where $C = C (L)$. As $N$ has the intertwining property, a short argument shows that
 $$H_{G_1\ast G_2} (x) = \int G_1 (y) H_{G_2} (x - y) d\haarG (y)$$
for $x\in K$ for all $G_2 \in C_c (\GG)$ with support contained in $K$ and all $G_1\in C_c (\GG)$ with support contained in $U$. This gives
in particular, that $H_{G_1\ast G_2}$ is a continuous function (on $K$) for all such $G_1, G_2$. Now
\eqref{star2.9} yields
$$\| (H_{D \ast G} - H_G)\|_\infty = \| H_{D\ast G- G} \|_\infty \leq C \| D \ast G - G\|_\infty.$$
Taking an approximate unit for $D$, we wee that $H_G$ can be approximated uniformly by continuous functions. This shows continuity of $H_G$. Moreover, by $\eqref{star2.9}$  again
the map
$$G \mapsto H_G (x)$$
is continuous for each $x$. Thus, we can indeed define the measure $\gamma$ with
$$\gamma \ast G (x) = H_G (x).$$
By construction $\gamma$ has the desired properties.
\end{proof}

\section{Almost periodicity}\label{Almost}
Almost periodicity is closely linked to nature of diffraction via the Fourier transform. This concept allows one to compare properties of the autocorrelation and the diffraction measure. In particular, it can be used to characterize   pure point and continuous diffraction. The considerations of this section play a role in subsequent parts of the paper: They are used in Section \ref{pc:split}  to decompose an arbitrary  processes into a part  with pure point diffraction and a part with continuous diffraction, and in Section \ref{Purepoint} to prove the equivalence of pure point diffraction spectrum and pure point spectrum.  The material of this section derives from \cite{GA} and, in a more accessible account \cite{SM}, and  from  \cite{LR,LS} as well. For further studies of aspects of  almost periodicity in our context we refer to \cite{Gouere,Sol,Lag}.

\bigskip

Almost periodicity is most commonly defined by a compactness condition: $f \in C_u(\GG)$ is almost periodic if the translation orbit  $\GG.f$ of $f$ is compact. Here $C_u(\GG)$
is the space of uniformly continuous $\CC$-valued functions on $\GG$. The key point,
however, is in which topology is the translation orbit to be compact? In the case of the
sup-norm topology the concept defines {\it strong almost periodicity}. In the case
that the topology is defined by the family of semi-norms induced by the set of
all continuous linear functionals on $C_u(\GG)$, it is called {\it weak almost periodicity}.
Strong almost periodicity coincides with H.~Bohr's original concept of almost periodicity:
$f\in C_u(\GG)$ is strongly periodic if and only if for every $\epsilon >0$ the set
of $t\in \GG$ for which $||T_t f - f||_{\sup} < \epsilon$ is relatively dense. It is harder
to get an intuitive feel for weak almost periodicity. Fortunately what one needs to know
about it is fairly straightforward:

\begin{itemize}
\item all positive definite functions on $\GG$ are weakly almost periodic;
\item for every weakly almost periodic function $f$ the {\it mean} of $f$ defined as
\[ \lim_{n\to\infty} \frac{1}{\haarG(A_n)} \int_{A_n} f(x+t) \dd\haarG(t) \,\]
exists for any van Hove (more generally F\o lner) sequence $\{A_n\}$ in $\GG$ and any $x\in \GG$ and
its value is independent of both $x$ and the choice of van Hove sequence.
\end{itemize}

One says that a weakly almost periodic function $f$ is {\it null weakly almost periodic} if its mean is $0$. These concepts lift to measures in a simple way: a Borel measure  $\phi$ on $\GG$ is called {\it strongly} (resp. {\it weakly}, {\it null weakly}) {\it almost periodic} if $f*\phi$ is a strongly (resp. weakly, null weakly) almost periodic function
for every $f\in C_c(\GG)$.

\smallskip

Positive definite measures on $\GG$ turn out to be
weakly almost periodic and are also, as we have noted already, Fourier transformable.

\begin{theorem} \label{thm-GA} {\rm (Gil de Lamadrid, Argabright \cite{GA})}
Every weakly almost periodic measure $\phi$ is uniquely expressible as the sum of a strongly
almost periodic measure and null weakly almost periodic measure. If the measure $\phi$ is
Fourier transformable then so too are the strong and null weak components, and
furthermore the Fourier transforms of these components are the pure point and continuous parts of $\widehat\phi$.  \qed
\end{theorem}

These considerations can be applied to strongly continuous unitary representations and hence to processes as well. This is discussed next.  The crucial connection is given by the following (well known) lemma.

\begin{lemma} \label{Fft-is-wap}  Let $T$ be a strongly continuous unitary representation of $\GG$ on $\cHH$ then for any $f \in \cHH$, the function $t\mapsto \langle f \mid T_t f\rangle =:F_f (t) $ is positive definite and hence weakly almost periodic and Fourier transformable. Its Fourier transform is the spectral measure $\rho^{}_f$.
 \end{lemma}
\begin{proof} As $T$ is unitary, we have  for any $t_1,\ldots, t_n\in \GG$ and $c_1,\ldots, c_n\in \CC$
$$ \sum_{k,l=1}^n c_k \overline{c_l} F_f (t_k - t_l) = \|\sum_{l=1}^n c_l F_f (t_l)\|^2 \geq 0$$
and $F_f$ is shown to be positive definite. The remaining statements follow from the above discussion.
\end{proof}

Combining the previous lemma and the previous theorem, we infer the following corollary.

\begin{coro} Let $T$ be a strongly continuous unitary representation of $\GG$ on $\cHH$ then for any $f \in \cHH$, the function $t\mapsto \langle f \mid T_t f\rangle$ is the sum of a strongly almost periodic function and a null weakly almost periodic function. This strongly almost periodic function is given by
the (inverse) Fourier transform of the pure point part of the spectral measure $\rho^{}_f$ and the null weakly almost periodic function is given by the (inverse) Fourier transform of the continuous part of the spectral measure $\rho^{}_f$.
\end{coro}

Of particular relevance  is the question whether the spectral measures are pure point. Here, the following holds as shown in \cite{LS}, Lemma 2.1.

\begin{lemma}\label{characterizationpp}
Let $T$ be a strongly continuous unitary representation of $\GG$ on $\cHH$.
 Then, the following assertions are equivalent for $f\in \cHH$:

\begin{itemize}

\item[(i)] The map $G\longrightarrow \cHH$, $t\mapsto T_t f$, is almost periodic in the sense that for any $\varepsilon >0$ the set $\{t \in \GG : \|T_t f - f \|\leq \varepsilon\}$ is relatively dense in $\GG$.

\item[(ii)] The hull $\{T_t f: t\in \GG\}$ is relatively compact.

\item[(iii)] The  map $G\longrightarrow \CC$, $t\mapsto \langle f \mid T_t f \rangle $, is strongly almost periodic.

\item[(iv)] $\rho_f$ is a pure point measure.

\item[(v)]  $f$ belongs to $\cHH_p$.
\end{itemize}
\end{lemma}

The preceding equivalence allows one to show some stability properties of $\cHH_p$. This is discussed next (see \cite{LS} as well for a proof of (a) of the following theorem).
There, we say that a \textit{measure $\mu$ is supported on a measurable set $S$}  if $\mu$ of the complement of $S$ is zero.

\begin{theorem}\label{stabilitypp}
Let $T$ be a strongly continuous unitary representation of $\GG$ on
$\cHH$.

(a)  Let $C : \cHH \longrightarrow \cHH$ be continuous with $T_t C f =
C T_t f$ for each $t\in G$ and $f\in \cHH$.  Then, $C$ maps $\cHH_p$ into
$\cHH_p$. If $f$ belongs to $\cHH_p$ and $\rho_f$ is supported on the subgroup $S$
of $\Gdual$, then so is $\rho_{Cf}$.

(b)  Let $M : \cHH\times \cHH\longrightarrow \cHH$ be continuous with $T_t M (f,g) = M (T_t f, T_t g)$. Then, $M $ maps $\cHH_p \times \cHH_p$ into $\cHH_p$.
\end{theorem}
\begin{proof} The statements on $C$ and $M$ are   clear from the equivalence of (v) and (ii) in the Lemma~\ref{characterizationpp} and the continuity and intertwining property of $C$ and $M$. The statement on the support in (a) is proven in \cite{LS}.
\end{proof}

\medskip

We now further restrict attention to processes with second moment. Then,  Theorem \ref{thm-GA}  can directly be applied to give the following result (which is well known from  \cite{BM,LS,MS, Gouere}).

\begin{prop} A stationary process $(N,\cHH,T)$ with second moment is pure point diffractive if and only if its autocorrelation is strongly almost periodic. It has continuous diffraction if and only if the autocorrelation is null weakly almost periodic. \qed
\end{prop}

We can then combine the above results to obtain the following.

\begin{prop}  For any stationary process $(N,\cHH,T)$ with second moment
the mapping \;$t \mapsto \langle N(F)\mid T_t N(F)\rangle$ is
a weakly almost periodic function on $\GG$ for all $F \in C_c(\GG)$.
The process  $(N,\cHH,T)$
then has pure point (resp. continuous) diffraction if and only if \;$t \mapsto \langle N(F)\mid T_t N(F)\rangle$ is
a strongly (resp. null weakly) almost periodic function on $\GG$ for all $F \in C_c(\GG)$.
Furthermore, the diffraction is pure point if and only if for all $F\in C_c(\GG)$,  $t \mapsto T_t N(F)$ is a strongly almost periodic function on $\GG$ with respect to the norm on
$\cHH$.
\end{prop}
\begin{proof} The first statement follows from Lemma \ref{Fft-is-wap}. We now turn to the second statement.  By Theorem \ref{thm-GA} and the definition of the diffraction spectrum we have pure point (resp. continuous) diffraction if and only if $\gamma$ is strongly  (resp null weakly) almost periodic. Now,
  from Proposition \ref{def:gamma} and a  straightforward calculation we obtain  for all $F,G \in C^\RR_c(\GG)$
and  all $t \in \GG$,
\begin{equation} \label{conv2N}
 G  * \widetilde{F} * \gamma(t) =  \gamma(G *\widetilde{T_t F})
=\langle N(G) \mid  T_t N(F)\rangle\,.
\end{equation}
With $G$ set equal to $F$ we then get the required almost periodicity properties of $t \mapsto \langle N(F) \mid  T_t N(F)\rangle$ from the corresponding almost periodicity properties of  $\gamma$, by
our definition of these concepts. In the reverse direction, upon linearizing the
expressions $\langle N(F) \mid  T_t N(F)\rangle$ we obtain \eqref{conv2N}. From the almost periodicity properties of $t\mapsto  \langle N(G) \mid  T_t N(F)\rangle$ we then obtain
the desired almost periodicity properties  for all convolutions of the form $G^{'} \ast \widetilde{F} \ast \gamma (-t)$.
An approximate unit argument allows us to get the strong (resp. null weak)
almost periodicity of $\gamma$ from this, see \cite{SM}, Corollary  9. Here, we use that these almost periodicity properties are stable under uniform convergence.

\smallskip

The last statement is a consequence of Lemma \ref{characterizationpp}.
This finishes the proof.
\end{proof}

We finish this section with a discussion of stability properties of almost periodicity in the case of spatial processes.  We will need the following proposition on continuity of composition with a fixed function. The proposition  is not hard to prove  and can be found in  \cite{LS}.

\begin{prop}\label{continuity-of-phi}  Let $\{h_n\}$ be a sequence in $L^2 (X,\mu)$ converging to $h\in L^2 (X,\mu)$. Let $\phi$ be a continuous bounded function from the complex numbers to the complex numbers. Then, $\phi \circ h_n $ converges to $\phi \circ h$.
\end{prop}

Now, we turn to the following consequence of Theorem \ref{stabilitypp} (see \cite{LS} as well).

\begin{coro} \label{coro-stabilitypp}
Let  $(N,\cHH,T)$ be a spatial process with $\cHH = L^2 (X,\mu)$.  Let $\phi, \psi : \CC\longrightarrow \CC$ be bounded continuous functions and $f,g\in \cHH_p$ be given. Then, the following holds:

(a) The function  $\phi \circ f$ belongs to $\cHH_p$ and the spectral measure $\rho^{}_{\phi \circ f}$ is supported in the subgroup $S$ of $\Ghat$  if $\rho^{}_f$ is supported in this subgroup.

(b) The function $(\phi \circ f) (\psi \circ g)$ belongs to $\cHH_p$. If both $\rho^{}_f$ and $\rho^{}_g$ are supported in the subgroup $S$ of $\Ghat$  then so is $\rho^{}_{(\phi \circ f)(\psi \circ g)}$.
\end{coro}

\begin{proof} By the previous proposition the maps $C: \cHH\longrightarrow \cHH$, $f\mapsto \phi \circ f$ and $M:\cHH\times \cHH\longrightarrow \cHH, M(f,g) :=(\phi \circ f) (\psi \circ g)$ are continuous. They obviously commute with the action of $\GG$. Both (a) and the first statement of (b)  follow from Theorem \ref{stabilitypp}. It remains to show the second statement of (b): Let $\cHH^{S}_p$ be the subspace of $\cHH_p$ consisting  of elements whose spectral measure is supported on $S$. Then, it is not hard to see that $\cHH^{S}_p$ is a closed subspace of $\cHH$.
By (a) and the assumption on $f$ and $g$, the spectral measures of  both $\phi\circ f$ and $\phi \circ g$  belong to $\cHH^{S}_p$. Moreover, $\phi\circ f$ and $\phi \circ g$ are bounded functions. Thus, it suffices to show that $\cHH^{S}_p \cap L^\infty (X,\mu)$ is an algebra. This can be shown by mimicking the proof of Lemma 1 in  \cite{BL} (see \cite{LMS1} as well). For the convenience of the reader we sketch a proof:

Note first that every
eigenfunction can be approximated  by bounded eigenfunctions via a simple cut-off
procedure viz if  $e$ is an eigenfunction, then  for an arbitrary $m>0$,
the function
\begin{equation} \label{cutoff}
   e^{(m)} (x) \; := \; \begin{cases} e(x ),
           & |e(x )|\le m \\ 0 \, ,
           & \text{otherwise}
           \end{cases}
\end{equation}
is again an eigenfunction (to the same eigenvalue). Evidently, the $e^{(m)}$ converge to $e$ in $L^2$ as
$m\rightarrow\infty$. Now, let  non-zero functions  $a,b\in \cHH^{S}_p\cap L^\infty (X,\mu)$ be given and choose $\varepsilon>0$ arbitrarily. As
 $a$ belongs to  $\cHH^S_p$, there exists
a finite linear combination $a^{'}= \sum a_i e_i$ of eigenfunctions to eigenvalues in $S$ with
\begin{equation*}
   \|a  - \sum a_i e_i \|_2 \; \le \; \frac{ \varepsilon}{\|b\|_\infty}\, .
\end{equation*}
By the conclusion after  Eq.~(\ref{cutoff}), we can
assume that all $e_i$ are
bounded functions. Thus, in particular, $\|a^{'}\|_\infty < \infty$.

Similarly, we can choose another finite linear combination $b^{'} = \sum b_j
e_j$ of bounded eigenfunctions to eigenvalues in $S$ with
\begin{equation*}
   \|b - \sum b_j e_j\|_2  \; \le \;
   \frac{\varepsilon}{\|a^{'}\|_\infty}  \, .
\end{equation*}
Then,
\begin{equation*}
   \| a b- a^{'} b^{'}\|_2 \; \le \; \|a^{'}\|_\infty \, \| b - b^{'}\|_2 +
   \|b\|_\infty \, \|a - a^{'} \|_2  \; \le \; 2\hspace{0.5pt} \varepsilon.
\end{equation*}
The   product of bounded eigenfunctions to eigenvalues in $S$  is again a bounded
eigenfunction to an eigenvalue in $S$ (as $S$ is a group) and  the claim follows.
\end{proof}

\section{Decomposing  processes, the  pure  point-continuous split, and the
Gelfand construction}\label{pc:split}

In this section we first discuss  how any stationary process splits
into two subprocesses, one of which has pure point spectrum and one of which has continuous spectrum. If the original process has a second moment then so have the two subprocesses  and their diffraction measure will be pure point and purely continuous respectively (Theorem \ref{split}).  If furthermore the original process is spatial, then its  pure point part  is canonically isomorphic to a spatial pure point process (Theorem \ref{prop-decomposition}). This last result is based on
the Gelfand construction of commutative Banach algebras.

\bigskip

Whenever $T$ is a strongly continuous unitary representation of $\GG$ on the Hilbert space $\cHH$  and $U$ is a $T$ invariant closed  subspace of $\cHH$ with  corresponding orthogonal projection $P_U$ we can form the restriction $T_{U}$ of $T$ to $U$ and this will be a strongly continuous unitary representation as well. The orthogonal complement $U^\perp$ of $U$ in $\cHH$ is also $T$ invariant and we can form the restriction $T_{U^\perp}$ as well.
 We will now focus on the  special decomposition into the  continuous and its point part. More precisely, we  can define the continuous Hilbert space $\cHH_c$ and the point Hilbert space $\cHH_p$ by
$$  \cHH_p:=\{f\in \cHH : \rho_f \mbox{ is a pure point measure}\},\;\: \cHH_c:=\{f\in \cHH : \rho_f  \mbox{ is continuous}\}. $$
Then, both $\cHH_p$ and $\cHH_c$ are closed $T$-invariant with
$$\cHH=\cHH_p \oplus \cHH_c.$$
With the projections $P_c$ and $ P_p$ onto $\cHH_p$ and $\cHH_c$ we then have $ P_c \oplus P_p = id$ and $P_* T_t = T_t P_* =: T_t^*$ for $*=p,c$. This gives the decomposition of the representation $T$ as $T = T_p \oplus T_c$ with $T_*:= P_* T$, $* = p,c$. Obviously, all spectral measures associated to $T_p$ are purely discrete and all spectral measures associated to $T_c$ are purely continuous.  Accordingly, we call $T_p $ the \textit{pure point part  of the representation $T$} and $T_c$ the \textit{continuous part of the representation $T$}.

\smallskip

 We now turn to the case  that we are given a stationary process $(N,\cHH,T)$.
 Of course, $T$ is then a strongly continuous unitary representation of $\GG$ on $\cHH$.  So, any $T$ invariant subspace $U$ gives rise to a decomposition of $T$. Moreover, the $\GG$-invariance of $N$ then implies that $N$ can as well  be decomposed into $N_U := P_U N$ and $N_{U^\perp} = P_{U^\perp} N$. Accordingly, the process $\cN$ can be decomposed in to the two processes $\cN_U := (N_U, U, T_U)$ and $\cN_{U^\perp} := (N_{U^\perp}, U^\perp, T_{U^\perp})$ and we write $\cN = \cN_U \oplus \cN_{U^\perp}$.

 \smallskip

 As just discussed,  $T$  can  be decomposed into its pure point part and its continuous part.  into $N = N_p \oplus N_c$ with $N_*=P_* N$, $*=p,c$.  Hence, any process $\cN:= (N,\cHH,T)$ can be decomposed into the two processes $\cN_p:= (N_p,\cHH_p,T_p)$ and $\cN_c:= (N_c,\cHH_c,T_c)$ which are called the \textit{pure point part} and the \textit{continuous part} of $\cN$ respectively.

 Now, assume that the  original process has a second moment and associated autocorrelation measure $\gamma$ and diffraction measure $\omega = \widehat{\gamma}$. By Theorem \ref{thm-GA},  this $\gamma$ can be decomposed uniquely into its strongly almost periodic part $\gamma_{sap}$ and its null weakly almost periodic part $\gamma_{0wao}$ and the Fourier transform of $\gamma_{sap}$ is the pure point part $\omega_p$ of $\omega$ and the Fourier transform of $\gamma_{0wap}$ is the continuous part $\omega_c$ of $\omega$. It turns out that these are just the autocorrelation  and diffraction measures of the pure point part and the continuous part of $\cN$. This is the content of the next theorem.

 \begin{theorem} \label{split}  Let $\cN = (N, \cHH, T)$ be a stationary process with a second moment and associated autocorrelation $\gamma$, diffraction measure $\omega$ and diffraction to dynamics map $\theta$. Let $ \gamma_{sap}$ and $\gamma_{0wap}$ be the  strongly almost periodic part and the  null weakly almost periodic part  of $\gamma$  and $\omega_p$ and $\omega_c$ the respective Fourier transforms. Then,  the point part $\cN_p$  of $\cN$ has a second moment with autocorrelation given by $\gamma_{sap}$ and diffraction measure given by $\omega_p$ and the continuous part $\cN_c$ of $\cN$ with autocorrelation given by  $\gamma_{0wap}$ and diffraction measure given by $\omega_c$.
 \end{theorem}
\begin{proof} By  Proposition \ref{def:gamma} it suffices to show  that
\begin{equation} \label{toshow}  \gamma_{sap} (F\ast \tilde G) = \langle P_p N (F)\mid P_p N (G)\rangle\:\;\mbox{and}\;\: \gamma_{0wap} (F\ast \tilde G) = \langle P_c N(F)\mid P_c N(G)\rangle \end{equation}
for all $F,G\in C_c (\GG)$.   By linearity it suffices to consider the case  $F = G$.
 Let $F_t$ be the function $T_t F = F(\cdot- t)$  and consider the function $t \mapsto \gamma (F\ast  \tilde F_t)$.
We will decompose this function in a strongly almost periodic part and a null weakly almost periodic part in two ways and then conclude the desired statement from the uniqueness of the decomposition. Here are the details:
  By definition of $\gamma_{sap}$ and $\gamma_{0wap}$  we can decompose this function as
 $$ \gamma (F \ast  \tilde F_t) = \gamma_{sap} (F \ast  \tilde F_t ) + \gamma_{0wap} ( F \ast  \tilde F_t)$$
 with a strongly almost periodic function $t \mapsto \gamma_{sap} (F \ast  \tilde  F_t )$ and a null weakly almost periodic function $t\mapsto \gamma_{0wap} ( F\ast  \tilde  F_t)$. On the other hand by definition of $\gamma$ and the fact that $P_c$ and $P_p$ are orthogonal we also obtain the decomposition
 $$ \gamma (F  \ast \tilde F_t) = \langle N(F) \mid T_t N(F )\rangle =  \langle P_p N(F) \mid T_t P_p N(F )\rangle + \langle P_c N(F)\mid T_t P_c N(F)\rangle.$$
 Now, $t\mapsto \langle P_p N(F) \mid T_t P_p N(F)\rangle = \int (t,\xi) d \rho_{P_p N(F)}$ is strongly almost periodic as it is the inverse Fourier transform of a pure point measure and $t\mapsto \langle P_c N(F) \mid T_t P_c N(F)\rangle = \int (t,\xi) d \rho_{P_c N(F)}$ is null weakly almost periodic as it is the inverse Fourier transform of a continuous measure.  From Theorem \ref{thm-GA} we conclude that
 $$ \gamma_{sap} (F \ast  \tilde  F_t ) = \langle P_p N(F) \mid T_t P_p N(F )\rangle \;\:\mbox{and}\;\: \gamma_{0wap} ( F \ast  \tilde  F_t) = \langle P_c N(F)\mid T_t P_c N(F)\rangle$$
 for all $t\in \GG$. Setting $t = 0$ we obtain \eqref{toshow}.
\end{proof}

In the situation of the previous theorem we can easily see that
the diffraction to dynamics map $\theta : L^2 (\Gdual, \omega)\longrightarrow \cHH$ can be decomposed as
\begin{equation}
\theta = \theta_p \oplus \theta_c.
\end{equation}
Here,  $\theta_p$ and $\theta_c$  are the diffraction to dynamics maps of $\cN_p$ an $\cN_c$ respectively and  and the natural  decompositions $L^2 (\Gdual, \omega) = L^2 (\Gdual, \omega_p) \oplus L^2 (\Gdual, \omega_c)$  and $\cHH = \cHH_p \oplus \cHH_c$ are implicitly used.

\bigskip

We now restrict attention to spatial processes.  In this situation  we can induce subprocesses via subalgebras of $L^\infty$. This is discussed next. We will use the main result of Gelfand theory that   any commutative $C^\ast$-algebra $\cA$ with a unit is canonically isomorphic to the algebra $C(X_\cA)$ of continuous functions on the compact space $X_\cA$ given by the maximal ideals of $\cA$. We will denote the space of all continuous bounded complex valued functions on the complex plane by $C_b (\CC)$.

\begin{prop} \label{prop-inducing}  Let $\cN:=(N,\cHH,T)$ be a spatial process with Hilbert space  $\cHH=L^2 (X,\mu)$. Let $\cA$ be a $\GG$-invariant subalgebra of $L^\infty (X,\mu)$ which is closed under complex conjugation, contains the constant functions and is closed with respect to the sup norm.  Let $U = U (\cA)$ be the subspace of $\cHH$ generated by $\cA$ and  let $X_\cA$ the maximal ideal space of $\cA$.
 Then, $X_\cA$ can be equipped in a unique way with a $\GG$ action and a $\GG$-invariant measure $\mu_\cA$ such that the canonical Gelfand  isomorphism
 $J:C(X_\cA)\longrightarrow \cA$ extends to an unitary $\GG$-map (also called $J$)
$$J: L^2 (X_\cA,\mu_\cA)\longrightarrow U.$$
The map $J$ is compatible with the algebraic structure in that it satisfies
\begin{itemize}
\item  $J( fg) = J(f) J(g)$ for all $f\in L^\infty (X_\cA,\mu_\cA) $ and $g\in L^2 (X_\cA, \mu_\cA)$ and

\item $J(\phi (f)) = \phi ( J(f))$ for all $f\in L^2$ and for all continuous bounded
$\phi:\CC \longrightarrow \CC $.
\end{itemize}
The subspace  $U$ is a $\GG$-invariant subspace and $\cN$ can be decomposed as
$$\cN = \cN_U \oplus \cN_{U^\perp},$$
where $\cN_U$ is isomorphic via $J$ to the spatial process $\cN_\cA$ with $N_\cA = J^{-1} P_U N$, $\cHH_\cA = L^2 (X_\cA,\mu_\cA)$. If $\cN$ is ergodic, so is $\cN_U$.
\end{prop}
\begin{proof} All statements are  rather straightforward up to the compatibility of $J$ with the algebraic structure and the ergodicity.

\smallskip

 This compatibility with algebraic structure will be discussed next: By definition, $J$ is an algebra isomorphism from $C(X_\cA)$ and $\cA$. In particular, $J( f g) = J(f) J(g)$ for $f,g\in C (X_\cA)$. By a simple approximation argument, we then obtain  $J( fg) = J(f) J(g)$ for all
 $f\in L^\infty(X_\cA)$ and $g\in L^2(X_\cA)$. Here, we use that $ \{f g_n\}$ converges to $\{f g\}$ in $L^2((X_\cA)$  whenever $\{g_n\}$ converges to $g$ in $L^2(X_\cA)$ and $f$ is bounded.  This shows that $J$ is compatible with multiplication.

 As $J$ is an isomorphism of algebras, we also have $J (q (f)) = q( J(f))$ for any polynomial $q$ and any $f\in C (X_\cA)$. By a Stone/Weierstrass type argument, this gives $J (\phi (f)) = \phi ( J(f))$ for any $\phi \in C_b (\CC)$ and $f\in C(X_\cA)$. Consider now the set
 $$\cC :=\{f\in L^2 (X_\cA, \mu_\cA) : \phi (J(f)) = J(\phi (f))\;\mbox{for all $\phi \in C_b (\CC)$}\}.$$
 Then, $ C(X_\cA)  \subset\cC$ by what we have just shown. Moreover,
    $\cC$ is closed in $L^2$ by Proposition \ref{continuity-of-phi}. Thus,
 $\cC$ must agree with $L^2 (X_\cA, \mu_\cA)$ as $C(X_\cA)$ is dense in $L^2 (X_\cA,\mu_\cA)$
(as $\{\phi (h_n)\}$ converges to $\phi (h)$ whenever $\{h_n\}$ converges to $h$).


\smallskip

It remains to show the statement on ergodicity. By definition ergodicity means that the eigenspace of $T$ to the eigenvalues $1$ is one-dimensional. This is obviously stable under the above constructions.
\end{proof}

\begin{definition} The process $\cN_\cA$ constructed in  the previous proposition is called the subprocess of $\cN$ induced by $\cA$.
\end{definition}

\begin{theorem} \label{fullProcess}
Let $\cN:=(N,\cHH,T)$ be a spatial process with Hilbert space  $\cHH=L^2 (X,\mu)$. Let $\cA$ be the closure with respect to the supremum norm of the  subalgebra of $L^\infty (X,\mu)$ generated by $\psi_1\circ N (F_1)\ldots \psi_n \circ N(F_n)$ with $n\in \NN$, $F_j \in C_c (\GG)$ and $\psi_j$ continuous bounded functions from $\CC$ to $\CC$. Let $U = U (\cA)$ be the subspace of $\cHH$ generated by $\cA$. Then $U$ contains the range of $N$ and the subprocess $\cN_U$ is isomorphic to the  full spatial process $\cN_\cA$. The process $\cN_\cA$ is ergodic if $\cN$ is ergodic. Moreover,
$$\langle N_\cA (F) \mid N_{\cA} (G)\rangle = \langle N(F) \mid  N(G)\rangle$$
holds for all $F,G\in C_c (\GG)$. In particular, $\cN$ has a second moment if and only if $\cN_\cA$ has a second moment  and their autocorrelations agree in this case.
\end{theorem}
\begin{proof} By definition of $\cA$, the range of $N$ is obviously contained in the subspace $U$. Thus, $P_U N = N$ and therefore
$$ \langle N_U (F)\mid  N_U (G)\rangle = \langle N (F) \mid  N(G)\rangle$$
holds for all $F,G\in C_c (\GG)$. Given this all statements of the theorem follow immediately from the previous proposition.
\end{proof}

\begin{remark} The theorem says that to any (ergodic) stationary process $\cN$ we can find an (ergodic) stationary subprocess  $\cN_U$ which is full and has the same second moment. In the situation of the theorem fullness of the original process $\cN$ then just means that $\cN = \cN_{U}$.
\end{remark}

Now let  $(N,\cHH,T)$ be a spatial process with $\cHH=L^2 (X,\mu)$. Let
$V \subset L^2(X,\mu)$ be the closure of the linear span of products of the form
$$ \psi_1\circ N_p (F_1)\ldots \psi_n \circ N_p (F_n) ,\; n\in \NN,  \psi_j \in C_b (\CC), F_j\in C_c (\GG).  $$
Then, $V$ is a closed  invariant subspace of $L^2(X,\mu)$.  Let $P_V$ be the orthogonal projection onto $V$. Then, $P_V$ commutes with $T$.  Then,  $(N_V, V, T_V)$ with $N_V =P_V N$ and $T_V = P_V T = T P_V$ is a stationary process. We call $(N_V,V,T_V)$ the {\it augmented  point part} of the stationary process $(N,\cHH,T)$.  Let $V^\perp$ be the orthogonal complement of $V$ in $\cHH$. Then, $V^\perp$ gives rise to a stationary process $(N_{V^\perp},V^\perp, T_{V^\perp})$ and $N$ is the sum of these processes.

\begin{theorem} \label{prop-decomposition}  Let $\cN:=(N,\cHH,T)$ be a spatial process with Hilbert space  $\cHH=L^2 (X,\mu)$ and associated augmented point part $\cN_V:=(N_V, V, T_V)$.  Then the following holds:

\begin{itemize}
\item $\cN = \cN_V \oplus \cN_{V^\perp}$;
\item $\cN_V$ is isomorphic to a full spatial process;
\item $\langle P_V N(F)\mid  P_V N(G) \rangle = \langle P_p N(F)\mid P_p N(G)\rangle$ and $\langle P_{V^\perp} N(F)\mid P_{V^\perp} N(G)\rangle = \langle P_c N(F)\mid P_c N(G)\rangle$  hold for all $F,G\in C_c (\GG)$.
\end{itemize}
If $\cN$ has a second moment with autocorrelation $\gamma$ and diffraction $\omega$  furthermore the following is true:

\begin{itemize}
\item  $\cN_V$ has a second moment with autocorrelation   $\gamma_{sap}$ and diffraction measure  given by the pure point part of $\omega$.
\item   $\cN_{V^\perp}$ has a second moment with autocorrelation $\gamma_{0wap}$  and the diffraction measure given by  the continuous part of $\omega$.
\end{itemize}

\end{theorem}
\begin{proof} Let $\cN_p^{'} = (N_p, \cHH, T)$.
Let the algebra $\cA$ be the closure with respect to the sup-norm of the linear span of products of the form
$$ \psi_1\circ N_p (F_1)\ldots \psi_n \circ N_p (F_n) ,\; n\in \NN,  \psi_j \in C_b (\CC), F_j\in C_c (\GG).  $$
Then, $V$ is just the subspace generated by $\cA$.

We can then apply the previous theorem to $\cN_p^{'}$ and the algebra $\cA$. This gives a decomposition $\cN_p^{'} = \cN_{V}^{'} \oplus \cN_{V^\perp}^{'}$, where $\cN_{V}^{'}$ is isomorphic to a full spatial process and
$$ \langle N_V^{'}  (F), N_V^{'}  (G)\rangle = \langle N_p (F), N_p(G)\rangle = \langle P_p N (F)\mid P_p N (G)\rangle $$
holds for all $F,G\in C_c (\GG)$. Here, the last equality follows by definition of $N_p$ as $P_p N$.

\smallskip

It remains to show that $P_V N (F) = P_p N(F)$ (then all remaining statements follow easily). To do so, we first note that the range of  $N_p = P_p N $ is contained in the range of $U$ by construction. Moreover, by general principles the algebra $\cA$ and hence the subspace $V$  is contained in the pure point part of $T$ \cite{BL,LMS1}. Thus, $P_V N = P_p N$ holds.

Similar remarks apply to $N_c$. Note that $P_{V^\perp} N = N - P_{V}N = N - P_p N = P_c N$.
\end{proof}

The considerations of this section and in particular the preceding theorem  have the following consequence: whenever we are given a spatial process we can  restrict attention to its point part and obtain a full spatial process with pure point spectrum. In particular, the whole theory developed below for  full  spatial processes with pure point spectrm  will apply to the point part of any spatial process.

\section{Spatial processes with pure point diffraction}\label{Purepoint}
In this section we discuss spatial processes with pure point diffraction. It is these processes that we will be able to classify. We show that for these processes the two notions of pure pointedness of the spectrum defined above viz pure point diffraction spectrum  and pure point dynamical spectrum actually agree.  As noted in Theorem \ref{prop-decomposition}, any spatial process with second moment can be  decomposed into a full  spatial process with pure point diffraction and a general stationary process with continuous diffrction. Thus, the results of this section apply to the corresponding parts of any spatial process.

\bigskip

Let $(N,X,\mu,T)$ be a stationary process with pure point diffraction as discussed in Section \ref{Stochasticprocesses}. Thus, its diffraction measure $\difm$ is a pure point measure and the  set of its atoms denoted by $\Bragg=\Bragg (\difm)$ is given by
$$\Bragg =\{k\in \Gdual: \difm (k) >0\}.$$
Let $$\theta: L^2 (\Gdual,\difm)\longrightarrow L^2 (X,m)$$
be the associated diffraction to dynamics map and $f_k := \theta (1_k)$, $k\in \Bragg$.

While the subspace $\theta (L^2 (\Gdual, \difm)$ of $L^2 (X,\mu)$ may be  small in some sense, in another sense it controls the whole Hilbert space $L^2 (X,\mu)$ due to our assumption of pure point spectrum. A precise version is given next.

\begin{theorem} \label{ergodic+pure point}
Let $\cN = (N,X,\mu,T)$ be a  spatial stationary process, which is full and possesses a second moment. Then, the following assertions are equivalent:
\begin{itemize}
\item[(i)] The diffraction measure $\difm$ of $(N,X,\mu,T)$ is a pure point measure.
\item[(ii)] The representation $T$ has pure point spectrum.
\end{itemize}
In this case, the group $\ev$ of eigenvalues of $T$ is generated by the set
$\Bragg(\difm)$,  and any eigenfunction is a (multiple of a) product of eigenfunctions of the form $f_k$, $k\in \Bragg(\difm)$.
\end{theorem}

\begin{proof} We mimick the argument of \cite{LS} (see \cite{BL, LMS1} as well). By definition of the spectral measures $\rho_f$  given in \eqref{spectralmeasure} we have
$$ \langle N(F) \mid T_t N(F)\rangle = \int_{\Gdual}  (\gamma,t) d\rho_{N(F)} (\gamma)$$
for all $F\in C_c (\GG)$. Moreover, by definition  of $\difm$ and $\gamma$ and the fact that $N$ is a $\GG$-mapping  we find
$$\langle N(F) \mid T_t N(F)\rangle = \gamma (F* \widetilde{F(T_{-t} \cdot)}) = \int_{\Gdual} (\gamma,t) |\widehat{F}|^2 d\omega (\gamma)$$
first for real-valued  $F\in C_c (\GG)$ and then by linearity for all $F\in C_c (\GG)$.
Taking inverse Fourier transforms we infer
\begin{equation}\label{ift}
\rho_{N(F)} = |\widehat{F}|^2 d\difm\;\:
  \end{equation}
for all $F\in C_c (\GG)$.

\smallskip

 The implication $(ii)\Longrightarrow (i)$ is now clear from \eqref{ift}. (We also saw this
 earlier from the properties of $\theta$.)  We next show $(i)\Longrightarrow (ii)$: As the diffraction measure is a pure point measure, the spectral measures associated to $N(F)$ are pure point measures (supported on $\Bragg$) by \eqref{ift}. Thus, by part (a) of Corollary \ref{coro-stabilitypp}, the spectral measures of $\psi \circ N (F)$ for $\psi \in C_c (\CC)$, $F\in C_c (\GG)$, are then also pure point measures supported on the group generated by $\Bragg$. By the fullness assumption and part (b) of Corollary \ref{coro-stabilitypp}  we then infer that $T$ has pure point spectrum with eigenvalues contained in the group generated by $\Bragg$.

Now, multiplying eigenfunctions $f_k$, $k\in \Bragg$, we obtain eigenfunctions for arbitrary elements in the group generated by $\Bragg$. This shows that the group generated by $\Bragg$ is indeed the group of eigenvalues. At the same time it gives an eigenfunction for each eigenvalue. By ergodicity the multiplicity of each eigenvalue is one and these eigenfunctions are all eigenfunctions.
\end{proof}

\begin{remark}  (a) There is quite some history to Theorem \ref{ergodic+pure point}: The implication $(ii)\Longrightarrow (i)$ is sometimes known as Dworkin argument. It  was first  established  by Dworkin \cite{Dworkin} with later extensions   by Hof \cite{Hof1} and Schlottmann \cite{Martin} (see \cite{Rob2,Sol2} for remarkable applications as well). In special situations equivalence was then shown by Lee/Moody/Solomyak in \cite{LMS1}. This was then generalized in various directions by Gou\'{e}r\'{e} \cite{Gouere},  Baake/Lenz \cite{BL}, Deng/Moody \cite{XR}, and Lenz/Strungaru \cite{LS} (the latter result containing all earlier ones).  The above result contains all these  results (provided the underlying process is real).

(b) The equivalence between diffraction and dynamical spectrum cannot hold for the complete spectrum, as discussed  by van Enter/Mieskicz \cite{EM}. Indeed,  \cite{EM} gives an example of    mixed spectrum where  the pure point part of the dynamical spectrum is completely missing in the diffraction (up to the constant eigenfunction).
\end{remark}

Due to the previous theorem we do not need to distinguish between pure point diffraction and pure point dynamical spectrum when dealing with full processes with second moments. This suggests the following definition.

\begin{definition} \label{def-pure-point-process} A spatial stationary  process $ (N,X,\mu,T)$ is called  \textit{pure point}  if it is full, possesses a second moment and its diffraction measure is a pure point measure.
\end{definition}

If $0 \in \Bragg$ in the situation of Prop.~\ref{eigenExpansionOfN(F)}
then by Lemma~\ref{fminusk} $f_0$ is a real function belonging to the
$k=0$-eigenspace of $L^2(X,m)$ and $||f_0||= \difm(0)^{1/2}$.
Since constant function ${\it 1}_X: \xi \mapsto 1$
for all $\xi \in X$ is an eigenfunction for $k=0$ and the dynamical system is ergodic,
it spans the $k=0$-eigenspace of $L^2(X,m)$. Thus  $f_0 = \pm \difm(0)^{1/2}{\it 1}_X$.
Now, we have already pointed out that if $N$ is a stationary process then so is
any non-zero real multiple of it. In particular $\pm N$ are stationary processes
and one of them has the corresponding function $f_0 = \difm(0)^{1/2}{\it 1}_X$. It causes
unnecessary awkwardness later on in the discussion of phase forms to deal with the case $f_0 = -\difm(0)^{1/2}{\it 1}_X$, so wish
to normalize $N$ to avoid this situation:

\bigskip

\begin{assumption}
\label{0-eigenfunction assumption}
\begin{center}
\fbox{\parbox[c]{115mm}
{\centerline{In the pure point ergodic case
 we shall always assume that}
 \centerline{$f_0 = \difm(0)^{1/2}{\it 1}_X$.}}}
\end{center}
  \end{assumption}

\medskip

\section{Relators and associated phase forms}\label{Cycles}
As mentioned already, our aim is to describe all stationary processes with a given pure point diffraction measure $\omega$. This description will be given in terms  of  a set of objects called {\it relators}  arising out of the Bragg spectrum of $\omega$. While the later sections need the material of this section, it can be read independently of the previous ones.

\bigskip

Let $\Gdual$ be a locally compact abelian group.
Let $\Bragg \subset \Gdual$ be given with $\Bragg = -\Bragg$ and
let $\ev$ be the subgroup of $\Gdual$ generated by $\Bragg$.
We shall give this group the discrete topology, and in order to keep this
clear we shall usually write $\ev_d$ instead of $\ev$.

In the sequel we often use $m$-tuples $(k_1, \dots, k_m) \in \Bragg^m$ (for various $m$) of elements of $\Bragg$ or of $\ev_d$ and we often simply denote these by bold letters ${\bf k}$. If ${\bf k} = (k_1, \dots, k_m), {\it l} = (l_1, \dots, l_n)$ then
 ${\bf k}{\bf l}$ is the concatenation
 $$ {\bf{ k l}} = (k_1, \dots, k_m,l_1, \dots, l_n)$$
and
$$[{\bf k}] := k_1 + \dots + k_m \in \ev_d.$$
 In this way $\bbS:= \bigcup_{m=0}^\infty \Bragg^m$ becomes
a monoid under concatenation, with the empty $0$-tuple $\emptyset$ as the identity element.


A {\it relator} is any $m$-tuple ${\bf k} = (k_1, \dots, k_m) \in \Bragg^m$ with
$$[{\bf k}]= k_1 + \dots + k_m =0.$$
 The empty tuple $\empty$ is also taken to be a relator. Thus the relators are a
subset of $\bbS$.
Let
\begin{eqnarray*}
Z_n & := & \{{\bf k} = (k_1, \dots, k_n) \in \Bragg^n\,:\, [{\bk}] =0 \} \,, n\ge 1 \,;\\
Z_0 &:= & \{\emptyset\}\,; \\
Z& := & \bigcup_{n=0}^\infty Z_n \, .
\end{eqnarray*}
Note that $Z_1= \{(0)\}$ if $0\in \Bragg$, otherwise it is empty.

Concatenations of relators are relators and this makes $Z$ a submonoid of $\bbS$.

We introduce an equivalence relation on $\bbS$ by transitive extension
of the three rules:
\begin{itemize}
\item[(R1)] ${\bf k} \sim {\bf l}$ if ${\bf l}$ is a permutation of the symbols of ${\bf k}$;
\item[(R2)] if ${\bf k}= (k_1, \dots, k_m) \in \bbS$ and $0\in \Bragg$ then $(k_1, \dots, k_m, 0) \sim (k_1, \dots, k_m)$ (along with the first rule, this means that when $0\in \Bragg$, zeros can be dropped or added wherever they appear);
\item[(R3)]  ${\bf k} \sim {\bf l}$ if ${\bf l}$ can be obtained from ${\bf k}$ by inserting removing pairs $\{k,-k\}$, $k\in \Bragg$.
\end{itemize}

The second item here is related to our Assumption ~\ref{0-eigenfunction assumption},
see below.

It is easy to see that ${\bf k} \sim {\bf l}\,,\, {\bf k'} \sim {\bf l'} \Rightarrow
{\bf k}{\bf l} \sim {\bf k'}{\bf l'}$, so multiplication descends from $\bbS$ to $\bbS^\sim:= \bbS/\sim$,
whereupon $\bbS^\sim$ becomes an Abelian group under this multiplication with
$\emptyset^\sim$ as the identity element. Note that if $0 \in \Bragg$ then
$0^\sim = \emptyset^\sim$. We give $\bbS^\sim$ the discrete topology, so in particular it is a locally compact Abelian group.

Since the sum $[{\bf k}] = k_1 +\dots + k_m$ of an element of $\Bragg^m$ is constant on the entire
equivalence class ${\bf k}^\sim$, we see that $Z$ is the union of all equivalence classes
with component sum equal to $0$, and we obtain $\cZ := Z/\sim$ as a subgroup of $\bbS^\sim$.
In addition, we see that there is a surjective homomorphism
$$\phi: \bbS^\sim \longrightarrow \ev_d,\;\mbox{with}\:\; {\bf k} \mapsto [{\bf k}]$$
 and its kernel is precisely $\cZ$.

\begin{definition} The group $\cZ= \cZ(\Bragg)$ is called the {\it relator group}. The elements of its dual group
$\widehat{\cZ}$ are called {\it phase forms}. Thus a
 phase form is a group homomorphism
\[ a^*: \loops \longrightarrow U(1)\]
(the latter being the unit circle in $\CC$).
\end{definition}

\begin{remark} (a)  It is useful to think of relators in terms of the cycle structure of the Cayley graph of the group of eigenvalues with respect to the generators $\cS$. In this way the group $\cZ$ could then be seen as a a kind of abelianized homotopy group or  a first cohomology group of the graph.  This works in the following way:
We begin with $\ev$ and the set  $\Bragg$ which satisfies $\Bragg = -\Bragg$
and generates $\ev$. The Cayley graph has vertices $\ev$ and edges consisting of all pairs
$(x, x+k)$ where $x\in \ev$ and $k \in \Bragg$. Since for each edge $(x, x+k)$ we have the
reverse edge $(x+k, x+k-k=x)$, we can think of the graph as being non-directed. Since $\Bragg$
generates $\ev$, the graph is connected. Each
${\bf k} = (k_1, \dots, k_m)$ can be thought of as an edge sequence. Given any $x\in \ev$
we obtain a path $x,x+k_1,x+k_1+k_2, \dots, x+k_1+ \cdots + k_m$ from it. This path is a cycle  if
and only if
${\bf k}$ is a relator, that is, $[{\bf k}]=0$.

There is a sort of homotopy of edge sequences, which is generated by rules corresponding
to (R1), (R2), (R3) above. The first amounts to saying that for all $x\in \ev$ and for
all $k,l\in\Bragg$, the paths $x,x+k,x+k+l$
and $x,x+l, x+l+k$ are homotopic; the second says that if $0\in \Bragg$ then
the path $x, x+0$ is homotopic with the empty path from $x$; and the third
says that the path $x,x+k, x+k+(-k)$ is homotopic to the empty path from $x$. The
group $\cZ$ may be thought of as a sort of homotopy group for cycles originating
(and terminating) at $0$.

(b) It should be noted that the homotopy group $\cZ$ may be infinite even when the
group $\GG$, and hence the associated graph, is finite. We shall see this in the examples
below.
\end{remark}

The phase forms play a central role in homometry problem. We note that they depend
only on $\Bragg$, not on the actual values of the diffraction measure $\difm$. We shall find it convenient
to sometimes abuse the notation and write things like $a^*(k_1,\dots, k_m)$ where
we mean $a^*((k_1,\dots, k_m)^\sim)$ when using the phase form $a^*$.

Of the various ${\bf k} = (k_1, \dots, k_m )$ that can represent a given element in
$\cZ$ there is (at least) one of minimal length $n$. This minimal
length is denoted by $\len({\bf k})$, and is called the {\it reduced length} of ${\bf k}$. Here, of course, the length of $(k_1,\ldots, k_m)$ is given by $m$.

Define
$\cZ_n$ to be the set of elements which have reduced length less than or equal to $n$.
We have
\begin{eqnarray*}
\cZ &=& \bigcup_{n=0}^\infty \cZ_n,\\
 \cZ_n \cZ_p &\subset& \cZ_{n+p} \quad\mbox{ for all non-negative integers n and p.}
\end{eqnarray*}

We let $\cF(\Bragg)$ be the Abelian group (with discrete topology) generated by symbols
$e(k)$, $k\in \Bragg$, subject only to the relations\footnote{The second of these relations comes from the underlying assumption that we are dealing with {\em real}  stationary processes. If this condition were dropped then these relations would also be dropped.} $e(0) =1$ (if $0\in \Bragg$) and $e(-k)e(k)=1$ for all $k\in \Bragg$.

\begin{prop}\label{F(S) homomorphism} The mapping
$\varphi: e(k) \mapsto (k)^\sim$ for all $k\in\Bragg$ extends to an isomorphism
$\varphi: \cF(\Bragg) \longrightarrow \bbS^\sim$.
\end{prop}

 \begin{proof} The existence of $\varphi$ as a surjective homomorphism is clear since $\bbS^\sim$ is the Abelian
 group freely generated by the symbols $(k)^\sim$ subject only to the relations defining
 $\cF(\Bragg)$. These correspond precisely to  $(k)^\sim (-k)^\sim = (k,-k)^\sim = \emptyset^\sim$
 and $(0)^\sim = \emptyset^\sim$ when $0\in \Bragg$ which hold in $\bbS^\sim$.

 A typical element $e(k_1) \dots e(k_m)$ of $\cF(\Bragg)$ is in the kernel of $\varphi$ if and
 only if $(k_1)^\sim \dots (k_m)^\sim = (k_1, \dots, k_m)^\sim  = \emptyset$. This happens only if the non-zero components
 of $(k_1, \dots, k_m)$ cancel in pairs $k,-k$. Then also the corresponding terms
 in $e(k_1) \dots e(k_m)$ cancel in pairs. Since also $e(0) =1$ in $\cF(\Bragg)$
 when $0 \in \Bragg$, we see that $e(k_1) \dots e(k_m) $ reduces to the identity element.
\end{proof}

We find it convenient to identify $\cF(\Bragg)$ and $\bbS^\sim$ through Prop.~\ref{F(S) homomorphism}. Having done this we obtain from $\phi:\bbS^\sim \longrightarrow \ev_d$, defined above, the homomorphism
$$\varphi: \cF(\Bragg) \longrightarrow \ev_d,\:\;\mbox{with}\;\: e(k) \mapsto k \,,$$
whose kernel can be identified with $\cZ$.

Thus, we have the following
 exact sequence of groups (all given the discrete topology)
\begin{equation}\label{exactSeq1}
1 \longrightarrow \cZ  \longrightarrow \cF(\Bragg)
\longrightarrow \ev_d \longrightarrow 1.
\end{equation}
Dualization then gives the exact sequence
\begin{equation}\label{exactSeq2}
1 \longleftarrow \widehat{\cZ}   \longleftarrow \widehat{\cF(\Bragg) }
\longleftarrow \TT \longleftarrow 1,
\end{equation}
where $\TT := \widehat\ev_d $ is a compact Abelian group
by  Pontryagin duality.  Note that we get surjectivity from the continuity of
the mappings and the
compactness of these groups. In fact, surjectivity of the map $\widehat{\cF(\Bragg)} \longrightarrow \widehat{\cZ}$ is crucial for our subsequent considerations.

\begin{definition}
The mapping $a : \Bragg \longrightarrow U(1)$
satisfies the $m$-{\it moment condition} if
\begin{equation} \label{mMoment}
a(k_1)a(k_2) \dots a(k_{m}) =1
\end{equation}
whenever $k_1, \dots, k_m \in \Bragg$ and $k_1 + \cdots + k_m = 0$.
\end{definition}

We shall see the reason for calling these moment conditions in
\S\ref{Special}. Notice that the first moment condition is empty (and so trivially holds!) unless $0 \in \Bragg$,
in which case it says that $a(0) = 1$. The second moment condition is equivalent to the statement that $a(-k) = \overline{a(k)}$ for all $k \in \Bragg$. If $0\in \Bragg$
then the second moment condition alone gives $a(0)a(0)=1$, so $a(0) = \pm 1$. See
Remark~\ref{0-eigenfunction assumption}. If the first moment condition is not empty and holds then $a(0)=1$ and then the $m$-moment condition includes
the $n$-moment condition for all $n<m$ since additional slots can be filled with
zeros.

Obviously, any mapping $a: \Bragg \longrightarrow U(1)$ satisfying the first and second moment conditions determines a character $\cF(\Bragg)$, i.e. an element
of $\widehat{\cF(\Bragg)}$, via $e(k) \mapsto a(k)$ for all
$k\in \Bragg$. Conversely every character of $\cF(\Bragg)$ clearly determines
a mapping $a$ satisfying the first and second moment conditions. (Notice that
when $0\in\Bragg$ then $e(0)$ is the identity element of $\cF(\Bragg)$ and also
$a(0) = 1$.)
In the
sequel we will use these two concepts interchangeably and use
the symbol $a$ both as a mapping on $\Bragg$ and as the corresponding
character on $\cF(\Bragg)$.
When thought of as elements of $\widehat{\cF(\Bragg)}$ they
are called {\it elementary phase forms}.

By restriction any elementary phase form $a$ determines a phase form $a^*$, that is, an element
of $\widehat{\cZ}$. The exact sequence \eqref{exactSeq2} shows that every phase form arises
by restriction from an elementary phase form and two elementary phase forms restrict
to the same phase form if and only their ratio
is in $\TT$, that is to say, if and only if their ratio is a character on
$\ev_d$.
This is the essential part of determining equivalence of pure point processes.

\begin{prop} \label{characterCondition}
\begin{itemize}
\item[{(i)}] Any mapping
$a: \Bragg \longrightarrow U(1)$ satisfying the first and second moment conditions
determines a unique elementary phase form and every elementary phase form
arises in this way.
\item[{(ii)}]
Any phase form is the restriction of an elementary phase form. The ratio of any two elementary phase forms giving the same phase form is an element of \,$\TT$.
\item[{(iii)}] An elementary phase form satisfies the $m$-moment conditions
for $m= 1,\dots,n$ if and only it kills $\cZ_n$.
\item[{(iv)}] An elementary phase form $a$ determines a character on $\ev_d$,
i.e. lies in $\TT$,
if and only if it satisfies all $m$-moment conditions, $m =1,2, \dots$.
\end{itemize}
\end{prop}
\begin{proof} Parts (i) and (ii) are already done above.

$\cZ_n$ is generated by all the ${\bf k} = (k_1, \dots, k_m)$
with $k_1 + \cdots + k_m =0$, $m \le n$. An elementary phase form $a$
kills the equivalence class of $\it k$ in $\cF(S)$ iff
\[a(k_1) \dots a(k_m) =1 \,. \]
This proves (iii). An elementary phase form $a$ kills all of $\cZ$ iff $a^*$ is trivial
which happens iff
$a \in \TT = \widehat{\ev_d}$, which gives (iv).
\end{proof}

Of particular interest to us is the case where the group $\cZ$ is generated by $\cZ_n$ for some $n\geq 2$. This can be understood on the level of phase forms in the following way.

\begin{lemma}\label{charChar} Let $n\ge2$. Then the following assertions are equivalent:
\begin{itemize}
\item[(i)]  $\langle \cZ_n \rangle = \cZ$.

\item[(ii)]  Any  mapping $a:S \longrightarrow U(1)$ satisfying  the $m$-moment conditions for $1\le m \le n$ extends to a character on $\ev_d$ (i.e. lies in $\TT$).
\end{itemize}

\end{lemma}
\begin{proof} (ii)$\Longrightarrow$ (i):  The fact that $a:S \longrightarrow U(1)$ satisfies the $m$-moment conditions for $1\le m \le n$ is equivalent to saying that the corresponding character $a$
on $\cF(S)$ kills  $\cZ_n$, or equivalently, kills the subgroup $\langle \cZ_n \rangle$
that it generates. But Prop.~\ref{characterCondition} says that $a$ extends to a character
on $\ev_d$ iff $a$ kills $\cZ$. Suppose that  fact that the mapping $a:S \longrightarrow U(1)$ satisfies the $m$-moment conditions for $1\le m \le n$ is enough to sufficient to guarantee that whenever $a$ kills $\langle \cZ_n \rangle$ it must also
kill $\cZ$. Then  $\langle \cZ_n \rangle=\cZ$
because $\langle \cZ_n \rangle \subset \cZ$, both are closed subgroups
of $\cF(\Bragg)$, and the Pontryagin duality theory says that the characters can distinguish
distinct closed subgroups.

The implication (i)$\Longrightarrow$ (ii) is clear.
\end{proof}

\medskip

The preceding considerations suggest to consider
$$n_0:=n_0 (\cZ):= \inf\{n\in \NN: \mbox{ $\cZ$ generated by $\cZ_n$}\}.$$
Here, the infimum of the empty set is defined  to be  $\infty$.

\begin{remark} \label{2m+1 theorem} If $\ev = \Bragg + \dots + \Bragg$ (with $r$ summands) then any relator can be written
as a product of relators of length at most $2r+1$ so $\cZ_{2r+1}$ generates $\cZ$. Thus, in this case $n_0 \leq 2 r + 1$ \cite{LM}.
\end{remark}


\smallskip

\begin{definition} The restriction of a  phase form or an elementary phase form  to $Z_m$ is called its $m$th-moment.
\end{definition}

\begin{prop} \label{mthMomentCharacterization}
Two elementary phase forms $a$ and $b$ give rise to the same $m$th moments if and only if their ratio $u =b/a$ satisfies the $m$th moment condition.
\end{prop}
\begin{proof} Elementary phase forms $a,b$ give rise to equal $m$th moments if and only if
$a(k_1, \dots, k_m) = b(k_1, \dots, k_m)$ whenever $k_1, \dots, k_m \in \Bragg$ with
$k_1 + \cdots + k_m =0$. However, being elementary phase forms, these equations can
be written as $a(k_1) \dots a(k_m) = b(k_1) \dots b(k_m)$, which leads to the equivalent statement
that
\[u(k_1) \dots u(k_m) =1 \]
whenever $k_1, \dots, k_m \in \Bragg$ with
$k_1 + \cdots + k_m =0$. This is the $m$th moment condition.
\end{proof}

Two elementary phase forms $a,b$ give rise to the same phase form $a^*$ if and only if their ratio is
an element of $\TT$ and so satisfies all the moment conditions. By Prop.~\ref{mthMomentCharacterization} this means that all their moments are the same, so that there is no ambiguity
in speaking of the $m$th {\it moments of phase forms}.

\begin{lemma}\label{momentsCondition} Assume  $\langle \cZ_n \rangle = \cZ$ for some $n\geq 2$. Then two phase forms agree if and only if they have the same $m$th moments for $m=1,\ldots, n$. In particular, two elementary phase forms restrict to the same phase form if and only if they have the same  $m$th moments for $m=1,\dots, n$.
\end{lemma}
\begin{proof}  The first statement is  clear. The second statement is a direct consequence of the first statement.
\end{proof}

\begin{remark} Any of the phase forms $a$ that we are interested in from the point of view
of diffraction satisfy the first and second moment conditions by assumption.
 The relevance of the previous lemma  is the following. If two stationary
 processes have  the same pure point diffraction and their associated group of relators satisfies  $\langle \cZ_n \rangle = \cZ$, then the processes
are isomorphic if they have the same $m$th moments for $m=1,\dots, n$ (see Cororllary~\ref{momentsCondition1} and Corollary~\ref{momentsCondition2} for precise statements).
\end{remark}

\section{From spatial stationary processes with pure point diffraction to phase factors }\label{StochasticprocessestoDiffraction}
In this section we assume that we are given an ergodic spatial stationary process with pure point diffraction.
We will discuss how this  process gives rise to a diffraction measure $\difm$, a set $\Bragg$ which generates  a group $\ev$,  canonical eigenfunctions $f_k$, $k\in \Bragg$, and
 a  phase form  $a^*$.
In the next section we will see that $(\omega,a^*)$ in some sense uniquely determines the process.

\bigskip

Let $(N,X,\mu,T)$ be an ergodic stationary spatial process with pure point diffraction. From the considerations of Section \ref{Stochasticprocesses} we then obtain the following:
The diffraction measure $\difm$  of $(N, X,\mu,T)$ is a pure point measure and the set of its atoms, which are called Bragg peaks,  is denoted by $\Bragg=\Bragg (\difm)$
i.e.
$$\Bragg = \Bragg(\difm) :=\{k\in \Gdual: \difm (k) >0\}$$
and satisfies $\Bragg = - \Bragg $.
We let $\ev$ be the subgroup of $\Gdual$ generated by $\Bragg$ and
let $\ev_d$ denote this group when it is given the discrete topology.
We are now clearly in the situation of the previous section. In particular, there is an associated  relator group $\cZ = \cZ (\difm)$ at our disposal.  Let $$\theta: L^2 (\Gdual,\difm)\longrightarrow L^2 (X,m)$$
be the associated diffraction to dynamics map and for  each $k \in \Gdual$, let
$1_k$ be the characteristic function on the set $\{k\}$. Let $U$ be the action
of $\GG$ acting on $L^2(\Gdual, \difm)$ defined by
$$U_t h (k) := (k,t) h (k)$$
and observe that whenever $\difm(k) \ne 0$, $1_k$
is a $k$ eigenfunction for this action,. It is easy to see that up to scaling these are the only eigenfunctions of $L^2(\Gdual, \difm)$.
By the orthogonality
of eigenfunctions we have,
$$\langle 1_k \mid 1_{k'} \rangle_{\Gdual} = \difm(k) \delta_{k,k'}$$
and $$\langle F \mid  G \rangle_{\Gdual} = \sum_{k\in \Bragg} F(k) G(k) \difm(k) $$
for all $F,G \in L^2(\Gdual,\difm)$.

To each element of $  \Bragg$ there exists a corresponding canonical eigenfunction. More specifically, we
define $f_k$, $k\in \Bragg$, by $f_k = \theta(1_k)$. Then each $f_k$ is an eigenfunction in
$L^2(X,\mu)$ for the eigenvalue $k$ and $||f_k||_2 = ||1_k||_{\Gdual}
= \difm(k)^{1/2}$. Since $|f_k|$ is a constant function (due to ergodicity)
and $\mu$ is a probability measure, $|f_k| = \difm(k)^{1/2}$.
Then, for all $k_1, \dots, k_m \in \Bragg$,
\[||f_{k_1} \dots f_{k_m}||^2 = \int_X f_{k_1} \dots f_{k_m}\overline{f_{k_1} \dots f_{k_m}}
\dd\mu = \int_X |f_{k_1}|^2  \dots |f_{k_m}|^2 \dd\mu = \difm(k_1) \dots \difm(k_m) \,.\]

At this point we invoke Assumption ~\ref{0-eigenfunction assumption}
which says that if $0\in\Bragg$ then $f_0 = \difm(0)^{1/2} {\it 1}_X$.

\smallskip

For  ${\bf k} = (k_1, \dots, k_m)$ with $k_j \in \Bragg$
and $[{\bf k}]=0$  the product  $f_{k_1}\ldots f_{k_m}$ is an eigenfunction to $0$ and hence (by ergodicity) a multiple of ${\it 1}_X$, which yields
\begin{equation}\label{phaseCreated}
f_{k_1} \dots f_{k_m} = a^*(k_1,\dots, k_m)\difm(k_1)^{1/2}\dots \difm(k_m)^{1/2}
\end{equation}
for some $a^*({\bf k}) = a^*(k_1,\dots, k_m) \in U(1)$.
In this way we have a mapping $a^*: Z = Z(\Bragg)\longrightarrow U(1)$.
It clearly respects the multiplication by concatenation on $Z$.
In view of our assumption on $f_0$ and
Lemma~\ref{fminusk} we have $f_{-k} = \overline{f_k}$, which shows this
$a^*$ is well-defined on the equivalence classes of $Z = Z(\Bragg)$
defined in \S\ref{Cycles}, and in this way we obtain a phase form
\begin{equation}\label{phaseFormCreated}
a^* :\cZ \longrightarrow U(1) \,.
\end{equation}

The preceding considerations show that any pure point stationary spatial process comes with a natural measure $\omega$ and a phase form $a^*$ as summarized in the following proposition.

\begin{prop}\label{fromsptopf} Each ergodic pure point stationary spatial process $(N,X,\mu,T)$
gives rise to a pair   $(\difm, a^*)$ consisting of a pure point measure $\difm$
on $\Gdual$  characterized by
$$ \int_{\Gdual} \widehat{F} \overline{\widehat{G}} d\difm = \langle N(F)\mid N(G)\rangle$$
for all $F,G\in C_c (\GG)$  and the phase form $a^*$ on the relator group $\cZ=\cZ(\difm)$ satisfying
$$ f_{k_1} \dots f_{k_m} = a^*(k_1,\dots, k_m)\difm(k_1)^{1/2}\dots \difm(k_m)^{1/2}.
$$
\end{prop}

Next we show that  $\difm$ and $a^\ast$ are   a complete set of invariants for the underlying process.

\section{Isomorphism of pure point processes}\label{Uniqueness}
Recall that we have introduced a notion of isomorphism between spatial  processes in \S\ref{Stochasticprocesses}.  In this section we show that two pure point ergodic  spatial stationary processes  are  isomorphic  (in the sense of Definition \ref{definition-isomorphism}) if they yield the same  diffraction measure $\omega$ and the same phase form $a^*$.

\bigskip

We remind the reader that our concept of a pure point spatial stationary process given in Definition \ref{def-pure-point-process} entails that the process in question is full and possesses a diffraction measure (which is pure point).

\begin{theorem}\label{uniquenesstheorem} Two pure point  ergodic stationary spatial  processes
  based on the same group $\GG$
 are isomorphic if and only if their
associated diffraction measures and phase forms are the same.
\end{theorem}
\begin{proof}
Assume that the full processes  $(N,X,\mu,T)$ and $(N',X',\mu',T')$ are pure point ergodic  and
have the same diffraction measure $\difm$ and the same phase form $a^*$.
Set $\Bragg := \Bragg (\omega) = \Bragg (\omega')$. Let $f_k\in L^2 (X,\mu)$ and $f_k'\in L^2 (X',\mu')$, $k\in \Bragg$,  be the corresponding  natural eigenfunctions.  Thus,  $N (F) = \sum_{k\in \Bragg} \widehat{F} (k) f_k$ and $N' (F) = \sum_{k\in \Bragg} \widehat{F} (k) f_k'$
for all $F\in C_c(\GG)$. Moreover, by Theorem \ref{ergodic+pure point} products of the $f_k$ and the $f_k'$ respectively provide orthonormal bases consisting of eigenfunctions  of  $L^2 (X,\mu)$ and $L^2 (X',\mu')$ respectively.

\smallskip

We construct an isomorphism $M$ as follows:
We define $M$ to be the unique linear map with
$$f_{k_1}\ldots f_{k_m}\mapsto f_{k_1}'\ldots f_{k_m}'.$$
This mapping is well-defined, for suppose
that $f_{k_1}\ldots f_{k_m}$ and $f_{l_1}\ldots f_{l_n}$ are linearly
dependent. Then $k_1 + \cdots + k_m = l_1 + \cdots + l_n$,
so $k_1 +\cdots + k_m -  l_1 - \cdots -l_n =0$ and
\[f_{k_1}\ldots f_{k_m}f _{-l_1}\ldots f_{-l_n}= \difm(k_1)^{1/2}\ldots \difm(k_m)\difm(l_1)^{1/2}\ldots \difm(l_n)^{1/2} a^*(k_1, \dots, k_m,-l_1,\dots, -l_n) \,. \]
Using Lemma~\ref{fminusk} we obtain
\[f_{k_1}\ldots f_{k_m} = \difm(k_1)^{1/2}\ldots \difm(k_m)^{1/2}\difm(l_1)^{-1/2}\ldots \difm(l_n)^{-1/2} a^*(k_1, \dots, k_m,-l_1,\dots, -l_n)f _{l_1}\ldots f_{l_n}  \,. \]
By our assumptions we obtain the same relation when the $f$s are replaced
with $f'$s, so the same linear dependence occurs in $L^2(X',\mu')$.

As products of the $f_k$ and $f_k'$ respectively give an orthogonal basis of $L^2 (X,\mu)$ and $L^2 (X',\mu')$, and as
\[||f_{k_1} \dots f_{k_m}||^2 = \difm(k_1) \dots \difm(k_m) = ||f'_{k_1} \dots f'_{k_m}||^2\,.\]
our mapping is a unitary map. In particular it is continuous.

The definition of $M$ easily shows that it intertwines $T$ and $T'$
and maps $N(F)$ to $N' (F)$ for all $F\in C_c (\GG)$.  Moreover,
 by definition $M$ satisfies
$$ M(f_{k_1} \ldots f_{k_n} f_{l_1}\ldots f_{l_m}) = M ( f_{k_1} \ldots f_{k_n}) M ( f_{l_1}\ldots f_{l_m})$$
for all $k_1,\ldots, k_n\in \Bragg$ and $l_1,\ldots, l_m\in \Bragg$.  Note that both $f_{k_1} \ldots f_{k_n} $ and $M ( f_{k_1} \ldots f_{k_n}) =  f_{k_1}' \ldots f_{k_n}' $ are bounded functions. Thus, we can take
linear combinations (in $L^2$) and their  limits to obtain
$$ M(f_{k_1} \ldots f_{k_n} g) = M ( f_{k_1} \ldots f_{k_n}) M (g)$$
for all $g\in L^2 (X,\mu)$ and $k_1\ldots k_n\in \Bragg$. In particular,
$$ M(f_{k_1} \ldots f_{k_n} g) = M ( f_{k_1} \ldots f_{k_n}) M (g)$$
holds for all $g\in L^\infty (X,\mu)$ and $k_1\ldots k_n\in \Bragg$.  Another approximation argument now yields $M(f g) = M(f) M(g)$ for all $f,g\in L^\infty (X,\mu)$. This proves the isomorphism.

\smallskip

Conversely suppose that $M$ is an isomorphism between the pure point ergodic stationary
spatial processes $(N,X,\mu,T)$ and $(N',X',\mu',T')$ on $\GG$.
As $M$ is a unitary mapping, one sees from directly from the definitions that their correlation measures,
and hence their diffraction measures are the same (if  second moments exist), and then that the diffraction
to dynamics mappings are related by $\theta' = M\circ \theta$.  Then looking at the
eigenfunctions we have
$M(f_k) = M(\theta(1_k) = \theta'(1_k) = f'_k$. Since the $f_k$ are bounded functions,
the multiplicative property of $M$ gives
$M ( f_{k_1} \ldots f_{k_n}) =  f_{k_1}' \ldots f_{k_n}' $
for all $k_1, \dots, k_m \in \Bragg$. Then the definition of the phase form
\eqref{phaseCreated} and \eqref{phaseFormCreated} shows
that both processes have the same phase form.
 \end{proof}

\begin{coro} \label{momentsCondition1}
Let $\omega$ be a positive symmetric pure point measure on $\Gdual$ and $\cZ$ the associated group of relators.
If $\langle \cZ_n\rangle = \cZ$, then   two stationary spatial processes with diffraction $\omega$ are spatially  isomorphic if and only if the first $n$ moments of their phase forms agree.
\end{coro}
\begin{proof} This is a direct consequence of the previous theorem and Lemma \ref{momentsCondition}.
\end{proof}

\begin{remark} \label{automorphism}
(a) The above   theorem deals with isomorphism between two processes. We can also characterize the  automorphisms of a given pure point stationary spatial point processes with  Bragg peaks $\Bragg$ and eigenfunctions $f_k$, $k\in \Bragg$. These automorphisms are in one to one correspondence with elements from $\TT$ in the following way: Any character $\phi \in \TT$ yields an automorphism
of $(N,X,\mu,T)$ via $f_k \mapsto \phi(k)f_k$. Conversely any automorphism of $(N,X,\mu,T)$ must necessarily map map $f_k$ to a multiple  $\phi (k) f_k$ with $\phi (k)\in U(1)$. The arising function $\phi$ must then belong to $\TT$ as  \eqref{phaseCreated} gives
$\phi(k_1) \ldots \phi(k_m) =1$
whenever $k_1 + \cdots +k_m =0$.

(b) In the previous theorem, the isomorphism $M$ between isomorphic pure point stationary
spatial processes is not unique. However, if $M_1,M_2$ are two such isomorphisms then
$M:= M_2^{-1}M_1$ is an automorphism of the first, say $(N,X,\mu,T)$.
By (a) of this remark, we conclude that the isomorphisms are unique up to elements of $\TT$.
\end{remark}

\section{From phase forms to stationary processes with pure point spectrum: The torus approach}\label{CyclestoStochasticprocesses}
In this section we show how to  construct an ergodic full stationary process with a given pure point measure as its diffraction measure  $\omega$ and a given phase form $a^*$ as its associated phase form. This yields a canonical model realizing a given diffraction measure.

\bigskip

We assume that we are given a a locally compact Abelian group $\GG$, a pure point positive symmetric backward transformable measure measure $\difm$ on $\Gdual$. These data give rise to
$$\Bragg :=\{k\in \Gdual: \difm(k) >0\}$$
and $\ev = \langle \Bragg \rangle_{\rm group} \subset \Gdual$. There is an associated group of relators $\cZ$  as discussed in Section \ref{Cycles}.  We are also given $a^*\in\widehat{\cZ}$.

We are going to construct a stationary process corresponding to $(\difm, a^*)$.

From Prop.~\ref{characterCondition} we know  that there is an elementary phase form
$a:\cF(\Bragg) \longrightarrow U(1)$ (i.e. a character of $\cF(\Bragg)$) which restricts to $a^*$:
in other words, we can find $a$ so that
\[a^*(k_1,\dots, k_m) = a(k_1) \cdots a(k_m)\]
whenever $[{\bf k}]=0$.
(This $a$ is not unique, but the ratio of any two of such
elementary phase forms is a character on $\ev_d$.)
Thus, we can assume without loss of generality that $a^*$  comes from an elementary phase form $a \in \widehat{\cF(\Bragg)}$.
Since $\difm$ is backward transformable there is a measure $\gamma$ on
$\GG$ so that for all $F\in C_c(\GG)$ we have
\begin{equation} \label{eighteen}
\int_{\Gdual} |\widehat F|^2 \dd \omega = \int_{\GG} F*\tilde F \dd \gamma < \infty\, .
\end{equation}

We consider $\ev$ to have the induced topology from  $\Gdual$ and, as usual, let
$\ev_d$ denote the same group $\ev$ but with the discrete topology.
The dual of $\ev_d$
is a compact Abelian group $\TT$. We denote the Haar measure of total
volume $1$ on $\TT$ by $l_\TT$ and the counting measure on
$\ev_d$ by $l_d$.
Then, \eqref{eighteen} can be reformulated as saying that
$$ (k\mapsto \widehat{F} (k) \omega (k)^{1/2}) \in \ell^2 (\ev).$$
For $F \in C_c(\GG)$, define a
function $n_a(F)$ on $E_d$ by
\[n_a(F) = \sum_{k \in \Bragg} \widehat F(k) a(k) \difm(k)^{1/2} 1_{k} \]
where we take the positive square roots and $1_k$ is the function
on $\ev_d$ whose value at $k$ is $1$ and which takes the value $0$
everywhere else. We have
\begin{eqnarray*}
\int_{\ev_d} n_a(F) \overline{n_a(F)} \dd l_d &=& \sum_{k\in \Bragg}
n_a(F)(k) \overline{n_a(F)(k)} = \sum_{k\in\Bragg} |\widehat F(k)|^2 \difm(k)\\
&=& \int_{\Gdual} |\widehat F|^2\dd \difm = \int_\GG F*\tilde F \dd \gamma
< \infty \nonumber \, ,
\end{eqnarray*}
which implies in particular that $n_a(F) \in L^2(\ev_d, l_d)$.
The Fourier transform provides a  fundamental isomorphism between $L^2(\ev_d, l_d)$ and
$L^2(\TT, l_\TT)$ taking  $1_{-k}$ to the character defined by $k$
on $\TT$. We usually denote this character by $\chi_k$. Thus we obtain, by applying  the Fourier transform to $n_a (F) $,
\begin{equation}\label{def-na}
N_a(F) := \sum_{k \in \Bragg} \widehat F(k) a(k) \difm(k)^{1/2} \chi_k
\in L^2(\TT, l_\TT) \,
\end{equation}
with
\begin{equation}\label{normEquation}
\langle N_a(F) | N_a(F)\rangle =\int_\TT N_a(F) \overline{N_a(F)} \dd l_\TT =
\int_{\ev_d} n_a(F) \overline{n_a(F)} \dd l_d = \int_\GG F*\tilde F \dd \gamma \, .
\end{equation}
This provides us with the mapping
\begin{equation*}
 N_a \,:\, C_c(\GG) \longrightarrow L^2(\TT, l_\TT) \, ,
 \end{equation*}
and shows that for any compact subset $K \subset \GG$ we have
$||N_a(F)||_2 \le (|\gamma|(K+K))^{1/2} ||F||_{\sup}$ for all
$F\in C_c^\RR(\GG)$ with support inside $K$, which
shows that $N_a$ is continuous.

The continuous embedding of $\ev_d \rightarrow \Gdual$ defined by inclusion
leads to a dense homomorphism $\GG \rightarrow \TT$. Then each
function on $\TT$ ``restricts" to one on $\GG$ and in particular
we have an obvious meaning to the functions $t \mapsto \chi_k(t)$
for all $t\in \GG$ and for all $k \in \ev_d$.  We also obtain
a natural ergodic action of $\GG$ on $\TT$. For $t\in \GG$ and $\xi \in \TT$,
\[t.\xi : k \mapsto \chi_k(t) \xi(k)\]
for all $k \in \ev_d$.
For $F \in C_c(\GG)$ and $k\in \ev_d$
\[\widehat F(k) \overline{\chi_k(t)} = \int_\GG F(x) \overline{\chi_k(x)}\overline{\chi_k(t)}
\dd \haarG(x) = \int _\GG (T_t F)(x) \overline{\chi_k(x)} \dd \haarG (x) = \widehat{T_t F}(k) \,.\]
(Here we use the invariance of the measure $\haarG$.)

This leads by a straightforward calculation that $T_t N_a(F) = N_a(T_t F)$, which
shows that $N_a$ is a $\GG$-equivariant map.

 \smallskip

Thus given $(\difm,a^*)$ we obtain an ergodic  spatial
stationary process $(N_a,\TT,l_\TT)$.  The above considerations also show that $\difm$ is the diffraction measure associated to $N_a$.

We will show next that $N_a$ is full and has associated phase form $a^*$.  To do so we consider the map
$$\theta : L^2 (\Gdual, \difm)\longrightarrow L^2 (\TT,l_\TT), \theta (h)\mapsto \sum_{k\in\Bragg} h(k) a(k) \difm(k)^{1/2}\chi_k.$$
Then, $\theta$ is an isometry with $\theta (\widehat{F}) = N_a (F)$ for all
$F\in C_c(\GG)$. By uniqueness of the dynamics-to-diffraction map we infer that $\theta$ {\em is} the dynamics-to-diffraction map. This gives in particular, that
$\theta(1_k) = a(k) \difm(k)^{1/2}\chi_k$ for all $k\in \Bragg$. These are the functions
$f_k$ of the previous sections. From this equality and the fact that
$a$ is a character on $\cF(\Bragg)$ it follows easily that
 $a^*$ is the phase form associated to $N_a$ (see \eqref{phaseCreated}
 and \eqref{phaseFormCreated}).  Moreover, Proposition \ref{diffractiontodynamics} now shows that
  $a(k) \difm(k)^{1/2}\chi_k = \theta (1_k)$ belongs to the closure of the linear span of the $N(F)$, $F\in C_c (\GG)$, for any $k\in\Bragg$. As products of these $\chi_k$ form an orthonormal basis of $L^2 (\TT,l_\TT)$ we obtain fullness of the stationary process $N_a$. We summarize the conclusions of this section.

 \begin{prop} \label{inverseProblem} Let a pure point positive  symmetric backward transformable measure $\difm$ on $\GG$ be  given and $a^*$ be a phase form. Then for
 any choice of elementary phase form $a$ restricting to $a^*$, $(N_a,\TT,l_\TT,T)$ is an ergodic full stationary process on $\GG$ with  diffraction measure $\difm$ and phase form $a^*$.
  \end{prop}

\smallskip

The pure point process constructed in Prop.~\ref{inverseProblem} depends
on the choice of the elementary phase form $a$ that we choose to represent
the phase form $a^*$.
What happens if we choose another representing elementary phase form, $b\in\widehat{\cF(\Bragg)}$? We know already by Thm.~\ref{uniquenesstheorem}
that the resulting process will be isomorphic, and since both use the
same constructed dynamical system $X = \TT$, this difference between the two
processes is in effect an automorphism. Also, Proposition \ref{characterCondition} (ii) says that the ratio $u = b / a$ is an element of $\TT$. From the perspective of $N_a$ and $N_b$, hten $N_b$ is in effect a translation of $N_a$ ba $u\in \TT$ i.e.
\begin{eqnarray}\label{shifting TT}
N_b(F) &=& \sum_{k\in \Bragg} \widehat F(k) b(k) \omega^{1/2} \chi_k
=  \sum_{k\in \Bragg} \widehat F(k) a(k) \omega^{1/2} \chi_k(u) \chi_k\\
&=& \sum_{k\in \Bragg} \widehat F(k) a(k) \omega^{1/2} T_{-u}(\chi_k)
= (T_{-u} N_a)(F) \,. \nonumber
\end{eqnarray}

\medskip

\begin{remark} \label{continuousCase}
If $u = b/a$ above is actually continuous with respect to the original topology on $\ev$
then $u$, being a character on $\ev$, lifts to a character on $\overline{\ev}$
and then by general character theory to all of $\Gdual$. Thus $u$ can be identified with an element of $\GG$. Now
\eqref{shifting TT} becomes
\begin{eqnarray}\label{shifting G}
N_b(F) &=& \sum_{k\in \Bragg} \widehat F(k) b(k) \omega^{1/2} \chi_k
=  \sum_{k\in \Bragg} \widehat F(k) a(k) \omega^{1/2} \chi_k(u) \chi_k\\
&=& \sum_{k\in \Bragg} \widehat{T_{-u} F}(k) a(k) \omega^{1/2} \chi_k
= N_a(T_{-u} F) \,. \nonumber
\end{eqnarray}
This time the difference between $N_a$ and $N_b$ is a translation on $\GG$.
This situation occurs in the periodic situation, when $\GG$ is compact
and $\Gdual$ is discrete.
\end{remark}

\begin{coro} \label{momentsCondition2} Let $\omega$ be a positive backward transformable pure point measure on $\Gdual $ and $\cZ$ the associated group of relators.
If $\langle \cZ_n\rangle = \cZ$, then   two mappings $a,b:\Bragg \longrightarrow U(1)$
define isomorphic stationary pure point processes if and only if their first through $n$th
moments are equal.
\end{coro}
\begin{proof} This is a direct consequence of
Corollary \ref{momentsCondition1}.
\end{proof}

\section{The homometry problem}\label{Homometry}
In this section we discuss a main result of the paper viz our solution to the homometry problem for pure point diffraction.

\bigskip

\begin{definition}\label{diffractionClass}
Given $\GG$ and a measure $\difm$ on $\Gdual$, we let $\cN(\GG, \difm)$ denote
the set of all spatial  isomorphism classes of ergodic  full stationary spatial processes with diffraction
$\difm$ satisfying Assumption \ref{0-eigenfunction assumption}.
\end{definition}

The considerations of the previous sections rather directly prove  the following result, which is in effect a solution of the homometry problem
for pure point diffraction.

\begin{theorem} \label{homometrytheorem}  Let $\GG$ be a locally compact abelian group. Let $\difm$ be a positive backward transformable pure point measure on $\Gdual$ and $\cZ = \cZ (\difm)$ the associated group of relators. Then, the map
$$\widehat{\cZ}\longrightarrow \cN(\GG, \difm),$$
$$ a^*\mapsto [(N_a,\TT, l_\TT, T)],$$
where $a$ is any elementary phase form representing $a^*$,
is a bijection.
\end{theorem}
\begin{proof}
By Proposition \ref{inverseProblem}, $(N_a,\TT,l_\TT, T)$ is an ergodic spatial stationary process with diffraction $\omega$ and phase $a^\ast$. Moreover, by Theorem \ref{uniquenesstheorem}, the map is well-defined (i.e. independent of the choice of the elementary phase form representing $a^\ast$) and one-to-one.   It is surjective by Proposition \ref{fromsptopf} and Theorem \ref{uniquenesstheorem}.
\end{proof}

We note the following immediate consequence of the previous theorem and Corollary \ref{momentsCondition1} (see Corollary \ref{momentsCondition2} as well).

\begin{coro}\label{finite-coro} Assume the situation of the theorem. If $\langle \cZ_n\rangle = \cZ$, then two  ergodic spatial stationary processes
over $\GG$ with diffraction measure $\difm$ are spatially isomorphic  if and only if their phase forms  have the same $m$-th moments for $m = 1, \ldots, n$.
\end{coro}

\begin{remark} Although it is obvious from the theory of pure point spatial processes
that we have developed, it is worthwhile noting explicitly that the problem
of classification of pure point spatial processes with a given $\difm$ depends on
only the positions of the Bragg peaks, namely $\Bragg$. Once one has $\Bragg$, one has
the group $\cZ$, which depends only on the relationship of $\Bragg$ to the subgroup
$\ev$ of $\Gdual$ that it generates. And in fact this purely algebraic problem is in reality the entire homometry problem!  See in \S~\ref{compactG}.2 for a $1D$ case when $\GG$ is finite.
\end{remark}

\section{Compact Groups}\label{compactG}
Cases when $\GG$ is compact (or even finite) are important since they arise in the study of periodic and limit periodic cases of diffractive structures, and their relative simplicity allows us to see more clearly the nature of the spatial processes that we have introduced.
In this section we first look at what happens in general when we are given that the
group $\GG$ is compact.

We then move to some examples. These examples revolve
around the case that $\GG$ is finite and/or around one of the most basic of all diffraction measures, namely the diffraction of $\ZZ$, $\delta_\ZZ$. We also look at an
interesting result of Gr\"unbaum and Moore that involves rational periodic diffraction
on the line. The fact that there is anything to say in this remarkably simple case
and that it seems to be very difficult to extend their results to a two dimensional setting is a bit of
a testimony to the difficulties that are inherent in the homometry problem.

\subsection{The compact setting}\label{The compact setting}
We now look at the method of \S\ref{CyclestoStochasticprocesses} in this compact case.
We assume that $\difm$ is positive, pure point,
centrally symmetric ($\difm(-k) = \difm(k)$, for all $k\in \Bragg$), and backward transformable. Let $\Bragg$ be the support of $\difm$. It is convenient to assume
that $\Bragg$ generates $\Gdual$ as a group, for otherwise we may use $\langle \Bragg\rangle$ and $\widehat{\langle \Bragg\rangle}$ instead. Then $\ev = \Gdual = {\ev}_d$
and the carrying space for our spatial process is
$$X= \TT := \widehat{\ev_d} = \GG.$$
 When we wish to distinguish two roles of $\GG$, first as a group of translations and second as the set of states of dynamical system, we shall use notation
like $t$ and $\xi = \xi_t$ respectively. Moreover, for any character $k$ on $\TT$ (i.e. $k\in \Gdual$) we will denote the map
$$L^2 (X)\longrightarrow \CC, \; F\mapsto \widehat{F} (k)$$
by $k$ as well.

When we come to the transformability into the autocorrelation measure $\gamma$, we require that  for all $F\in C_c(\GG) = C(\GG)$,
\begin{equation}
\label{transformability-assumption} \gamma(F*\tilde F) = \int_{\Gdual} |\hat F|^2 d\omega =
\sum_{k \in \Gdual} \omega(k) |\hat F|^2(k)
= \sum_{k \in \Gdual} \omega(k) \widehat{F*\tilde F}(k).
\end{equation}
Thus the autocorrelation Fourier dual to $\difm$ is given by
\begin{equation} \label{definition-gamma}
\gamma := \sum_{k\in \Bragg} \difm(k)\;k = \sum_{k\in \Gdual} \difm(k)\;k \,
\end{equation}
provided that this is indeed  a measure. Since $\GG$ is compact, and $\gamma$
is supposed to be a Borel measure, this measure must be finite.

Assuming transformability we can, for any character  $a:\cF(\Bragg) \longrightarrow U(1)$, form the associated
stationary process
\[N_a:C_c(\GG) \longrightarrow L^2(X,\mu) \, ,\]
namely,
\[N_a(F) = \sum_{k\in\Gdual} \hat F(k) a(k) \difm^{1/2}(k) \, k \]
for all $F\in C(\GG)$. Note that $N_a (F)$ belongs indeed to $L^2 (X,\mu)$ due to
$\eqref{transformability-assumption}$.

\smallskip

Our interpretation is that $N_a$ represents some sort of density on the space
$\GG$. Each point $\xi_t \in X$ represents an instance of the yet unspecified
structure which can be paired with elements of $C_c (\GG)$ to give

\begin{equation}
\label{measureInterpretation}
 \langle \xi_t,  F\rangle = N_a(F)(\xi_t) = N_a (F) (t),
 \end{equation}
where the left hand side needs to be given an interpretation.  It turns out that such an interpretation can be given easily if we make the  finiteness assumption that
\begin{equation}\label{finiteness-assumption}
\omega \in L^1(\Gdual,l_{\Gdual}).
\end{equation}
Note that this assumption immediately implies that $\omega$ is backward transformable as it yields that $\gamma$ given in $\eqref{definition-gamma}$ is indeed a measure.

 Given $\eqref{finiteness-assumption}$, we can
define the function $\rho_a \in L^2 (X)$ by
\begin{equation}\label{densityDefined}
\rho_a:= \sum_{k\in \Gdual} a(k) \omega(k)^{1/2} \overline k
\end{equation}
by standard theory of Fourier series.  Then, a short calculation invoking unitarity of the  Fourier transform  (see e.g. \eqref{def-na} as well)  gives
\begin{eqnarray}\label{processAsMeasure}
N_a(F)(t) &=&\sum_{k\in \Gdual} \hat F(k) a(k) \omega(k)^{1/2} (k,t )
\nonumber \\
&=&   \int_\GG T_t\rho_a(x) F(x) d\haarG(x).
\end{eqnarray}
Note that we have indeed pointwise (in t) existence of $N_a (F) (t)$ as both $\widehat F$ and $\omega$ are square summable i.e. belong to $\ell^2 (\ev_d)$.  Comparing with $\eqref{measureInterpretation}$, we find that we can identify
$\xi_t$   with the $L^2$ function  $T_t \rho_a$ and the pairing between $ \xi_t$ and $C_c (\GG)$ with  ordinary integration.

When can this `density'
be interpreted as a measure $\rho_a$ on $X = \GG$?

The idea is that for some $\xi_0 \in X$ (which we can take to be $0 \in \GG$) the structure is $\rho_a$ and that as we translate around the translate $T_t\rho_a$ represents the structure at $\xi_t$. Usually one would have to assume by ergodicity the denseness of this orbit, but in the compact case as we have it here, the orbit is the entire space $X$. This entails that
\begin{equation}
\langle T_t \rho_a, F\rangle = N_a(F)(t).
\end{equation}
where
and then
$T_t\rho_a = \sum_{k\in \Gdual}  (k,t) a(k) \omega(k)^{1/2} \overline k$.
This gives \eqref{measureInterpretation} when we treat $\rho_a$ as a measure.

As $\GG$ is compact the $L^2$-function $T_t \rho_a$ belongs to $L^1$ as well and we can consider it to be the measure
$$ T_t \rho_a d\haarG.$$
Thus, in this situation we can realize the process as a measure process.

If $\sum \omega(k)^{1/2} <\infty$, then $\rho_a$ will even  be a continuous function.
If $\sum \omega(k)^{1/2} =\infty$, the situation becomes different. Such a situation can easily arise:
it suffices to  find $a$ and $\omega$ such that $\rho_a$ does not belong to $L^\infty  (X)$.
Let for example $\GG$ be  compact  admitting  a function $\rho\in L^2 (\GG)\setminus L^\infty (\GG)$.  As $\rho$ belongs to $L^2$, we
   can expand $\rho$ in a Fourier series
 $$\rho = \sum_ k a(k) \omega(k)^{1/2} \overline{k}$$
 with $a(k)\in U(1) $ and $\omega (k)\geq 0$ with $ \sum \omega (k) < \infty$.
The diffraction of $\rho$ is seen directly to be equal to $\omega$.

For more on interpreting $N_a$ as some sort of measure induced density
on $\GG$ see \S\ref{ssp-measures}.

\section{Examples around the diffraction measure $\difm =\delta_\ZZ$}

In general the number of solutions to the inverse problem is vast, even for
the most basic diffraction patterns. In this section we consider the most famous
example of a diffraction pattern, the diffraction of the integers. More properly
this is the diffraction of the Dirac comb $\delta_\ZZ= \sum_{x\in\ZZ} \delta_x$
which represents a point density of one at each of the integers. Its diffraction is also
$\difm =\delta_\ZZ$.

Crystallographers have long relied on additional information (like the periodicity of crystals
and their known constituents)
to find solutions to the inverse problem.
We shall see that by imposing periodicity and/or imposing arithmetic conditions on the
solutions, the inverse problem we can do the same. One of the advantages
of our approach, that applies to all locally compact Abelian groups, is that it also applies
to compact and to finite groups. Although our object of attention is
$\delta_\ZZ$, much of this section is devoted to cases in which $G$ is finite
and to showing what the theory then looks like and tells us. We also deal with the case $G=U(1)$.
\medskip

In the case of diffraction
from a periodic set or periodic crystal, it is customary to compute the diffraction directly from
the originating measure describing the density $\rho$, by taking the square-absolute value
of the Fourier transform of $\rho$. In the case of $\delta_\ZZ$ this approach gives the answer
directly as a consequence of the Poisson summation formula. This also works directly for
finite groups.

We start with  the pair of dual groups $\frac{1}{M}\ZZ/\ZZ$ and $\ZZ/M\ZZ$
where $M$ is a positive integer. The duality is most conveniently written down in the form
\[ (k,x) = e^{2 \pi i kx} \]
for $x \in \frac{1}{M}\ZZ/\ZZ$ and $k \in \ZZ/M\ZZ$. In the context of this paper, one of these
groups plays the role of $\GG$ and the other of $\Gdual$. Most of the time
it will be $\GG = \frac{1}{M}\ZZ/\ZZ$ and $\Gdual=\ZZ/M\ZZ$. To help clarify things, elements of $\Gdual$ will be written in the form $\chi_x$ or $\chi_k$ accordingly.

The idea here is to use the fact that the diffraction $\delta_\ZZ$ is periodic and then
represent it in $\ZZ/M\ZZ$ and look for all solutions for an originating distribution $\rho$
in $\frac{1}{M}\ZZ/\ZZ$. The resulting distribution will then be interpreted as a periodic
(modulo $\ZZ$) measure
on $\frac{1}{M}\ZZ$. In this way we pick up periodic solutions to the inverse problem. Alternatively
we could do things the other way around and consider $ \delta_\ZZ$ represented in
$\frac{1}{M}\ZZ/\ZZ$
and find solutions to the inverse problem in $\ZZ/M\ZZ$ which we could interpret as periodic
solutions in $\ZZ$. However, we shall see that does not lead to new solutions
to the inverse problem.

We note that for finite $\GG$ and $\Gdual$ the Haar measure we want to use
is counting measure
normalized to total measure $1$ by dividing by the order of the group. For
$\chi_k \in \ZZ/M\ZZ$, simple calculation gives the expected result that
$\widehat{\chi_k} = \delta_k$ (note that $\chi_k$ is a function on $\GG$ and
its Fourier transform is a function on $\Gdual$).

\smallskip

Since the groups are finite, there is no complication in computing diffraction.
For a density or distribution of density $\rho$ on $\GG$, the autocorrelation
is $\gamma = \rho*\widetilde \rho$ and the diffraction is $\difm = (\rho*\widetilde{\rho})\widehat{\phantom{x}}= | \widehat \rho |^2$.

\begin{remark} We should note that in general the approach via the square-absolute value of the Fourier
transform is not applicable to aperiodic structures.
 As A. Hof has shown \cite{Hof1} and \cite{Hof2}, in the study of
aperiodic crystals one usually meets distributions of density that are not Fourier transformable.
\end{remark}

\medskip
Our approach is to follow the lines laid out above and for given diffraction $\difm$
and elementary phase form $a$ define the associated density via
\eqref{densityDefined}:
\begin{equation}
\rho_a:= \sum_{k\in \Gdual} a(k) \omega(k)^{1/2} \overline k \,,
\end{equation}
This is a finite sum and there are no questions about convergence or ambiguities
about how to interpret it.

\subsection{Specific values of $M$}

\subsubsection{The case $M=1$}
This seems too trivial to consider, but it is interesting nonetheless.
$\GG= \frac{1}{1}\ZZ/\ZZ$, $\Gdual =\ZZ/\ZZ$,
$\difm = \delta_0$. Here $\Bragg = \{0\} \in \ZZ/\ZZ $ and
$\cF$ is generated by $e(0)=1$. So
$\cF$ is the trivial group and its only character is the trivial character. Also $\cZ = \cF$.
The density measure $\rho$ defined by the trivial character $a$ is
$a(0)\difm(0)^{1/2}\overline{\chi_0} = \chi_0$ which is simply the identity function on $\GG$.
Thus we obtain the expected solution to the inverse problem: $\widehat{\chi_0} = \delta_0$
whose absolute square is again $\difm = \delta_0$. The function $\chi_0$ is the identity function
on $\frac{1}{1}\ZZ/\ZZ$, which is also the measure  $\delta_0$ on $\GG$.

\subsubsection{The case $M=2$}
(1) \quad $\GG =\frac{1}{2}\ZZ/\ZZ$, $\Gdual=\ZZ/2\ZZ$,
$\difm = \delta_0 + \delta_1$. The two elements of
$\Gdual$ are $\chi_0$ and $\chi_{1}$. We have $\Bragg = \{0, 1\} =\Gdual$ and $\cF$
is generated by $e(0) =1 $ and $e(1)$, with $e(1)^2=1$ since $-1 \equiv 1 \mod 2\ZZ$.
Thus $\cF$ is cyclic of order $2$. There are only two characters on $\cF$, i.e. elementary
phase forms,
$a_+:e(1) \mapsto 1$ and $a_-:e(1) \mapsto -1$.
 The exact sequence
 \[1 \longrightarrow \cZ \longrightarrow \cF \longrightarrow \ev = \Gdual \longrightarrow 1\, ,\]
comes from the mapping $e(k) \mapsto k$ which here reads
$e(0) \mapsto 0$ and $e(1) \mapsto 1$. Thus the kernel $\cZ$ is trivial. The significance
of this is that there is only one phase form (trivial!) and thus the two solutions we derive
from the two elementary phase forms will be the same
up to translation. In detail, we have, using $a_+$
\[\rho_+ = a_+(0) \difm(0)^{1/2}\overline{\chi_0} + a_+(1)\difm(1)^{1/2}\overline{\chi_{1}}
= \chi_0 + \chi_{1}\,.\]
Note that this is a function on $\GG = \frac{1}{2}\ZZ/\ZZ$. Its values
are $\rho_+(0) =2, \rho_+(1/2) = 0$.
This is a nice transformable measure and its Fourier transform is the measure
is $\delta_0 + \delta_1$ on $\Gdual = \ZZ/2\ZZ$.
Its absolute square is $\difm$, as we expected.

Similarly, using $a_-$ we obtain
\[\rho_- = a_-(0) \difm(0)^{1/2}\overline{\chi_0} + a_-(1)\difm(1)^{1/2}\overline{\chi_{1}}
=  \chi_0 - \chi_{1} \,.\]
Its values are $\rho_-(0) =0, \rho_-(1/2) =2$,
so the density distribution $\rho_-$ on $\GG$ is just translation by $1/2$ of $\rho_+$
found in the case of $a_+$.
The Fourier transform is $\delta_0 - \delta_1$
and again the absolute square of this is $\difm$.

\bigskip

(2) \quad Next we reverse the roles of $\GG$ and $\Gdual$: $\GG =\ZZ/2\ZZ$, $\Gdual = \frac{1}{2}\ZZ/\ZZ$,
$\difm = \delta_0$. Notice that $\Gdual$ represents the reduction
of the integers and the half-integers modulo $\ZZ$. The reduction of
the diffraction $\delta_\ZZ$ modulo $\ZZ$ is thus supported only on the
$\ZZ$ part which is represented by the class of $0$ modulo $\ZZ$.
Following the general discussion of compact groups we should,
at this point, replace $\Gdual$ by the subgroup $\langle \Bragg \rangle = \{0\}$.
This then reduces us to the case $M=1$ already discussed.

 \bigskip
(3) \quad It is interesting to consider
$\GG =\ZZ/2\ZZ$ and $\Gdual = \frac{1}{2}\ZZ/\ZZ$ with the diffraction
$\difm = \delta_0 + \delta_{1/2}$. This may be considered as the reduction
modulo $\ZZ$ of $\delta_\ZZ + \delta_{\frac{1}{2} + \ZZ}$. The two elements of
$\Gdual$ are $\chi_0$ and $\chi_{1/2}$. The discussion parallels that in (1). We have $\Bragg = \{0, 1/2\}$ and $\cF$
is generated by $e(0) =1 $ and $e(1/2)$ with $e(1/2)^2=1$ since $-1/2 \equiv 1/2 \mod \ZZ$.
Thus $\cF$ is cyclic of order $2$. There are only two characters on $\cF$, i.e. elementary
phase forms,
$a_+0:e(1/2) \mapsto 1$ and $a_-:e(1/2) \mapsto -1$.
 The exact sequence
 \[1 \longrightarrow \cZ \longrightarrow \cF \longrightarrow \ev = \Gdual \longrightarrow 1\, ,\]
comes from the mappinng $e(k) \mapsto k$ which here reads
$e(0) \mapsto 0$ and $e(1/2) \mapsto 1/2$. Thus again the kernel $\cZ$ is trivial
and the two solutions we derive will
be the same
up to translation.  Using $a_+$
\[\rho_+ = a_+(0) \difm(0)^{1/2}\overline{\chi_0} + a_+(1/2)\difm(1/2)^{1/2}\overline{\chi_{1/2}}
= \chi_0 + \chi_{1/2} \,.\]
Thus the distribution on $\ZZ/2\ZZ$ is $\rho_+(0) =2, \rho_+(1)=0$.
Its Fourier transform is the measure $\delta_0 + \delta_{1/2}$, whose absolute square
is $\difm$.

Going back to $\ZZ$ we have $2\delta_{2\ZZ}$ on $\ZZ$ whose diffraction is
$\frac{1}{2}(2 \delta_{\frac{1}{2}\ZZ}) = \delta_\ZZ + \delta_{\frac{1}{2} +\ZZ}$, which
agrees with what we obtained after removing all periodicity.

Similarly, using $a_-$ we obtain
\[\rho_- =  a_-(0) \difm(0)^{1/2}\overline{\chi_0} + a_-(1/2)\difm(1/2)^{1/2}\overline{\chi_{1/2}}
= \chi_0 - \chi_{1/2} \,,\]
and the distribution on $\ZZ/2\ZZ$ is
$\rho_-(0) = 0, \rho_-(1) = 2$.
Then $\widehat{\rho_-} = \delta_0 - \delta_{1/2}$. Taking absolute squares we again
obtain the diffraction $\difm$.

Again going back to $\ZZ$, this time we have  $2\delta_{2\ZZ +1}$ on $\ZZ$ whose diffraction is
$\delta_\ZZ -  \delta_{\frac{1}{2} +\ZZ}$, which
agrees with what we obtained after removing all periodicity.

 \bigskip

\subsubsection{The case  $M=3$}
Here something more interesting happens. $\GG = \frac{1}{3}\ZZ/\ZZ$, $\Gdual = \ZZ/3\ZZ, \difm=\delta_0 + \delta_1+\delta_2$.
The three elements of
$\Gdual$ are $\chi_0,\chi_{1},\chi_{2}$, with $\chi_k(j/3) = (k,j/3) = \zeta_3^{kj}$,
where $\zeta_3 =e^{2 \pi i /3}$. We have $\Bragg = \{0, 1,2\} =\Gdual$ and $\cF$
is generated by $e(0) =1,e(1),e(2) $, with $e(2)= e(1)^{-1}$
since $-1 \equiv 2 \mod 3\ZZ$.
Thus $\cF \simeq \ZZ$ with $e(1)$ as a (multiplicative) generator.
There are infinitely many characters on $\cF$, i.e. elementary
phase forms, one for each $u \in U(1)$:
$a_u: e(1)^n \mapsto u^n \in U(1)$ and $e(2)^n\mapsto u^{-n}$ for all $n \in \ZZ$.
 The exact sequence
 \[1 \longrightarrow \cZ \longrightarrow \cF \longrightarrow \ev = \Gdual \longrightarrow 1\, ,\]
comes from the homomorphism that maps $e(1)^k \mapsto k \mod 3$ (which includes $
e(1)^{-k} = e(-1)^k\mapsto -k \mod 3$).
The kernel $\cZ \simeq 3\ZZ$ via the isomorphism $\cF \simeq \ZZ$.

For each $u\in U(1)$ we obtain

\begin{eqnarray}
\rho_u &=& a_u(0) \difm(0)^{1/2}\overline{\chi_0} + a_u(1)\difm(1)^{1/2}\overline{\chi_{1}} +
a_u(2)\difm(2)^{1/2}\overline{\chi_2}\\
&=& 1 \chi_0 + u\chi_{-1} + \overline u \chi_1\,.\nonumber
\end{eqnarray}
Thus the distribution on $\frac{1}{3}\ZZ/\ZZ$ is
$\rho_u(0) = 1+u+\overline u$,
$\rho_u(1/3) = 1+ u\zeta_3^2 + \overline u \zeta_3$, and
$\rho_u(2/3) =1 + u\zeta_3 + \overline u \zeta_3^2\,.$
Since $\widehat{\chi_k} = \delta_k$
we obtain $\widehat{\rho_u} = \delta_0 + u \delta_2 + \overline u \delta_1$.
Its absolute square is $\difm$ as required.

We note that replacing $u$ by $\zeta_3u$ or $\zeta_3^2u$ changes
the distribution $\rho$ by translation, which corresponds to the fact
that the elementary phase forms map $3:1$ to phase forms.

Thus there is a one parameter family of solutions to the inverse problem
for $\difm$, which are supported on $\ZZ$ and periodic of period $3$.

\smallskip

\subsubsection{Other cases of $M$}

If we look at the inverse problem to $\delta_0 + \cdots + \delta_{M-1}$
in $\frac{1}{M}\ZZ/\ZZ$ the situation evolves in a natural way.
At $M=4$ the value $2$ leads to
$e(2)^2=1$, and $e(1) = e(3)^{-1} \mapsto u \in U(1)$ is arbitrary, so we end up with solutions
$\chi_0 +  u\chi_{1/4} \pm \chi_{1/2} + \overline u\chi_{3/4}$.
At $M=5$ we are get two free parameters $u,v \in U(1)$, one relating to the pair
$1,4$ and one to the pair $2,3$ (integers taken modulo $5$), and the solutions are
$\chi_0 +  u\chi_{1/5} +v \chi_{2/5} + \overline v\chi_{3/5} +\overline u \chi_{4/5}$.
The results for higher  $M$ follow a similar pattern.

\subsection{What happens when $\GG = U(1), \Gdual = \ZZ$?}\label{G=U(1)}

We again consider the diffraction $\delta_\ZZ$, but now use $\ZZ$ as $\Gdual$
and $U(1)$ as $\GG$.
Then $\Bragg = \ZZ$ and $\cF$ is generated by $\{e(k) :k \in \ZZ \}$ with the single
set of relations $e(-k) = \overline{e(k)}$ (along with $e(0)=1$). Then we have elementary
phase forms $a:\cF \longrightarrow U(1)$ with $a(k) \in U(1)$ arbitrary except that
$a(0) =1, a(-k) = \overline{a(k)}$.
Each such $a$ leads to a formal solution
\[ \rho_a=\sum_{k\in\ZZ} a(k) \difm(k)^{1/2} \overline{\chi_k}
= \sum_{k\in\ZZ} a(k)  e^{-2 \pi i k.(\cdot)} \, ,\]
This is simply a formal Fourier series on $\ZZ$. However, given that for all $k$,
$|a(k)|=1$, there is no convergence as a Fourier series. This is expected
since in the simplest case when $a(k)=1$ for all $k$ its Fourier transform is the
unbounded measure $\delta_\ZZ$.

However, the solution for $N_a$:
\[N_a(F) = \sum_{k\in\ZZ} \hat F(k) a(k) \difm(k)^{1/2} \, k
= \sum_{k\in\ZZ} \hat F(k) a(k) \, k \]
for all $F\in C(U(1))$ is well-defined as an $L^2$-function, and we can learn something from it.
In the case that $a$ is the trivial elementary phase form we obtain
\[N_a(F)(t) = \sum_{k\in\ZZ} \hat F(k) (k,t)  = F(-t) = \delta_0(T_t F)\,.\]
Thus the associated density is $\delta_0$, which is what we expect. Recall
\S\ref{The compact setting} here.
The compact group $\TT$, which is $U(1)$ here, is viewed both as a dynamical system
with states $\xi_t$
and as a group of translations $t\in U(1)$ which translate the states about. $N_a(F)(t)$ is something that gives the effect of the process on the function $F$ at the state $\xi_t$.

Now let $K=-K$ be a finite symmetric subset of $\ZZ\backslash \{0\}$ and let
$a$ be an elementary phase form which satisfies $a(k) =1$ for all $k \notin K$
and $a(k) = \overline{a(-k)} \in U(1)$ arbitrary for all $k\in K\cap \ZZ_+$. Now for
$F\in C(U(1))$ we have
\[N_a(F)(t) =  \sum_{k\in\ZZ} \hat F(k) (k,t) + \sum_{k\in K} (a(k) -1) \hat F(k) (k,t) \,,\]
for all $t\in U(1)$. Following the same reasoning as in \eqref{processAsMeasure}
we arrive at $N_a(F)(t) = \delta_0(T_t F) + \rho(T_tF)$ where
\[\rho = \sum_{k\in K} (a(k)-1) \overline{\chi_k}\,.\]
Remarkably the density $\delta_0 +  \sum_{k\in K} (a(k)-1) \overline{\chi_k}$ on $U(1)$
is indeed a solution to the inverse
problem for $\delta_0$, as can be seen by directly computing its autocorrelation
and taking its Fourier transform. This shows the ubiquity of quite reasonable looking solutions
to the problem.

\subsection{A simple periodic case with imposed arithmetic conditions} \label{simplePeriodic}
In this section we
consider the very simple situation of
the diffraction from the integers when they are weighted periodically with
period $M$ by a set $w_0, \dots, w_{M-1}$ of real numbers. In this case the
diffraction is also periodic, and we will assume that the ambient space
is $\GG = \ZZ/M\ZZ$ and the diffraction occurs in its dual, $\Gdual \simeq \frac{1}{M}\ZZ/\ZZ$,
say
\begin{equation}\label{finiteDiffraction}
\difm = \sum_{k\in \ZZ/M\ZZ} \difm(k) \delta_k \,.
\end{equation}

As we noted above, making the assumption that $\GG= \ZZ/M\ZZ$ is not entirely innocent
since it puts considerable limitations on the kinds of solutions available -- namely
that they are periodic of period $M$ and the distribution of intensity is entirely confined to integers,
neither of which is mandatory. Taking a larger ambient space will in general produce more
solutions (many more!).

Equation \eqref{densityDefined} defines the density on $\GG$ arising from
$\difm$ and the elementary phase form $a$. Since its diffraction is to be
$\difm$, this leads to the $M$ equations
\begin{equation}\label{weightEquations}
\rho_a(x) =\sum_{k\in \Bragg} a(k) \difm(k)^{1/2} \overline{ (k,x)} = w_x, \qquad x\in \ZZ/M\ZZ
\end{equation}
where $\Bragg \subset \Gdual =\frac{1}{M}\ZZ/\ZZ$ is the set of
non-vanishing points of $\difm$ and $ (k,x) = \exp(2 \pi i kx) = \zeta_M^{Mkx} \in
\QQ[\zeta_M]$, where $\zeta_M=\exp{2 \pi i /M}$ (the values of
$k$ are of the form $k_0/M \mod \ZZ$ where $k_0\in\ZZ$). Notice that for finite groups, the situation
we are in here, the densities $\rho_a$ are functions and also measures, since the distinction
between them disappears.

Gr\"unbaum and Moore \cite{GM} consider the case in which
$\rho_a$ is rational valued, i.e.
all of the weights $w_x$ are rational numbers. In this case, solving the system
of equations \eqref{weightEquations} for the `unknowns' $a(k) \difm(k)^{1/2}$
we see that they must all lie in $\QQ[\zeta_M]$, and then
$a(k) \difm(k)^{1/2}\overline{a(k) \difm(k)^{1/2}} = \difm(k) \in \QQ[\zeta_M]$ also.

Let $H:= {\rm Gal}(\QQ[\zeta_M]/\QQ) \simeq (\ZZ/M\ZZ)^\times$, the group of units
of the ring $\ZZ/M\ZZ$. We can view $H$ as both a group of automorphisms
permuting the $M$ roots of unity, or as a group of automorphisms of $\frac{1}{M}\ZZ/\ZZ$;
namely $s\in (\ZZ/M\ZZ)^\times$ acts as $\alpha_s:\exp(2 \pi i k) \mapsto \exp(2 \pi i sk)$
and as $\alpha_s:k \mapsto sk$ respectively. In the rational case we are now in, it is shown in
\cite{GM} that
$\alpha(\Bragg) = \Bragg$ for all $\alpha \in H$. We show  next why this happens:
We take $a(k)\difm(k)^{1/2} = 0$ if $k\notin\Bragg$.

If $k_1 \in \Gdual$ then
\begin{equation}\label{orthogonalSum}
\frac{1}{|\GG|} \sum_{x\in \Gdual}  k_1(x) \rho_a(x) =
\sum_{k\in \Bragg} a(k) \difm(k)^{1/2} \frac{1}{|\GG|} \sum_{x\in \Gdual}  k_1(x) \overline{k(x)} = a(k_1) \difm(k_1)^{1/2}\, ,
\end{equation}
which is zero if $k_1\notin\Bragg$.

\smallskip
Let $\alpha = \alpha_s \in H$. Since $\rho_a$ is rational valued,
$\rho_a = \alpha(\rho_a)
= \sum_{k\in \Bragg} \alpha(a(k) \difm(k)^{1/2})  \alpha(\overline k)$.
Applying the process of summation of equation \eqref{orthogonalSum} to the element $\alpha(k_1)$ where $k_1 \in \Bragg$ we obtain $a(\alpha(k_1))\difm(\alpha(k_1))$. But applying $\alpha$ dircectly
to \eqref{orthogonalSum} we obtain
$\frac{1}{|\GG|} \sum_{x\in \Gdual}  \alpha(k_1(x)) \alpha(\rho_a(x)) = \alpha(a(k_1)\difm(k_1)^{1/2})$. These two are equal, and this shows, in particular, that
$a(\alpha(k_1))\difm(\alpha(k_1))^{1/2} \ne 0$, and hence that $\alpha(k_1) \in \Bragg$,
which is what we wished to prove.
In short, if we identify
 $\Bragg$ as a subset of $\frac{1}{M}\ZZ/\ZZ$, then $k \in \Bragg$ and $s\in (\ZZ/M\ZZ)^\times$
 implies $sk \in \Bragg$. In effect $\Bragg$ is a union of orbits of $\frac{1}{M}\ZZ/\ZZ$
 under the action of the multiplicative group $(\ZZ/M\ZZ)^\times$.

 This condition
 appears in an equivalent form in the following key result of \cite{GM}. For
 $h\in \ZZ$, let $\mbox{\rm ord}(h)$ denote its order as a group element modulo $M$
 (under addition).

 \begin{prop} \label{GMThm}If the density $\rho_a$ of \eqref{weightEquations} is rational
 valued then for all $k\in \Bragg$ and for all $j \in \ZZ$ with $\gcd(j, \mbox{\rm ord}(Mk)) =1$
 one has $jk \in \Bragg$. In this case any mapping $a:\Bragg \longrightarrow U(1)$
 satisfying the $m$-moment condition for all $m \le 6$ if $M$ is even (respectively
 $m\le 4$ if $M$ is odd) extends to a character on $\frac{1}{M}\ZZ/\ZZ$.  \qed
 \end{prop}

 Using Proposition~\ref{charChar}, and Corollary~\ref{finite-coro} and the
 corresponding notation, we obtain
  \begin{coro} Under the conditions of Proposition~\ref{GMThm}, the
  set of distributions $\rho_a$ whose diffraction is $\difm$ of \eqref{finiteDiffraction}
  is classified by the first $6$ moments (resp. first $4$ moments)  if $M$  is even
  (respectively if $M$ is odd). \qed
 \end{coro}

We do not know of any extensions of this result into higher dimensional periodic
situations, and indeed it seems very difficult to do.

\begin{example} Here we give an example based on an example given in \cite{GM}
(see also \cite{XR3}).
We begin with $\GG:= \ZZ/6\ZZ$ and the two weighted Dirac combs
\begin{eqnarray}\label{Z6 density}
 11 \delta_0 + 25 \delta_1 + 42\delta_2 +45 \delta_3  + 31 \delta_4 +  14\delta_5\\
10 \delta_0 + 17 \delta_1 + 35 \delta_2 + 46 \delta_3 + 39 \delta_4 + 21 \delta_5\nonumber
 \end{eqnarray}
 on $\GG$, which represent two density distributions  for which we would like to find all possible
 homometric density distributions. What is remarkable is that not only are these two density
 distributions homometric, they actually have all of their first $5$ moments in common.

 Let us begin with the first of these two densities. The Fourier transform of this density is a function on $\Gdual \simeq \frac{1}{6}\ZZ/\ZZ$. The latter is most conveniently thought of as the $6$th roots of unity.
 We write the Fourier transform in the form
 \begin{eqnarray}
(5 +  25 w + 42 w^2&+&45 w^3 + 31 w^4 + 14 w^5 + 6 w^6)\\
&=&
(w + 1)(w^2 + w + 1)(2 w^2 + 5)(3 w + 1) =:P(w) \nonumber
 \end{eqnarray}
 where $w$ varies over the $6$th roots of $1$,
 $U_6:= \{e^{2 \pi ij/6}: j=0, \dots, 5\}$. Note that $w^6= 1$, so the coefficient
 of $w^0$ is $11/6$ as it is supposed to be.  The diffraction at $k \in \Gdual$ is $\difm(k) = P(w^k)P(\overline{w}^k)$, which works out
 to give the following values for $\difm(k)^{1/2}:$ $  28$, $  \sqrt{247/3}$, $ 0$ ,$0$ ,$0$ ,$ \sqrt{247/3}$
 at $k=0,1/6,\dots, 5/6$.
 Thus in this example $\Bragg = \{0,1/6,5/6\}$. In fact Gr\"unbaum and Moore constructed
 the polynomial $P$ precisely to make the Bragg spectrum have extinctions at
 $2/6,3/6,4/6$. In the rest of this subsection we shall make the notation a little
 simpler by suppressing  the ubiquitous denominators
 $6$ which occur in the values of $k$. Alternatively we multiply everything by
 $6$ and work in $\ZZ/6\ZZ$ instead of $\frac{1}{6}\ZZ/\ZZ$.

  In any $m$-tuple ${\bf k} = (k_1,\dots, k_m)$ in $\Bragg^m$, we can, using the equivalence
  relation $\sim$, drop any occurrences of $0$ and drop any pair $1,5$ that occurs amongst
 the $k_j$. This leads us to conclude the only equivalence classes of $\bbS$ are of the form
 $0^\sim, (1,\dots, 1)^\sim, (5, \dots, 5)^\sim$, where there are any positive integer
 $m$ entries in the two latter types. Since $(1,\dots, 1)^\sim$ and $(5, \dots, 5)^\sim$
 with $m$ entries are inverses of each other in $\bbS$, we conclude that
$\bbS^\sim \simeq \ZZ$. Since the mapping $\phi: \bbS^\sim\longrightarrow \ev_d \simeq
U_6$, we have $\varphi:\cF(\Bragg) \longrightarrow U_6$ (Prop.~\ref{F(S) homomorphism}).

Using the identification of $\bbS^\sim$ and $\cF(\Bragg)$ we have $\cF(\Bragg) \simeq \ZZ$
and we can put together the exact sequence \eqref{exactSeq1}:

\[1 \longrightarrow 6\ZZ \simeq \cZ  \longrightarrow \cF(\Bragg) \simeq \ZZ
\longrightarrow U_6 \simeq \ev_d  = \Gdual \longrightarrow 1.\]

Dualization then gives the exact sequence \eqref{exactSeq2}
\[1 \longleftarrow  \widehat{\cZ} \simeq U(1)/U_6\simeq \frac{1}{6}\ZZ/\ZZ \longleftarrow \widehat{\cF(\Bragg) }\simeq U(1)
\longleftarrow  \ZZ/6\ZZ \simeq \TT \longleftarrow 1 .\]
Each element of $U(1)$ can be viewed as a character of $\frac{1}{6}\ZZ$,
and when restricted to $\ZZ$
characters $a(\cdot)$ and $e^{2\pi ij(\cdot)} a(\cdot)$, $j=0,\dots, 5/6$, are the same.

We now write down the densities
\[\rho_a(x)  =  \sum_{k\in \Gdual} a(k) \omega(k)^{1/2} \overline{k}(x) =
 28 + \sqrt{247/3} \, a(1)  e^{-2\pi i x/6}+ \sqrt{247/3} \,a(5) e^{2 \pi i x/6} \,.\]
Here $x\in \GG \simeq \ZZ/6\ZZ$. The character $a$ of
$\cF(\Bragg)$ is a character of $\ZZ$ and hence is an element of $U(1)$. Recall
that $\cF(\Bragg)$ is generated by the tuples $(1,1,\dots)$ and their inverses
$(5,5,\dots,5)$ where $1,5$ represent classes modulo $6$ in $\ZZ$. There are
no further relations. Thus $a(k)$ here really means $a((k)^\sim)$, that is, the value
of $a$ on the the equivalence class of $(k)$ under $\sim$. Thus
by $a(1)$ we mean $a((1)^\sim)$ and by $a(5)$ we mean
$a((5)^\sim) = \overline{a(1)}$. There is no restriction on what
$a(1)$ might be in $U(1)$, so simply writing it as $a$ we may write $\rho$ as
\[\rho_a(x)  =
 28 + \sqrt{247/3} \, a  e^{-2\pi i x/6}+ \sqrt{247/3} \,\overline a e^{2 \pi i x/6} \,\]
 where $a\in U(1)$ is arbitrary. Each value of $a$ gives a solution to the
 diffraction problem.

 The following figure illustrates the graphs of $\rho_a(x)$.
 The situation that we began with occurs when $a \sim 0.443099$.
 Another solution indicated in \cite{GM} is $10, 17, 35, 46, 39, 21$ which occurs
 at $a\sim 0.520310 $.

 \begin{figure}
\centering
\includegraphics[width=7cm]{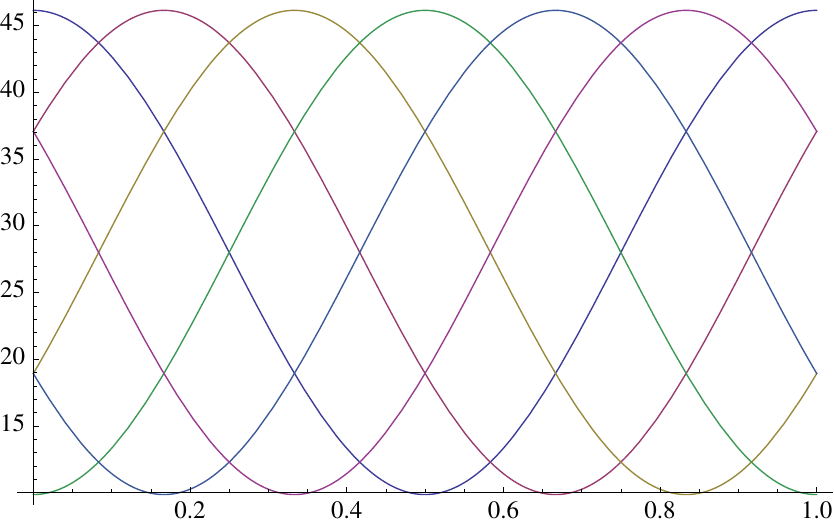}
\caption{The graphs of
 $\rho_a(x)$ as $a$ varies for each of the six values $0,1,2,3,4,5 \mod 6$ of $x$.
 Here $a$ varies over the unit circle $\{e^{2 \pi i t}: t\in [0,1)\}$, indicated here by the interval $[0,1)$
 and the six graphs give the corresponding values of the coefficients or weights of the deltas
 that make up the distribution on $\ZZ/6\ZZ$. The original weighting distribution
$11, 25 ,42, 45 ,31 ,14$ occurs when $a \sim 0.443099$. The second
density of \eqref{Z6 density} with coefficients $10, 17, 35, 46, 39, 21$ also occurs,
when $a \sim 0.520310$. }
\label{Z6Graph}
\end{figure}

Now we are going to show that in this example, $\langle \cZ_6 \rangle = \cZ$,
but $\langle \cZ_5 \rangle \subsetneq \cZ$. The argument is quite simple. First note
that $(1,1,1,1,1,1) \in \cZ_6$ but cannot be reduced, so its reduced length is $6$ (note
that denominators are still being suppressed).
So $\langle \cZ_5 \rangle \ne \cZ$. Now consider any  ${\bf k} = (k_1,\dots, k_n)\in Z_n$
with $n \ge 6$, which is of minimal length and can not be split into  a product of two (non trivial) elements from $\cZ$.  There can be no occurrences of $0$ since those can be discarded. Since $\{1,5\}$ and $\{2,4\}$ are annihilating pairs modulo $6$, we cannot
have both $1$ and $5$ in ${\bf k}$ and neither can we have both $2$ and $4$. Neither
can we have $2$ appearing three or more times since $(2,2,2) \in Z_3$ and similarly
for $4$. Likewise $3$ can appear at most once. Suppose $2$ or $3$ appears in ${\bf k}$.
The possibilities are $2,22, 3, 23, 223$. The remaining symbols are just $1$s or
$5$s. But $21111, 2211, 3111,255, 3555$ then show that ${\bf k}$ is not reduced. The
argument with $4$ is much the same. So we can assume that none of $0,2,3,4$ appear
and not both $1$ and $5$ appear. This leaves just $(1,\dots,1)$ and $(5,\dots,5)$
as the possibilities for reduced words, and in either case the number of entries must be
a multiple of $6$.

The consequence of this is that any phase form is uniquely determined by its
first through sixth moments. For if two have these moments in common then
taking corresponding elementary  phase forms $a,b$ we find that $ba^{-1} \in \cF(\Bragg)$
and satisfies the first $6$ moment conditions. Since $\langle \cZ_6 \rangle = \cZ$, $ba^{-1}$
lies in $\TT$ and so the corresponding phase forms were equal.

It is a surprising fact, pointed out in \cite{GM}, that the two densities of \eqref{Z6 density},
$11 \delta_0 + 25 \delta_1 + 42\delta_2 +45 \delta_3  + 31 \delta_4 +  14\delta_5$ and
$10 \delta_0 + 17 \delta_1 + 35 \delta_2 + 46 \delta_3 + 39 \delta_4 + 21 \delta_5$
indicated above, have all of their first through fifth moments in common, and it is
only at the sixth moments that they look different.

 \end{example}


\section{An alternative construction of spatial processes from pure point diffraction}\label{altConstruction}

In \S\ref{CyclestoStochasticprocesses} we offered a direct method of constructing a spatial process from the original data consisting of a positive  backward  transformable point measure
$\difm$ on $\Ghat$ and a phase form $a^*$ (or equivalently an elementary
phase form $a$). This approach can in a certain sense be understood as an elaboration of  the theorem of Halmos - von Neumann which
tells us in advance that since the required process will have an associated
dynamical system with pure point spectrum, it can be modelled (measure theoretically) on
a compact Abelian group.

In \S\ref{pc:split} we have used the Gelfand method to produce pure point spatial processes in the splitting of spatial processes into their pure and continuous parts.

In this present section we sketch out how we can also use the Gelfand method to produce a pure point spatial process out of the given pure point measure $\difm$ and the elementary phase form $a$. The method ultimately reverses many of the steps that we saw earlier in which, starting
from a process we constructed the diffraction and the diffraction-to-dynamics map. For this reason we do not give proofs but rather provide an outline of intermediate steps.

\bigskip

We use the notation of \S\ref{Cycles}. For simplicity of presentation we shall assume
that $0\in \Bragg$.
Define $c :\cS \longrightarrow \CC$ by
\[ c(k_1, \dots, k_n) = a(k_1, \dots, k_n )\difm(k_1)^\frac{1}{2}\dots \difm(k_n)^\frac{1}{2} \, , \]
or more compactly,
$c({\bf k}) = a({\bf k})\difm({\bf k})^\frac{1}{2}$,
where for each square root we take the non-negative root.  We define
$c(\emptyset) = 1$. Note that $c({\bf k}({-\bf k})) = \omega({\bf k})$.

Let
$\CC[F_k : k \in \Bragg]$
be the free associative algebra in
the variables $F_k$. For elements ${\bf k} = (k_1, \dots , k_m) $
where the $k_j \in \Bragg$, we will write
$F_{\bf k} := F_{k_1} \dots F_{k_m}$.
We define an
algebra $\ev$-grading on $\CC[F_k : k \in \Bragg]$ by assigning degree $k$ to
each $F_k$, so $F_{\bk}$ gets degree $[\bk]$.

Let $I = I(a)$ denote the ideal of $\CC[F_k : k \in \Bragg]$ generated by
all the elements of the form
$F_{\bf k} - c({\bf k})$
for all ${\bf k} = (k_1, \dots k_m) \in Z$, i.e. for
all ${\bf k}$ with $[{\bf k}] = 0$ (this includes
$F_0 - \difm(0)^{1/2} 1,\quad F_k F_{-k} - \difm(k)$)
and let
\[ \alg = \alg(a) := \CC[F_k : k \in \Bragg]/I \,. \]
We let \[\alpha: \CC[F_k : k \in \Bragg] \longrightarrow \alg\]
be the natural homomorphism and set
\[ f_k := \alpha(F_k), \quad f_{\bk} := \alpha(F_{\bk}) = f_{k_1} \dots f_{k_m} \]
for $\bk = k_1 + \dots + k_m$.

Since $c(k_1, \dots, k_m)$ is independent of the order of the $k_j$,
we see that $\alg$ is commutative.
Furthermore, since the relations imposed on $\CC[f_k : k \in \Bragg]$ are homogeneous of degree $0$, the resulting algebra $\alg(a)$ is also graded by $\ev$. We denote its space of elements of degree $\kappa \in \ev$ by $\alg_\kappa$.

The space $\CC[f_k : k \in \Bragg]_\kappa$ of elements of degree $\kappa \in \ev$
is the linear span of all the elements $f_{\bk}$ with $[\bk] = \kappa$. If
$[\bk] = [\bl] = \kappa$ then $\bk(-\bl) \in Z$ and the relations defining $I$
lead to
\begin{equation} \label{invertibility}
f_{\bk} f_{-\bl} =  f_{k_1} \dots  f_{k_m} f_{-l_1} \dots  f_{-l_n}
=a(\bk (-\bl)) \difm(\bk)^{1/2} \difm(-\bl)^{1/2} \in \CC^\times 1  \,.
 \end{equation}
From this we see first that
$ \alg_\kappa  = \CC f_\kappa$
where $f_\kappa$ is any one of the elements $f_{\bk}$
with $[\bk] =\kappa$.

We shall always assume that we are working with a phase form $a$ for which $\alg(a)$ is not reduced to the trivial ring $\{0\}$.
Such cycles always exist, for example the trivial cycle, which is $1$ everywhere.


Since the algebra is $\ev$-graded, non-zero elements of different degrees are linearly independent
and

\begin{equation*}
\alg = \alg(a) = \bigoplus_{\kappa \in \ev} \CC f_\kappa \, .
\end{equation*}
Define an conjugate linear form on $\CC[f_k : k \in \Bragg]$ by sesqui-linear extension (with conjugation on the second variable) of
\[\langle f_{\bf k} \mid f_{\bf l} \rangle =
\langle f_{k_1}  \dots  f_{k_m}  \mid f_{l_1} \dots  f_{l_n}\rangle
= \left \{ \begin{array}{ll}
c(({\bf k}) (-{\bf l}))  \; \mbox{if \; $[{\bf k}] = [{\bf l}]$} \, ,\\
0   \quad \mbox{otherwise} \, .
\end{array}
\right.
\]

It is a straightforward exercise to show the identity
\[ \langle f \mid g \rangle = \overline{\langle g \mid f \rangle} \]
for all $f,g \in \CC[f_k : k \in \Bragg]$. We shall write $||f||$ for $\langle f \mid f \rangle^{1/2}$.

\begin{prop} Let ${\bf k}, {\bf l}, {\bf m}$ be tuples of elements of
$\Bragg$.
\begin{itemize}
\item[{(i)}] $\langle f_{\bf k} \mid  f_{\bf k}\rangle = c({\bf k}(-{\bf k})) = \omega({\bf k}) >0$;
\item[{(ii)}] $\langle 1 \mid 1\rangle = 1$; $\langle f_\bl\,| \,1  \rangle = c(\bl)$
 \mbox{ when} $ [\bf l] = 0$;
\item[{(iii)}] $\langle f_{\bf k} f_{\bf m} \mid f_{\bf l} \rangle
= \langle f_{\bf k} \mid f_{-\bf m}f_{\bf l} \rangle$;
\item[{(iv)}]
$\langle f_{\bf k} - c({\bf k})1 \mid f_{\bf l} \rangle = 0 $ \; when $[\bk] =0$;
\item [{(v)}] The kernel of $\langle \cdot \mid \cdot \rangle $ is the ideal of relations
that we have imposed on $\CC[f_k : k \in \Bragg]$;
\item[{(vi)}] $\langle \cdot \mid \cdot \rangle $ induces an inner product (for which we use
the same notation) on $\alg$.
\end{itemize}
\end{prop}

Now $\alg$ is an inner product space. We wish next to consider $\alg$
acting on itself by multiplication and to make it into a normed algebra
by use of the operator norm. The details are as follows.
 A typical element of $\alg$ can be written as a finite sum
$x = \sum b_\bl f_\bl$  with $b_\bl \in \CC$, where there is only one
summand of any particular degree (i.e. $[\bl] \ne [\bl'] $ if
$\bl \ne \bl' $).
Then we have
\[ ||x||^2 =\sum |b_\bl|^2 ||f_\bl||^2 = \sum |b_\bl |^2 \omega(\bl) \,.\]
Let $\bk = (k_1, \dots, k_m) \in \Bragg^m$. Then $f_\bk x =  \sum b_\bl f_\bk f_\bl = \sum b_\bl f_{\bk\bl}$ with degrees
$[\bk] + [\bl]$ and
\begin{eqnarray*}
||f_\bk x||^2  &=& \langle \sum b_\bl f_{\bk\bl} \mid \sum b_\bl f_{\bk\bl}\rangle^{1/2}\\
&=& \sum |b_\bl|^2 ||f_{\bk \bl}||^2 =  \sum |b_\bl|^2 \omega(\bk \bl)
=\omega(\bk ) \sum |b_\bl|^2 \omega(\bl)|| =  ||f_\bk||^2||x||^2\,.
\end{eqnarray*}
Thus
\begin{equation} \label{norm}
||f_\bk x|| =  ||f_\bk|| \,||x|| \,
\end{equation}
which shows that for each tuple $\bk$, left multiplication by $f_\bk$ is a linear operator of norm $||f_\bk||$.

Now we see that the entire left representation
of $\alg$ on itself is a representation by bounded linear operators, and this provides
us with an operator norm $\nu$ on $\alg$ with $\nu(f_\bk) = ||f_\bk||$ and
$\nu(\sum b_\bl f_\bl) \le \sum |b_\bl| \nu(f_\bl)$ in general. In view of \eqref{norm}
we always have
\[ \nu(f_\bk) = \difm(\{\bk\})^{1/2} \;\mbox{and} \; ||x|| \le \nu(x)\;  \mbox{ for all } \; x \in \alg \,.\]
Being an operator norm, $\nu$ automatically satisfies
\[\nu(xy) \le \nu(x) \, \nu(y) \,.\]

\smallskip

We let $\cA$ denote the completion of $\alg$ with respect to $\nu$. This is a commutative Banach algebra. The norm will still be called $\nu$. We identify
$\alg$ with its image in $\cA$. Since $\nu(f) \ge ||f||$ for $f\in \alg$, the embedding
of $\alg$ in $\cA$ is faithful.

We define an $\GG$ action on $\alg(a)$ by
\[ T_t f_{k_1}  \dots  f_{k_m} = (k_1 + \dots + k_m)(t) f_{k_1}  \dots  f_{k_m} \, ,\]
or in abbreviated form
$T_t f_\bk = [\bk](t) f_\bk$,
for all $t\in \GG$. We note that $T_t$ acts as a unitary transformation relative
to the inner product and
we obtain a unitary representation $T$ of $\GG$ on $\alg$.
Each $T_t$ is an automorphism of $\cA$ as a Banach
algebra.

At this point we can apply Gelfand theory to produce the space $X = {\rm Spec}(\cA)$.
which is the set of algebra
homomorphisms  $\pi: \cA \longrightarrow \CC$.
Then $X$ is non-empty and compact in the Gelfand topology \cite{Lax}.
This is the weakest topology on $X$ making continuous the evaluation mappings
$ \eta_f: \pi \mapsto \pi(f)$
for each $f\in \cA$. This means that the open sets are generated
from the open sets
\[U(f,V) := \{ \pi \in X: \pi(f) \in V \}\, ,\]
where $f \in \cA$, $V$ open in $\CC$.
This is the way in which $\cA$ may be viewed as a space of continuous mappings on
$X$, namely $f(\pi) := \pi(f)$,  and this is the point of view that we take from now on.

One has for all $\pi \in X$, $f\in \cA$,
$ |\pi(f)| \le \nu(f)$
(see \cite{Lax}, Ch. 18, Prop.~1) from which for the sup norm we have
$||f||_\infty \le \nu(f)$.

For all  tuples $\bk$ from $\Bragg$ and for all $\pi \in X$,
\[ |f_\bk (\pi)| = ||f_\bk||   \,.\]

Since $T_t$ acts as an algebra automorphism of $\cA$, $T_t \pi =\pi \circ T_{-t}$
is another homomorphism of $\cA$ into $\CC$, and so is in $X$.
So $\GG$ acts on $X$. Since $T_t$ just permutes the defining open sets of the topology of $X$, this action is continuous.

\begin{lemma} \label{continuity}
For each $g\in \cA$ the mapping
\[ \GG \longrightarrow \CC, \quad \quad t \mapsto \nu(T_t g - g) \]
is uniformly continuous.
\end{lemma}

\begin{prop} \label{bicontinuity}
The mapping
\[ \GG \times X  \longrightarrow X \]
defined by the action of $\GG$ on $X$ is continuous.
$(X,\GG)$ is a topological dynamical system.
\end{prop}

Since $\GG$ acts on $X$, it acts naturally on $C(X)$ too. $\alg$ already has an action of  $\GG$ and acquires another one when treated as a subalgbra of $C(X)$; however, not surprisingly, the two actions of $\GG$ on $\alg$ the same.

  \begin{prop} $\alg$ is dense in the space $C(X)$ of all continuous functions under the sup-norm.
\end{prop}

\medskip
The next objective is to create an ergodic invariant measure for $(X,\GG)$.
We define a linear functional $\mu = \mu_{a}$ on $\alg$ by linear extension of
\begin{equation*}
\mu( f_\bk) :=
\begin{cases} c(\bk), &\text{if $[\bk] = 0$}\,;\\
0 &\text{otherwise}
\end{cases}\end{equation*}
for all tuples $\bk$ from $\ev$. Note that this is well-defined -- the relations amongst
the $f_\bk$ for equal values of $[\bk]$ are consistent with the definition --
and one can see immediately that $\mu$ is $\GG$-invariant since it picks up only the
$0$-eigenspace of $\alg$.

\smallskip

We need to extend this to a linear functional defined on all of $C(X)$.
Let $\{A_n\}_{n=1}^\infty$ be a van Hove sequence for $\GG$. For each
$f \in C(X)$ define a new function $M(f)$ on $X$ by
\[ M(f)(\pi)  := \lim_{n\to\infty} \frac{1}{\vol A_n} \int_{A_n} (T_t f)(\pi) \, dt \,, \]
where integration is with some prefixed Haar measure on $\GG$ and
$\vol(A_n)$ is the value of this measure on $A_n$.

\begin{lemma} \label{mean1}
For all $f \in \alg$, $M(f)$ is a constant function. Moreover,
\[M(f_\bk)  = \begin{cases}
c({\bk}),  &\text{if \; $[\bk] =  0$} \,;\\
0 \ &\text{otherwise} \,.
\end{cases}\]
\end{lemma}
This shows that $\mu$ is the same as the mean $M$ on $\alg$. But, as we now see,
 $M$ is defined on all of $C(X)$ and is a bounded linear operator on $C(X)$:

\begin{lemma} \label{mean2}
For all $g \in C(X)$, $||M(g)||_\infty \le ||g||_\infty$. For all $g \in C(X)$,
$M(g)$ is a constant function.
\end{lemma}

In this way, $M$ produces a measure on $X$, which coincides
on $\alg$ with $\mu$ defined above. We also call this measure $\mu$ and note
the fact that we have established, namely for all $g\in C(X)$, and for all $\pi \in X$,

\[ \mu(g) =\lim_{n\to\infty}  \frac{1}{\vol A_n} \int_{A_n} (T_t g)(\pi) \, dt \, , \]
which is a form of ergodic theorem.

\begin{prop} \label{cc}
For each $\kappa \in \ev$, the space of $\kappa$-eigenfuctions
in $C(X)$ (for the action of $T$ on $X$) is one dimensional and it is spanned by
any of the functions $f_\bk$ for which $[\bk] = \kappa$. In particular the $0$-eigenfunctions are constant functions. The complex conjugate of $f_\bk$
as a function on $X$ is $f_{-\bk} = c(\bk(-\bk))f_\bk^{-1}$.
\end{prop}

\begin{prop} \label{mu=measure}
$\mu$ is $\GG$-invariant probability measure on $X$ and
$(X,\GG,\mu)$ is a pure point ergodic dynamical system.
\end{prop}

\medskip

We now wish to define the ergodic spatial process $N$. To do this
we first define the diffraction-to-dynamics map.
For each $k \in \Bragg$, let $1_k : \Gdual \longrightarrow \CC$
be the function which is $1$ at $k$ and $0$ everywhere else.
Since the measure $\difm$ is only non-zero on the points of $\Bragg$,
it follows that every function in $L^2(\Gdual, \difm)$ can be represented
in the form
\[ \sum_{k \in \Bragg} x_k \,1_k, \; \mbox{where} \;
\sum_{k \in \Bragg} |x_k|^2 \difm(k) <\infty \, . \]

The diffraction-to-dynamics embedding is the mapping
\begin{eqnarray*}
\theta = \theta_{\tilde a}\,:\,L^2(\Gdual, \omega) &\longrightarrow& L^2(X,\mu)\\
\sum_{k \in \Bragg} x_k 1_k &\mapsto& \sum_{k \in \Bragg} x_k f_k \,.
\end{eqnarray*}

Comparing these two $L^2$-spaces we have
\[ \langle \sum x_k 1_k \, \mid \sum x_k 1_k\rangle =
\int_{\Gdual} \sum_{k,l} x_k \overline{x_l }\, 1_k \, 1_l d \,\difm=
 \sum_k |x_k|^2 \difm(k)
= \langle  \sum x_k f_k \mid \sum x_k f_k \rangle \, ,\]
which shows that $\theta$ is an isometric embedding.

Define an action $U$ of $\GG$ on $L^2(\Gdual, \difm)$ by setting, for each
$g\in L^2(\Gdual,\difm)$,  $U_t(g)$ to be the function
\begin{equation*}
U_t(g)(k) = k(t)g(k) \quad \mbox{for all} \, k\in \Gdual \, .
\end{equation*}
Under this action, $1_k \in L^2(\Gdual, \difm)$ is a $k$-eigenfunction for the action
and $\theta$ becomes $\GG$-equivariant.

Since $\difm$ is a positive, translation bounded, and backward transformable measure on $\Gdual$ there is a
unique positive definite measure $\gamma$ on $\GG$ for which $\difm
= \hat\gamma$. This measure is also translation bounded (and hence
automatically transformable (with Fourier transform equal to $\difm$), \cite{BF}). The basic relationship between $\gamma$ and $\difm$ is expressed by \eqref{omega-gamma}
for all $F \in C_c(\GG)$, see \cite{BF}, Ch.~1.4, which in particular gives $\hat F \in L^2(\Gdual, \difm)$. Direct calculation shows that for $F\in C_c(\GG)$,
$U_t \hat F = \widehat{T_{-t} F}$.

We define the process by
\[ N=N_a : C_c(\GG) \longrightarrow L^2(X, \mu), \quad N(F) = \theta(\hat F) \, .\]

Also $N$ is {\bf real}, linear, continuous, real, and stationary ($\GG$-equivariant). It is
also full.

\begin{prop} \label{stochasticProcess}
$\cN = (N,X,\mu,T)$ is a pure point ergodic spatial process with diffraction $\difm$
and associated phase form $a^*$.
\end{prop}

Thus, based on the positive measure $\difm$ on $\Gdual$
and a phase form $a^*$ we have constructed a pure point spacial process on $\GG$.

\section{Various concluding remarks}\label{conclusion}

In this section we collect together various small items that are of interest.

\subsection{The case that $N$ maps into $L^\infty$}\label{Special}
A spatial stationary process $\cN= (N,X,\mu,T)$ is said to have an $m$th moment if
$N(F_1)\ldots N (F_m)$ belongs to $L^1 (X,\mu)$ for any $F_1,\ldots, F_m\in C_c (\GG)$. In this case, the $m$th moment is defined to be the unique  linear map $\mu^{(m)}$ on  $C_c (\GG) \otimes \cdots \otimes C_c (\GG)$  ($m$ factors) with
$$\mu^{(m)} (F_1\otimes \cdots   \otimes F_m) := \int N(F_1)\cdots N(F_m) d\mu. $$

It it not hard to see that the measure $\mu$  determined by its moments $\mu^{(m)}$, $m\in \NN$, if all these moments exist and the process if full.  In fact, this is more or less the definition of fullness. In order to have a meaningful diffraction
 theory our running assumption is that the second moment exists and is a  measure. Then, also the first moment is a measure (as can easily be seen).

We say that  $\cN$ is {\it bounded} if $N(F)\in L^\infty (X,\mu)$ for every $F\in C_c (\GG)$. In this case, for every  $m\in \NN$ the $m$th moment exists.

We are now going to discuss how for bounded processes  this  notion of moment is related to the notion of moment of a phase form introduced above.

Thus, let $(N, X,\mu,T)$ be an bounded,  ergodic  full stationary process  with pure point spectrum and elementary phase form $a$. Then, the $m$th moment of $\mu$ and the $m$-moment of $a$ uniquely determine each other. More precisely, the following holds in this situation.

\begin{lemma} (a) Let $n\in \NN$  and  $F_1,\ldots, F_n \in C_c (\GG)$ be given. Then,
$$\int N ( F_1) \dots N (F_n) d\mu = \sum_{k_1\in
    \Bragg}\ldots\sum_{k_n\in \Bragg} \widehat{F_1} (k_1)\dots \widehat{F_n}
  (k_n) \int_X (f_{k_1}\dots f_{k_n}) d\mu$$
where the sums exist if taken  one after the other, and
for any $k_1, \dots, k_n \in \Bragg$,
$$
\int_X f_{k_1} \dots f_{k_m} d\mu
= \begin{cases} a(k_1, \dots, k_m) \difm(k_1)^{1/2} \dots \difm(k_m)^{1/2} &\mbox{if} \; k_1 + \cdots +k_m = 0;\\
0 &\mbox{if} \; k_1 + \cdots +k_m \ne 0 \,.
\end{cases}
$$

(b) For
  each $j=1,\ldots,n$, let $\{F_j^{(m)}\}$ be a sequence  in $C_c (\GG)$ whose
  Fourier transforms converge to  $\charF_{k_j}$ in $L^2 (\widehat{\GG},\difm)$.
Then,  $$f_{k_1}\dots f_{k_n} = \lim_{m_1\to \infty}  \lim_{m_2\to \infty}
  \ldots \lim_{m_n\to \infty} N(F_1^{(m_1)}) \dots N(F_n^{(m_n)}),$$
where the limits are taken in $L^2$ and one after the other.
In particular,  $$\int_X f_{k_1}\dots f_{k_n} d\mu  = \lim_{m_1\to \infty}
  \lim_{m_2\to \infty} \ldots \lim_{m_n\to \infty} \int_X  N (F_1^{(m_1)}) \dots N(F_n^{(m_n)}) d\mu  \,.$$
\end{lemma}
\begin{proof} This is essentially proven in Lemma 2.1 and Lemma 2.3 of \cite{LM}  for  special uniformly discrete point processes. The proof given there carries  over to our situation.
\end{proof}

\begin{coro}
Let two bounded spatial processes with pure point diffraction and the same diffraction measure $\omega$ be given. Then, these processes  have the same $m$-moment  if and only if their associated phase forms have the same $m$ moment. In particular, if furthermore  $\langle \cZ_n\rangle  = \cZ$, then  the two processes are isomorphic if and only if their $m$ moments agree for $m= 1,\ldots, n$.
\end{coro}

If every element of $\ev$ can be expressed as a sum of $m$ or less elements
of $\Bragg$, then $\langle \cZ_{2m+1}\rangle = \cZ$, and so $2m+1$ moments
suffice to determine the isomorphism class of a process, see Remark~\ref{2m+1 theorem}.

\subsection{A topological situation.}
In this section we look at  a pure point stationary process
$N: C_c(\GG) \longrightarrow L^2(X,\mu)$ and assume that $X$ is a topological space and all eigenfunctions are continuous. Here, continuity means that
for each eigenvektor $k$, we have an continuous  nontrivial function  $g_k$ with
$$g_k (T_t x) = (k,t)g_k (x)$$
for all $x\in X$ and $t\in \GG$. Without loss of generality we can then normalize the $g_k$ to satisfy
$$ g_k (x_0) = 1$$
for some fixed $x_0 \in X$. These normalized $g_k$ will then obey
\begin{equation}\label{character}
 g_{k_1} g_{k_2} = g_{k_1 + k_2}\;\:\mbox{and}\:\: g_{-k_1} = \overline{g_{k_1}}.
 \end{equation}
Of course, as discussed in Section \ref{StochasticprocessestoDiffraction}, the pure point process comes with a diffraction measure
$\difm$ a set of eigenvalues $\ev$ and a phase form $a$. Set $\TT:= \widehat{\ev_d}$. Then the considerations of Section \ref{StochasticprocessestoDiffraction} provide a torus system $(N_a, \TT, l_{\TT}, T)$. This system can be directly calculated as follows:

\begin{theorem} Assume the situation outlined above. Then,
$$\pi : X\longrightarrow \TT, x\mapsto (k \mapsto g_k (x)),$$
is a surjective $\GG$-map (and hence $\TT$ is a factor of $X$) and
$$ M : L^2 (\TT)\longrightarrow L^2 (X), g\mapsto g\circ \pi,$$
is an isomorphism of point processes.

\end{theorem}
\begin{proof}
Due to the normalization and  \eqref{character}, $\pi (x)$ belongs indeed to $\TT$. Obviously, $\pi$ is continuous (as $\ev$ is given the discrete topology).
Moreover, $M$ can easily be seen to map the character $k\in \ev$ of $\TT$ to the eigenfunction $g_k$. Thus, $M$ is unitary (as it maps an orthonormal basis to an orthonormal basis).
As $M$ is an isometry, $\pi$ must have dense image. As it it a $\GG$ map it is then onto.    Moreover, a short calculation
invoking \eqref{character} again shows that
$$M (g_{k_1} g_{k_2}) = g_{k_1} g_{k_2}.$$
By usual limiting arguments, this gives that $M (f g) = M (f) M (g)$ for all  $f,g\in L^\infty$.

As $M$ maps the character $k\in \ev$ to the eigenfunction $g_k$, its inverse $M^{-1}$  must map the eigenfunction $g_k$ to the character $k$. Combined with the usual limiting arguments this shows that $M^{-1}$ also satisfies $M^{-1} (u v ) = M^{-1} (u) M^{-1} (v)$ for bounded $u,v$.

So, if we now define the $N$-function $N_{\TT}$ on the torus by
$$N_{\TT} := M^{-1} N$$
we have indeed an isomorphism of processes. Thus, the phase form belonging to $N_{\TT}$ must then be given by $a$ as well. By uniqueness then $[N_a] = [N_{\TT}]$.
\end{proof}

Another issue around continuity is the possible continuity of
the mappings
$a:\Bragg \longrightarrow U(1)$ in the topology on $\Bragg$
induced from $\Gdual$. For these one uses this same topology to make
$\cF(\Bragg)$ into a topological group. We do not know if continuity here is
physically relevant but we do note that
in this case the classification results for elementary phase forms
compared to actual phase forms change
from $\TT$ to $\GG$ (see the discussion around and including Remark~\ref{continuousCase}).

\subsection{Spatial processes arising from measures on $\GG$}\label{ssp-measures}
Our underlying understanding of a spatial process $\cN=(N,X,\mu,T)$ is that it
represents some sort of `density' on the space
$\GG$. Each point $\xi \in X$ represents an instance of the yet unspecified
structure $\rho_\xi$ which can be paired with elements of $C_c (\GG)$ to give

\begin{equation}
\label{measureInterpretation2}
 \langle \rho_\xi,  F\rangle = N(F)(\xi)\,.
 \end{equation}
 Thus $C_c(\GG)$ provides a set of test functions for some sort of distribution $\rho_\xi$
 at $\xi \in X$. Of course $N(F)$ is not necessarily defined at all points of $X$
 (it is only an $L^2$-function on $X$), and even if it is defined it is not clear whether
 or not $\rho_\xi$ could be interpreted as any particular kind of distribution.
 We do note however, that from $\GG$-invariance we have
 $N(F)(T_t \rho_\xi) = (T_{-t} N(F))(\xi) = N(T_{-t} F)(\xi)$, so
 \[\langle T_t \rho_\xi,  F\rangle = \langle  \rho_\xi,  T_{-t}F\rangle\,\]
 for all $t\in \GG$.

An important question, and one for which we cannot give a satisfactory answer,
is when can this `density' be interpreted as a measure $\rho_\xi$ on $X = \GG$?
In the theory of Delone point sets $\Lambda$ this actually happens; typically
$X$ is the orbit closure of $\Lambda$ under the local topology and $N$ is
the process that arises from the point measures $\delta_\xi :=
\sum_{x\in \xi} \delta_x$ for all $\xi \in X$. Thus for $F \in C_c(\GG)$,
$N(F)(\xi) := \sum_{x\in \Lambda} F(x) = \langle \delta_\xi, F\rangle$.

\medskip
We have already seen one way in which we can see from general principles that measures arise
in \S\ref{compactG}. But this
is unnecessarily  restrictive since in many basic situations $\difm$ is not summable. For example,
in \S~\ref{G=U(1)} we looked at
the diffraction of the point set $\delta_\ZZ$, when it is
modeled on $U(1)$ as $\delta_0$. Its diffraction is $\difm = \delta_\ZZ$ on $\ZZ$.
This is not $L^1$ with respect to the Haar measure of $\ZZ$ (counting measure).
Here $X$ can be identified with $U(1)$ (since it is compact) and for all $t \in U(1)$,
$F\mapsto N(F)(t)$ is just the measure $\delta_t$. So in this case there is a suitable
measure in spite of the lack of summability.

The exact conditions that $\rho_\xi$ defined by \eqref{measureInterpretation2} should be a
measure at $\xi$ are that $N(F)$ actually be defined at $\xi$ for all $F\in C_c(\GG)$ and
that for each compact subset $K$ of $\GG$ there is a constant $c_K$ so that for all $F\in C_c(\GG)$ with support inside $K$
we have $N(F)(\xi) \le c_K ||F||_\infty$. In fact this is precisely
the definition of measure via the Riesz representation theorem.

If for some $\xi\in X$, $N(F)(\xi)$ is defined, {\em positive}, and finite for each positive
$F\in C_c(\GG)$, then the right-hand side of \eqref{measureInterpretation2} provides a positive linear functional on $C_c(\GG)$ and so defines a positive measure.

\subsection{Some open questions}

There are numerous questions that remain to be investigated. Amongst these
we can point out:

\begin{itemize}
\item Are there simple process theoretical conditions on pure point measures $\difm$
that allow one to conclude that the processes atttached to $\difm$ can be
interpreted as measure valued processes?

\item Our homometry theorem shows that the stationary processes with fixed pure point diffraction form a group. Does this have any deeper meaning? Can we directly realize multiplication / inversion of stationary processes?

\item When there are moments of all orders, these moments characterize the corresponding
processes. There are various scenarios under which not all these moments are necessary to solve the inverse problem. For instance in \cite{XR} it is seen that real model
sets with real internal spaces are recoverable from their second and third moments. In \cite{LM}
it is seen that if $\Bragg + \cdots +\Bragg = \ev$ ($m$ summands) then one needs only
the first $2m+1$ moments to solve the inverse problem. This suggests that solutions to the
inverse problem can be organized into hierarchies depending on the number of moments
required to specify them.

\item The sequence of moments suggests that there is some sort of cohomology theory
that would allow one to organize the increasing information that is added as successive
moments are included in the picture. Such a cohomology theory would be a most welcome
addition to this study.
\end{itemize}

\end{document}